\documentclass[11pt]{article}
 \usepackage{xcolor} 
\usepackage[margin=1in]{geometry}
\usepackage{amsmath, amssymb, amsthm}
\usepackage{hyperref}
\usepackage{dsfont}
\usepackage[utf8]{inputenc}
\usepackage{natbib}
\usepackage{graphicx}
\usepackage{caption}
\usepackage{comment}
\usepackage{subcaption}
\usepackage{soul}
\usepackage[vlined,ruled,resetcount,linesnumbered]{algorithm2e}
\allowdisplaybreaks
 \usepackage{wrapfig}
 \usepackage{tikz}
\usetikzlibrary{positioning}
\usepackage{booktabs}
\newcommand*{\QEDA}{\hfill\ensuremath{\square}}% 

\usepackage[vlined,ruled,resetcount,linesnumbered]{algorithm2e}
\allowdisplaybreaks

\usepackage{color}
\usepackage{float}
\usepackage{placeins}
\newcommand{\m}{m}
\newcommand{\M}{M}
\newcommand{\e}{e}
\newcommand{\E}{E}
\newcommand{\qe}{q_e}
\renewcommand{\r}{r}
\newcommand{\R}{R}
\renewcommand{\b}{b}
\newcommand{\B}{B}
\newcommand{\deleq}{\stackrel{\Delta}{=}}
\newcommand{\xrzbopt}{x_r^{z*}(b)}

\renewcommand{\(}{\left(}
\renewcommand{\)}{\right)}

\newcommand{\xbaropt}{\xbar^{*}}
\newcommand{\thetafun}{\theta_r}

\newcommand{\zp}{z^1}
\newcommand{\zpp}{z^2}
\newcommand{\distre}{d_{r,e}}

\newcommand{\xrbbar}{\xbar_r(\bbar)}
\newcommand{\changetrip}{\Delta\alpha}
\newcommand{\changevm}{\Delta\beta}

\newcommand{\start}{z}
\newcommand{\delaym}{\ell_m}

\newcommand{\arrm}{\theta_m}

\newcommand{\ri}{r^i}
\newcommand{\Ri}{R^i}
\newcommand{\tripmi}{\alpha_m^i}
\newcommand{\xopti}{\x^{i*}}
\newcommand{\lambopti}{\lambda^{i*}}

\newcommand{\probi}{P^i}
\newcommand{\Ep}{E^1}
\newcommand{\Epp}{E^2}
\newcommand{\up}{u^1}
\newcommand{\tollp}{\toll^1}
\newcommand{\tollpp}{\toll^2}
\newcommand{\lambp}{\lambda^1}
\newcommand{\lambpp}{\lambda^2}
\newcommand{\Roptp}{R^{1*}}
\newcommand{\Roptpp}{R^{2*}}
\newcommand{\Ei}{E^i}
\newcommand{\Ropti}{R^{i*}}
\newcommand{\kopti}{k^{i*}}
\newcommand{\upp}{u^2}

\newcommand{\xhatopt}{\xhat^{*}}
\newcommand{\vrmbbar}{v_{m,r}(\bbar)}
\newcommand{\Ball}{\bar{\B}}
\newcommand{\ball}{\bar{\b}}

\newcommand{\Vbar}{V}
\newcommand{\Vrbbar}{\Vbar_r(\bbar)}
\newcommand{\rep}{h_r}
\newcommand{\bbar}{b}

\newcommand{\etamrz}{\eta_{m,r}}

\renewcommand{\l}{l}
\renewcommand{\L}{L}
\newcommand{\xbarr}{\bar{x}}
\newcommand{\yopt}{\xbarr^*}
\newcommand{\Ropt}{\R^{*}}

\newcommand{\barV}{\bar{V}}
\newcommand{\barVrz}{\bar{V}_r}
\newcommand{\barb}{\bar{b}}
\newcommand{\Vrb}{V_\r(\b)}

\newcommand{\x}{x}
\newcommand{\xp}{x^{1}}
\newcommand{\xpp}{x^2}

\newcommand{\trpp}{d_{\rpp}}
\newcommand{\trp}{d_{\rp}}
\newcommand{\rptwo}{r^{1'}}
\newcommand{\foptp}{f^{t,1*}}
\newcommand{\foptpp}{f^{t,2*}}
\newcommand{\that}{\hat{t}}
\newcommand{\fhatp}{\hat{f}^1}
\newcommand{\fhatpp}{\hat{f}^{2}}

\newcommand{\xopthatp}{\xhat^{1*}}
\newcommand{\xopthatpp}{\xhat^{2*}}
\newcommand{\G}{G}

\newcommand{\lambopt}{\lambda^*}

\newcommand{\tollopti}{\toll^{i*}}

\newcommand{\xxi}{x^i}
\newcommand{\hlbar}{\bar{h}_{l}}
\newcommand{\etalowl}{\lambda_l}
\newcommand{\phil}{\phi_{l}}

\newcommand{\ui}{u^i}
\newcommand{\tolli}{\toll^i}
\newcommand{\uopti}{u^{i*}}
\newcommand{\lambi}{\lambda^i}

\newcommand{\uoptp}{u^{1*}}
\newcommand{\uoptpp}{u^{2*}}
\newcommand{\Gi}{G^i}
\newcommand{\tripmp}{\alpha_m^1}
\newcommand{\tripmpp}{\alpha_m^2}
\newcommand{\lamboptp}{\lambda^{1*}}
\newcommand{\lamboptpp}{\lambda^{2*}}
\newcommand{\tolloptp}{\toll^{1*}}
\newcommand{\tolloptpp}{\toll^{2*}}

\newcommand{\xbar}{x}

\newcommand{\mhat}{\hat{m}}
\newcommand{\bhat}{\hat{b}}
\newcommand{\rhat}{\hat{r}}
\newcommand{\xrb}{\x_r(\b)}

\newcommand{\xrbopt}{\x_r^{*}(\b)}
\newcommand{\xoptmm}{\x^{*}_{-m}}
\newcommand{\pmopt}{p^{*}_{m}}
\newcommand{\p}{p}

\newcommand{\popt}{p^*}

\newcommand{\tr}{d_r}

\renewcommand{\pm}{p_m}
\newcommand{\tripm}{\alpha_m}
\newcommand{\vm}{\beta_m}

\newcommand{\toll}{\tau}
\newcommand{\koptrz}{k^{z*}_r}
\newcommand{\tollet}{\toll_e^t}
\newcommand{\um}{u_m}

\newcommand{\Bbar}{B}
\newcommand{\Sx}{S(\x)}
\newcommand{\X}{X}
\newcommand{\pmdag}{p^{\dagger}_m}

\newcommand{\tildeb}{\rep(\ball)}
\newcommand{\xhat}{\hat{x}}
\newcommand{\xhatp}{\xhat^{1}}
\newcommand{\xhatpp}{\xhat^{2}}

\newcommand{\jhat}{\hat{j}}
\newcommand{\bj}{b_j}

\newcommand{\tmin}{d_{min}}
\newcommand{\N}{N}
\newcommand{\tildeE}{\tilde{E}}
\newcommand{\rmin}{r_{min}}
\newcommand{\rphat}{\hat{r}^1}
\newcommand{\tildeR}{\tilde{R}}
\newcommand{\rpphat}{\hat{\r}^2}
\newcommand{\bjhat}{b_{\jhat}}

\newcommand{\lambdaopt}{\lambda^*}
\newcommand{\umaxm}{\u^{\dagger}_m}
\newcommand{\xopt}{\x^{*}}
\newcommand{\tildeq}{\tilde{q}}

\newcommand{\rp}{r^1}
\newcommand{\rpp}{r^2}
\newcommand{\Gp}{G^1}
\newcommand{\Gpp}{G^2}
\newcommand{\xoptp}{\x^{1*}}
\newcommand{\xoptpp}{x^{2*}}
\newcommand{\Xp}{X^1}
\newcommand{\Xpp}{X^2}
\newcommand{\Rpp}{R^2}
\newcommand{\Rp}{R^1}

\renewcommand{\u}{u}

\newcommand{\Uopt}{U^*}

\newcommand{\umax}{u^{\dagger}}

\newcommand{\Bhat}{\hat{B}}
\newcommand{\pdag}{p^{\dagger}}
\newcommand{\udag}{u^{\dagger}}
\newcommand{\tolldag}{\toll^{\dagger}}
\newcommand{\bl}{\ball_l}

\newcommand{\kopt}{k^{*}}

\newcommand{\fhat}{\hat{f}}
\renewcommand{\xopt}{x^{*}}
\newcommand{\uopt}{u^{*}}
\newcommand{\umopt}{u^{*}_m}
\newcommand{\tollopt}{\toll^{*}}

\usepackage{setspace}

\newtheorem{theorem}{Theorem}
\newtheorem{lemma}{Lemma}
\newtheorem{proposition}{Proposition}

\newtheorem{claim}{Claim}

\theoremstyle{definition}

\newtheorem{example}{Example}

\newtheorem{definition}{Definition}

\DeclareMathOperator*{\argmin}{arg\,min}

\DeclareMathOperator*{\argmax}{arg\,max}

\usepackage{xcolor}
\usepackage{array}
 \usepackage{appendix}

\begin{document}
\title{Market Design for Capacity Sharing in Networks}

\author{Saurabh Amin\\
\footnotesize{Operations Research Center, MIT, amins@mit.edu} \and 
Patrick Jaillet\\
\footnotesize{Operations Research Center, MIT, jaillet@mit.edu}\and
Haripriya Pulyassary\\
\footnotesize{Operations Research and Information Engineering, Cornell University, hp297@cornell.edu} \and
Manxi Wu\\
\footnotesize{Civil and Environmental Engineering, UC Berkeley, manxiwu@berkeley.edu}
} 

\date{}
\maketitle

\begin{center}
    (Authors are alphabetically ordered)
\end{center}

\begin{abstract}
We study a market mechanism that sets edge prices to incentivize strategic agents to efficiently share limited network capacity. In this market, agents form coalitions, with each coalition sharing a unit capacity of a selected route and making payments to cover edge prices. Our focus is on the existence and computation of market equilibrium, where challenges arise from the interdependence between coalition formation among strategic agents with heterogeneous preferences and route selection that induces a network flow under integral capacity constraints. To address this interplay between coalition formation and network capacity utilization, we introduce a novel approach based on combinatorial auction theory and network flow theory. We establish sufficient conditions on the network topology and agents' preferences that guarantee both the existence and polynomial-time computation of a market equilibrium. Additionally, we identify a particular market equilibrium that maximizes utilities for all agents and the outcome is equivalent to the classical Vickrey-Clarke-Groves mechanism. Furthermore, we extend our results to multi-period settings and general networks, showing that when the sufficient conditions are not met, an equilibrium may still exist but requires more complex, path-based pricing mechanisms that set differentiated prices based on agents' preference parameters.
\end{abstract}

\section{Introduction}\label{sec:introduction}
Transportation and logistics networks often face losses from congestion caused by inefficient use of limited capacity. One effective approach to improve resource utilization and reduce costs is to incentivize users to form resource-sharing coalitions. This can be done through a market-based approach that sets appropriate prices for capacity usage at specific times and locations. For example, setting edge toll prices for using road networks can incentivize agents to organize carpooling trips, reduce the number of vehicles on the road, and minimize congestion. Similarly, flexible and on-demand shipping applications can optimize deliveries and routes, reducing the number of cargoes and the fleet size required. By leveraging the complementarity between resource sharing and pricing, we can achieve substantial improvement in cost, time, and reducing environmental impacts for transportation and logistics networks.

The goal of this paper is to analyze a market mechanism that incentivizes strategic agents to share network capacities by forming coalitions. We begin with a basic static model that considers a network composed of resources modeled as directed edges with finite capacity. Agents form coalitions, and each coalition selects a route (a sequence of edges) that connects a source to a sink. We refer to the combination of an agent coalition and a route as a trip. Each trip utilizes one unit of capacity on all edges along the route, with the edge capacity limiting the number of trips that can use that edge. We consider trip value functions that are affine in the time cost with private and heterogeneous preference parameters that represent agents' sensitivity to travel time and the disutility of sharing capacity with others in the coalition.

The basic model considers a simple edge-based homogeneous pricing scheme: Each coalition pays a price for using a unit of capacity of an edge. The price is the same for all coalitions, and acts as the ``invisible hand" that governs agents' coalition formation and edge capacity utilization. When the price of an edge increases, agents are incentivized to either form coalitions to share trips and split the cost or switch to using the capacity of a less expensive route. Agents in each coalition make payments to cover the total cost of the chosen path. 
%We do not impose restrictions on how agents share the total price. 
A competitive market equilibrium is defined by the trip allocation, edge prices, and agent payments that satisfy the following conditions: \emph{(i)} individual rationality—ensuring all agents have non-negative equilibrium utility; \emph{(ii)} stability—no coalition of agents has an incentive to deviate from the equilibrium trip; \emph{(iii)} budget balance—agents' payments cover the edge prices; and \emph{(iv)} market clearing—prices are charged only on edges where the edge capacity is saturated. 

{A central feature of our model is the interplay between coalition formation and capacity utilization in terms of the induced flow in the network. Each coalition corresponds to a group of agents jointly selecting a route and sharing one unit of capacity along that route. These trip-level decisions naturally induce a flow over the network. However, unlike classical flow models where the optimal network flow is determined independently of agent groupings, here the coalition structure and the induced flow are tightly coupled: the value of a trip depends on both the route chosen and the identity and preferences of the agents in the coalition. Conversely, edge prices affect not only the agents' route choices but also their incentives to form coalitions, since sharing costs alters their utility. This interdependence poses a major challenge: capacity utilization in terms of the network flow cannot be determined without simultaneously considering how coalitions form in equilibrium, and vice versa. }

Our focus is on the existence and computation of market equilibrium. We identify sufficient conditions under which a market equilibrium with homogeneous edge-based pricing exists and can be computed in polynomial time. Specifically, these conditions include a \emph{series-parallel} network topology and \emph{homogeneous disutilities of capacity sharing} among agents. To illustrate the "tightness" of these conditions, we provide two counterexamples (Examples \ref{ex:wheatstone} and \ref{ex:carpool}), each violating one of the conditions and resulting in the failure of equilibrium existence. 
When these sufficient conditions do not hold, we explore extensions to the pricing schemes that restore equilibrium existence. In particular, we show that equilibrium can still exist in general networks when path-based pricing is applied. Furthermore, for agents with heterogeneous capacity-sharing disutilities, equilibrium can be achieved by segmenting agents into separate markets based on their disutilities and differentiating pricing accordingly. Table \ref{table:market_equilibrium} provides a summary of these findings.

\begin{table}[ht]
\centering
\begin{tabular}{|>{\centering\arraybackslash}m{3.7cm}|>{\centering\arraybackslash}m{3.7cm}|>{\centering\arraybackslash}m{3.9cm}|}
\hline
& \textbf{Series-parallel network} & \textbf{Non-series-parallel network} \\
\hline
\textbf{Homogeneous capacity sharing disutility} & Edge-based pricing; homogeneous across agents (Theorem \ref{theorem:sp}) & Path-based pricing; homogeneous across agents (Proposition \ref{prop:extension}) \\
\hline
\textbf{Heterogeneous capacity sharing disutility} & Edge-based pricing; population-specific  (Proposition \ref{prop:extension}) &  Path-based pricing; population-specific (Proposition \ref{prop:extension})\\
\hline
\end{tabular}
\caption{Conditions for market equilibrium existence}
\label{table:market_equilibrium}
\end{table}

% {\color{blue}Our key contribution is to identify structural conditions—namely, series-parallel network topology and homogeneous capacity-sharing disutilities—under which this coupling can be decoupled. Under these conditions, we show that the optimal flow can be computed independently of agent coalitions using a greedy algorithm. Coalition formation can then be resolved via an assignment problem that admits a Walrasian equilibrium. This decoupling is what allows us to establish both the existence and polynomial-time computability of market equilibrium in these structured settings.
% The structural assumptions we adopt—series-parallel network topology and homogeneous capacity-sharing disutility—are central to revealing the core algorithmic and economic tensions in capacity-sharing markets. These assumptions are not merely simplifying, but serve to isolate fundamental sources of complexity: in more general networks, fractional optimal flows and strategic interactions across overlapping paths create obstructions to both equilibrium existence and efficient computation. Likewise, heterogeneity in sharing preferences breaks the substitutability conditions needed for coalition stability and tractable mechanism design. By studying a structured setting where these difficulties can be resolved, our results establish a foundation that clarifies why equilibrium fails more generally and guides the design of more flexible pricing and coordination schemes—topics we pursue in later sections.}

{We next present our approach for addressing equilibrium existence and computation, highlighting the distinct role each condition plays in ensuring existence and tractability. Our starting point is to formulate the optimal trip allocation problem as an integer program, and show that the existence of a market equilibrium is equivalent to the zero integrality gap in its associated linear relaxation (Proposition \ref{prop:primal_dual}). This connection between a zero integrality gap and market equilibrium has been established in various market design settings involving both divisible and indivisible goods. Although standard, this initial step allows us to transform the equilibrium existence problem into the analysis of the integer program and its linear relaxation for optimal trip allocation. However, since the set of all possible coalitions is {exponentially large}, the linear relaxation of the optimal trip allocation problem has exponential size, which is computationally challenging even without integer constraints. 

The series-parallel network condition plays a critical role in decoupling the two intertwined layers of the optimal trip allocation problem: coalition formation and network flow allocation. We show that, in series-parallel networks, there always exists an integer route flow that is optimal for all agent preferences and coalition formations (Lemma \ref{lemma:FF}), and this flow can be computed using a polynomial-time greedy algorithm (Algorithm \ref{alg:flow}). This key structural property allows us to compute the optimal capacity utilization represented as a network flow—independently of coalition formation, which is computed separately in the next step. In contrast, in non-series-parallel networks, no single route flow is optimal for all feasible coalition formations. The optimal flow depends on the coalition structure and may be fractional, leading to a non-zero integrality gap in the linear relaxation of the trip allocation problem and failure of equilibrium existence, as illustrated in Example \ref{ex:wheatstone}.

The condition of homogeneous capacity-sharing disutility plays a critical role in the tractability of coalition formation. We show that coalition formation and allocation to the constructed route flow is mathematically equivalent to a Walrasian equilibrium in an auxiliary economy, where agents act as “indivisible goods” and route capacity units as “buyers” with an augmented trip value function. The homogeneity assumption ensures that this augmented value function satisfies monotonicity and the gross substitutes property, even though the original value function does not. As a result, a Walrasian equilibrium exists in the auxiliary economy (following \cite{gul1999walrasian}), which translates directly to a market equilibrium in our setting. In contrast, Example \ref{ex:carpool} shows that when the homogeneous capacity-sharing disutility condition is violated, even if the series-parallel network condition holds, an equilibrium may fail to exist.}

Furthermore, building on the connection between market equilibrium and the Walrasian equilibrium of the auxiliary economy, we identify a particular market equilibrium in which trip organization and payments align with the Vickrey-Clarke-Groves (VCG) mechanism. This structure enables a central planner to implement the equilibrium outcome without knowing agents’ private preferences. Notably, this equilibrium achieves the highest agent utilities among all possible market equilibria and collects only the minimum total edge prices (Theorem \ref{theorem:strategyproof}). This result leverages the lattice structure of Walrasian equilibria in the constructed auxiliary economy, as established in \cite{gul1999walrasian}.

The proof of Theorem \ref{theorem:sp} also enables a two-step polynomial-time algorithm for computing the market equilibrium, where we first compute the equilibrium route flow using a greedy algorithm (Algorithm \ref{alg:flow}), and then compute the equilibrium coalitions given the route flow capacity as the Walrasian equilibrium in the equivalent economy with augmented trip functions (Algorithm \ref{alg:allocation} in 
{Appendix \ref{ap:alg}}). {Several different approaches have been developed for computing Walrasian equilibrium in markets with indivisible goods and value functions that satisfy gross substitutes conditions. These include the original Kelso-Crawford algorithms \cite{kelso1982job}, sub-gradient based algorithms (\cite{ausubel2002ascending, parkes1999bundle}), primal-dual algorithms (\cite{de2007ascending}), ellipsoid method-based algorithm (\cite{nisan2006communication}) and combinatorial algorithms based on discrete convex analysis (\cite{murota2003application, murota1996valuated, paes2020computing})}. We build on the well-known Kelso-Crawford algorithm (\cite{kelso1982job}), but modify it to ensure an efficient way to iteratively compute the augmented trip value functions.

On general networks with multiple origin-destination pairs and heterogeneous capacity-sharing disutilities, the two sufficient conditions are violated, and market equilibrium may fail to exist under edge-based pricing that is homogeneous across agents. We show that equilibrium can still be achieved by adopting route-based pricing and serving agent populations in separate markets based on their capacity-sharing disutilities. However, computing the optimal capacity allocation across these populations is NP-hard. To address this, we develop a branch-and-price algorithm (Algorithm \ref{alg:branchAndPrice} in Appendix \ref{apx:multi}) to compute market equilibrium in such general settings.%On general networks} with multiple origin-destination pairs and heterogeneous capacity sharing disutilities, the two sufficient conditions are violated, and market equilibrium may not exist with edge-based pricing that is homogeneous for all agents. We show that market equilibrium still exists if we allow prices to be route-based rather than edge-based, and agent populations are served in separate markets based on their capacity sharing disutilities. However, computing the optimal capacity allocation across different agent populations is NP hard. We provide a branch-and-price algorithm (Algorithm \ref{alg:branchAndPrice} in Appendix \ref{apx:multi}) to compute equilibrium in such general markets. 

Finally, in Section \ref{sec:multipop}, we extend our static model to a multiple-period model, where agents have heterogeneous preferred arrival times and late arrival penalties, and they form coalitions to choose both the route and departure time. We show that the same equilibrium existence and computation results from the static model extend to the multi-period model. In multi-period setting, the series-parallel network condition again plays a crucial role as we show that a temporally repeated flow computed by greedy algorithm is the optimal trip route flow on the series-parallel network regardless of agents' preferences and coalition formation (Lemma \ref{lemma:FF_dynamic}). Under the condition of homogeneous capacity sharing disutility condition, the coalition formation and allocation to routes and departure times given the computed temporally repeated flow follows from that in the static model. We test our algorithm in the numerical example of designing a carpool market using the data from the California Bay Area highway network (Appendix \ref{subsec:numerical}).

\medskip 
\noindent\textbf{Literature review.}
Cost sharing in network resource utilization has been studied in various contexts. In particular, the competitive market framework was adopted by  \cite{ostrovsky2019carpooling} to model network capacity sharing in the context of autonomous carpooling markets. The authors considered general trip value function in a static model, and defined the concept of market equilibrium. They demonstrated that a market equilibrium, when it exists, maximizes social welfare by applying the first welfare theorem. However, their article did not investigate the equilibrium existence and computation complexity. Our paper has a distinct focus on existence and polynomial time equilibrium computation in both static and multi-period settings. We also studied the extensions from edge-based homogeneous pricing to more complex pricing mechanisms that guarantee equilibrium existence in general network and preference parameter settings. 

One line of related work concerns cost sharing in cooperative traveling salesman games, where a group of agents shares a delivery tour and seeks a stable division of costs. The paper \cite{potters1992traveling} introduced the traveling salesman game and demonstrated that the core may be empty even with symmetric costs. Follow-up studies focused on characterizing stable cost sharing and bounding the budget-balance gap for both symmetric \cite{faigle1993some, faigle1998approximately, tamir1989core} and asymmetric traveling salesman games \cite{blaser2008approximately, toriello2013traveling}. Our model differs from cooperative traveling salesman games in  that the underlying optimization problem in our setting is a network flow problem with capacity sharing on edges, whereas the traveling salesman games literature centers on constructing traveling salesman tours. Additionally, while the traveling salesman game assumes fixed costs (e.g., the length of the optimal tour), our model features edge prices that are endogenously determined through market equilibrium. Like the traveling salesman games literature, our analysis leverages linear programming relaxations; however, our focus lies in identifying conditions under which market equilibrium exists (i.e., when the integrality gap of the linear programming relaxation is zero), rather than in approximating the core or analyzing the integrality gap.

Furthermore, a separate line of literature has examined cost sharing in network design, where agents construct edges to connect source-sink pairs and share the associated costs—either equally or according to some rule \cite{chekuri2006non, epstein2007strong, von2013optimal, anshelevich2008price, chen2010designing, albers2009value, hoefer2011competitive, harks2021efficient}. Cost sharing has also been explored in scheduling and congestion games, where the cost of each resource varies with aggregate agent utilization, and agents may split the costs among themselves. This involves the cooperative setting (\cite{fotakis2008atomic, kollias2011restoring, rozenfeld2006strong, christodoulou2004coordination, immorlica2009coordination, cole2011inner, gkatzelis2016optimal, holzman2015strong}) and non-cooperative routing with proportional cost splitting (\cite{harks2011worst, johari2004efficiency}). In cooperative settings, strong Nash equilibrium is proposed as a solution concept, describing a scenario where no coalition of agents can improve their outcomes (\cite{aumann1959acceptable, epstein2007strong, rozenfeld2006strong, holzman1997strong, andelman2009strong, ferguson2023collaborative}). Previous studies have shown that strong Nash equilibria may not always exist, and have also characterized the bounds of the price of anarchy and the price of stability for strong Nash equilibria when they exist, e.g. \cite{epstein2007strong, chen2010designing, von2013optimal}. The paper \cite{epstein2007strong} showed that strong Nash equilibrium exists in fair and general network connection games when the network is series-parallel. Additionally, another paper \cite{hao2024price}  showed that the bound of price of anarchy in routing games is also lower in a series-parallel network than that in a general network.

Our model and results differ from the existing lines of work in two key aspects. First, in our model, the network is fixed with edges having fixed and finite capacities, while the price of using a unit of capacity on each edge is endogenously determined by the market equilibrium outcome. This contrasts with the cost-sharing literature for network design, where the cost of building an edge is fixed and agents decide which edges to be constructed. Second, agents with heterogeneous preferences can form coalitions to share one unit of capacity, with the formation of such coalitions being influenced by the edge prices and the way these prices are divided among agents to align their incentives. Such coalition formation leads to reduced capacity usage, an aspect not explored in previous works on capacity sharing and network routing with cooperative agents. These two features necessitate developing the new approach in this work to study the market equilibrium that accounts for both the network flow and coalition formation. Our work establishes new methods and results on equilibrium existence and computation for market equilibrium in capacity-sharing networks.

\section{The static basic model}\label{sec:model}
% \subsection{Network, Agents, and Trips}
For notation brevity, we present our main results in the static basic model. All results extend to the multi-period model, which is detailed in Section \ref{subsec:extension_1}. We consider a network with a single source and a single sink connected by a set of resources $E$ modeled as directed edges. The capacity of each edge $\e \in \E$ is a positive integer $\qe \in \mathbb{N}_{+}$, and the time cost of traversing the edge is a positive real number $d_e \in \mathbb{R}_{>0}$. A route $\r \in \R$ is a sequence of edges that connect the source and the sink with time cost $d_r = \sum_{e \in r} d_e$. The set of all routes is $\R$. A finite set of agents $M =\{1, \dots, |M|\}$ form coalitions with coalition size at most $A$. We denote agent coalition as $\b \in \Bbar \deleq \left\{b \in 2^{M}\left\vert~ |\b|\leq A \right.\right\}$, where $\Bbar$ is the set of all feasible coalitions. Each coalition $\b \in \B$ chooses a single route $\r$ and consumes one unit of  capacity for each edge $e \in \r$. We define a tuple $(\b, \r)$ as a \emph{trip}, and the set of all feasible trips is $\B \times \R$. We allow an agent to leave the market if that agent does not join any coalition and does not choose any route. 

% to share a unit capacity of a selected route $r \in R$. We ,  where 
% We define a bundle as a tuple $ \left(\b, \r\right)$, where $\b \in $ is an agent coalition with the maximum size of $A$, and $\r \in \R$ is the route that the agent coalition takes. Each bundle $(\b, \r)$ consumes one unit on every edge $e \in \r$. 

%\endnote{Thus, in our setting, each edge has an L-shaped cost function: cost is a constant when the edge load is below the edge capacity, and becomes extremely high once the load exceeds capacity. In the context of traffic congestion: when the traffic load is below the road capacity, all vehicles pass through the segment at the free-flow speed. However, when the traffic load exceeds the capacity, the travel time significantly increases due to congestion. In our market mechanism, the edge prices are set to ensure that the load of each edge does not exceed its capacity. }

% A \emph{trip} is defined as a tuple $ \(\start, \b, \r\)$, where $\start \in \{1, 2, \dots, T\}$ is the departure time at the origin, $\b \in \Bbar \deleq \left\{2^{M}\left\vert~ |\b|\leq A \right.\right\}$ is a coalition of agents who share the trip with maximum coalition size of $A$, and $\r \in \R$ is the route that the trip takes. A trip is \emph{feasible} if the arrival time at the destination is before time $\T$. The set of all feasible trips is given by: 
% \begin{align}\label{eq:Trip}
% \B \times \R \deleq \left\{(\start, \b, \r) \left\vert\begin{array}{l}
% \start=1, 2, \dots T, ~ \r \in \R,\\
%  \start+\tr \leq T,~ \b \in \B
% \end{array}\right.\right\}.
% \end{align}

For each agent $m \in M$, the value of a trip $\left(\bbar, \r\right)$ such that $\b \ni \m$ is an affine function of the route time cost, and depends on the coalition size:   \begin{align}\label{eq:m_valuation_trip}
\vrmbbar &= \tripm - \vm \tr - \left(\changetrip_m(|b|)+ \changevm_m(|b|)\tr\right). %- (\dism(|\b|)+ \disumbbar \tr). %\quad \forall \bbar \in \{\B |\bbar \ni \m\}, \quad \forall \m \in \M, \quad \forall \r \in \R,
\end{align}%Agents' valuations of trips. We represent each agent's preference by the vector $\(\B \times \Rm, \vm, \dism, \disum, \arrm, \delaym\)$, 
where $\tripm, \vm\geq 0$ are the single-agent trip preference parameters, and
$\changetrip_m(|\b|), \changevm_m(|b|) \geq 0$ are agent $m$'s \emph{capacity sharing disutility parameters} for all $|b|=1, \dots, A$. We set $\changetrip_m(1)= \changevm_m(1)=0$ since coalition size $|b|=1$ indicates that agent $m$ does not share the capacity with any other agent. Moreover, $\changetrip_m(|\b|+1)- \changetrip_m(|\b|), \changevm_m(|\b|+1)- \changevm_m(|\b|) \geq 0$, and are non-decreasing in the coalition size $|\b|$ for all $|\b|=1, \dots, A-1$ and all $m\in M$. 
% \begin{itemize}
% \item[-] $\tripm$ is agent $\m$'s value of arriving at the destination.
% \item[-] $\vm$ is agent $\m$'s value of time. When taking route $r$, the disutility of spending time $\tr$ is $\vm\tr$. 
% \item[-] $\dism(|b|)+ \disum(|b|) \tr$ is agent $m$'s disutility of sharing a trip with size $|b|$, and the disutility is linear in the time cost $\tr$. %where $\dism(|\b|)$ is agent $\m$'s fixed disutility of carpool with coalition $\b$. The fixed disutility does not depend on the travel time, and it arises from the cost of detours made for picking-up and dropping-off other agents, and costs of walking to the pick-up and drop-off locations. 
% %\item[-] $\disum=\(\disum(|b|)\)_{|b|=1}^{A}$, where $\disumbbar$ is agent $\m$'s disutility of sharing the vehicle with coalition $\b$ for a unit of travel time. When taking route $r$, the disutility $\disumbbar\tr$ arises from the inconvinience of sharing the vehicle during the trip with time $\tr$. 
% \item[-] $\delaym((\start+\tr-\arrm)_{+})$ is agent $m$'s cost of delay, where $\arrm$ is agent $m$'s preferred latest arriving time, and $(\start+\tr-\arrm)_{+} = \max\{\start+\tr-\arrm, 0\}$ is the time of agent $m$ being late. The function $\delaym: \mathbb{R}_{\geq 0} \to \mathbb{R}_{\geq 0}$ can be any non-decreasing function, and $\delaym(0)=0$ for all $m \in M$.
% \end{itemize}
This indicates that sharing capacity with other agents in the coalition decreases agent's valuation of the trip, and the marginal decrease of valuation increases in the size of the coalition. The total value of each trip $\(\bbar, \r\)$ is the summation of the trip values for all agents in $\bbar$:  
\begin{align}\label{eq:value_of_trip}
    \Vrbbar &= \sum_{\m \in \bbar} \vrmbbar.% - c(\bbar, \r). %\notag\\
    %&= \sum_{\m \in \bbar} \(\tripm - \vm \tr\) - \sum_{\m \in \bbar} \(\dism+\disumbbar\tr\) - \sum_{\m \in \b} \delaym((\start+\tr-\arrm)_{+})-  \(\fixedcost+ \delta \tr\)|\b| . 
\end{align}%agents prefer to take solo trips rather than sharing with others, and the disutility increases with the travel time. Moreover, the marginal disutilities are non-decreasing in the coalition size $|\b|$. Therefore, disutilities of all agents are non-decreasing in the coalition size, and the extra disutility of adding one agent to any trip $\(\b, \r\)$ is non-decreasing in the coalition size $|\b|$.  

We consider a market outcome represented by a tuple $\left(\xbar, \p, \toll\right)$, where $\xbar$ is the trip allocation vector, $\p = \left(\pm\right)_{\m \in \M}$ is the payment vector, where $\pm$ is the payment charged from agent $m$, and $\toll=\left(\toll_e\right)_{\e \in \E} \in \mathbb{R}_{\geq 0}^{|\E|}$ is the edge price vector, where $\toll_e$ is the price for using a unit capacity of edge $\e$. Given the edge price vector $\toll$, the price for a trip $(\bbar, \r)$ equals $\sum_{\e \in \r} \toll_e$. The trip allocation vector is $\xbar=\left(\xrbbar\right)_{(\b, \r)\in \B \times \R} \in \{0,1\}^{|\B| \times |\R|}$, where $\xrbbar=1$ if one unit capacity on route $r$ is allocated to trip $\left(\bbar, \r\right)$ and $\xrbbar=0$ otherwise. A feasible allocation vector $\x$ must satisfy the following constraints:
\begin{subequations}\label{eq:feasible_x}
\begin{align}
   \sum_{\r \in \R}\sum_{\b \ni \m} \xrbbar &\leq 1, \quad \forall \m \in \M,  \label{subeq:at_most_one}\\
\sum_{\b \in \B} \sum_{\r \ni e}\xrbbar & \leq \qe, \quad \forall \e \in \E, \label{subeq:edge_capacity}\\
    \xrbbar &\in \{0, 1\}, \quad \forall (\b, \r) \in \B \times \R, \label{subeq:int}
\end{align}
\end{subequations}
where \eqref{subeq:at_most_one} ensures that no agent is allocated to more than 1 trip, and \eqref{subeq:edge_capacity} ensures that the total number of trips allocated to use edge $\e \in \E$ does not exceed the edge capacity $\qe$. Given any $\(\xbar, \p, \toll\)$, the utility of each agent $\m \in \M$ equals to the value of the trip that $\m$ takes minus the payment:
\begin{align}\label{eq:u_p}
    \um = \sum_{r \in R} \sum_{\b \ni \m} \vrmbbar \xrbbar - \pm, \quad \forall \m \in \M. 
\end{align}

We define the market equilibrium as an outcome $\left(\xopt, \popt, \tollopt\right)$ that satisfies four properties -- \emph{individual rationality}, \emph{stability}, \emph{budget balance}, and \emph{market clearing}. 
\begin{definition}\label{def:market_equilibrium}
A market outcome $\left(\xopt, \popt, \tollopt\)$ is an equilibrium if it satisfies
\begin{enumerate}
\item \emph{Individually rationality:} All agents' utilities $u^*$ as in \eqref{eq:u_p} are non-negative, i.e.
\begin{align}\label{eq:ir}
    \uopt_m \geq 0, \quad \forall \m \in \M.
\end{align}
\item \emph{Stability:} No coalition in $\B$ can gain higher total utility by organizing a different trip:
\begin{align}\label{eq:stability}
    \sum_{\m \in \bbar} \uopt_m \geq \Vrbbar- \sum_{\e \in \r} \toll_e^*, \quad \forall (\b, \r) \in \B \times \R. 
\end{align}
\item \emph{Budget balance:} The total payments of agents in a coalition equals to the sum of the edge prices of the trip. An agent's payment is zero if they are not in any coalition:
 \begin{subequations}\label{eq:trip_bb}
\begin{align}
    &x^*_r(\bbar)=1 \quad \Rightarrow \quad \sum_{\m \in \bbar} \popt_m = \sum_{\e \in \r} \toll_e^{*}, \quad \forall (\b, \r) \in \B \times \R, \label{subeq:bb}\\
    &x^*_r(\b)=0, \quad \forall (\b, \r) \in \{\B \times \R | \b \ni \m\} \quad \Rightarrow \quad \popt_m = 0, \quad \forall \m \in \M. \label{subeq:p_not_assigned}
\end{align}
\end{subequations}
\item \emph{Market clearing:} For any edge $e \in E$, the edge price $\toll_e^{*}$ is zero when the number of trips allocated to edge $e$ is below the edge capacity:
\begin{align}\label{eq:mc}
    \sum_{(\b, \r) \in \{\B \times \R | \r \ni \e\}} x_{r}^{*}(b)  < \qe \quad \Rightarrow \quad \toll_{e}^{*}=0, \quad \forall \e \in \E. 
\end{align}
\end{enumerate}
\end{definition}

% We next define four properties of the market outcomes, namely 

% Firstly, an outcome $\left(\xbar, \p, \toll\right)$ is  if 

% Secondly, an outcome $\(\xbar, \p, \toll\)$ is \emph{stable} if 

%  Thirdly, an outcome $\(\xbar, \p, \toll\)$ is \emph{budget balanced} if 
%  %\begin{align}\label{eq:platform_bb}
% %     \sum_{\m \in \M} \pm  \geq \sum_{\e \in \E} \tolle \(\sum_{\r \ni \e}\sum_{\bbar \in \Bbar}\xrbbar\) + \sum_{\r \in \R} \sum_{\bbar \in \Bbar} \crbbar \xrbbar.
% % \end{align}

% Fourthly, an outcome $\(\xbar, \p, \toll\)$ is \emph{market-clearing} if 

In Definition \ref{def:market_equilibrium}, individual rationality condition \eqref{eq:ir} prevents agents from opting out of the market. The stability condition in \eqref{eq:stability} ensures that the total utility $\sum_{m \in \bbar} u_m^{*}$ for any coalition $\bbar$ induced by the market equilibrium $(\xopt, \popt, \tollopt)$ is higher or equal to the maximum total utility that can be obtained by $\bbar$ taking any feasible trip $\left(\bbar, \r\right)$, which equals to the trip value $V_r(\bbar)$ minus the edge prices $\sum_{\e \in \r} \toll_e^{*}$. Thus, agents have no incentive to form another coalition not in the equilibrium. The individual rationality condition and the stability condition together guarantee that agents will follow the equilibrium trip allocation. The budget balance condition guarantees that the payments cover edge prices, while market clearing ensures existence of non-zero equilibrium edge prices only on fully utilized edges. The following integer program solves the socially optimal allocation vector that maximizes the total value $S(x)$ of all trip allocation: 
\begin{equation}\tag{$\mathrm{IP}$}\label{eq:IP}
\begin{split}
    \max_{\xbar} \quad &S(\xbar)=\sum_{(\b, \r) \in \B \times \R}\Vrbbar \xrbbar, \quad 
    s.t. \quad  \text{$\xbar$ satisfies \eqref{subeq:at_most_one} -- \eqref{subeq:int}.}
    \end{split}
\end{equation}%Additionally, the budget balance condition ensures that in equilibrium, payments are enough to cover the edge prices and trip costs of all organized trips. The market clearing condition guarantees that equilibrium edge prices are only non-zero on edges and time stages where the capacity is fully utilized. 

We introduce the linear relaxation of \eqref{eq:IP} as the primal linear program:  %\begin{taggedsubequations}{$\mathrm{LP}$}\label{eq:LP1bar}
\begin{subequations}
        \makeatletter
        \def\@currentlabel{$\mathrm{LP}$}
        \makeatother
        \label{eq:LP1bar}
        \renewcommand{\theequation}{$\mathrm{LP}$.\alph{equation}}
    \begin{align}
         \max_{\xbar} \quad &\Sx=\sum_{(\b, \r) \in \B \times \R}\Vrbbar \xrbbar, \notag\\
  s.t. \quad   &\sum_{\b \ni \m}\sum_{\r \in \R} \xrbbar \leq 1, \quad \forall \m \in \M,  \label{subeq:LP11}\\
&\sum_{\b \in B} \sum_{\r \in \R}\xrbbar  \leq \qe, \quad \forall \e \in \E, \label{subeq:LP12}\\
    &\xrbbar \geq 0, \quad \forall (\b, \r) \in \B \times \R. \label{subeq:LP13}
%   & \sum_{\r \in \R}\sum_{z=1}^{T} \sum_{\bbar \ni \m} \xrbbar \leq 1, \quad \forall \m \in \M,  \label{subeq:LP11}\\
% &\sum_{\r \ni \e} \sum_{\bbar \in \Bbar} x_r^{t-\distre}(b)  \leq \qe, \quad \forall \e \in \E, \quad \forall t=1, \dots, T, \label{subeq:LP12}\\
%  &\xrbbar \geq 0, \quad \forall \bbar \in \Bbar, \quad \forall \r \in \R.\label{subeq:LP13}
    \end{align}
\end{subequations}Note that the constraint $\xrbbar \leq 1$ is implicitly included in \eqref{subeq:LP11}, and thus is omitted. By introducing dual variables $\u=\(\um\)_{\m \in \M}$ for constraints \eqref{subeq:LP11} and $\toll = \(\toll_e\)_{\e \in \E}$ for constraints \eqref{subeq:LP12}, the dual of \eqref{eq:LP1bar} can be written as follows: 
% \begin{taggedsubequations}{$\mathrm{D}$}\label{eq:D1bar}
\begin{subequations}
        \makeatletter
        \def\@currentlabel{$\mathrm{D}$}
        \makeatother
        \label{eq:D1bar}
        \renewcommand{\theequation}{$\mathrm{D}$.\alph{equation}}
\begin{align}
    \min_{\u, \toll} \quad &U(\u, \toll)= \sum_{\m \in \M} \um+\sum_{\e \in \E} \qe \toll_e, \notag \\
s.t. \quad & \sum_{\m \in \bbar} \um +\sum_{\e \in \r} \toll_e \geq \Vrbbar, \quad \forall (\b, \r) \in \B \times \R, \label{subeq:D11}\\
        & \um \geq 0, \quad \toll_e \geq 0, \quad \forall \m \in \M, \quad \forall \e \in \E. \label{subeq:D12}
\end{align}
\end{subequations}
The dual variables $\u = \left(\um\right)_{\m \in \M}$ and $\toll = \left(\toll_e\right)_{\e \in \E}$ can be viewed as agents' utilities and the edge prices, respectively. In \eqref{eq:D1bar}, the objective $U(\u, \tau)$ equals the sum of all agents' utilities and the total collected edge prices, and \eqref{subeq:D11} is the same as the stability condition in \eqref{eq:stability}. %We take the dual variables $\u = \(\um\)_{\m \in \M}$ and $\toll = \(\tolle\)_{\e \in \E}$ as agents' utilities and the toll prices, respectively. Given any trip assignment $\xbar$ and any utility vector $\u$, from \eqref{eq:u_p}, the price of each agent $\m \in \M$ can be computed as the difference between agent's value of the assigned trip and the utility:
%The following theorem shows that there exists a carpooling mechanism that satisfies all five design goals if and only if the linear program \eqref{eq:LP1bar} has integer optimal solutions. Then, the optimal trip assignment of the mechanism is an integer optimal solution of \eqref{eq:LP1bar}. The toll prices and agents' utilities are the optimal solutions of the dual program \eqref{eq:D1bar}. 
It is well-established in the literature that the existence of market equilibrium is equivalent to the primal linear program (LP) having an integral optimal solution, see for example in the book \citet{vohra2011mechanism}. For completeness, we state this result as the following proposition.
\medskip 

\begin{proposition}\label{prop:primal_dual}
A market equilibrium $\(\xbaropt, \popt, \tollopt\)$ exists if and only if \eqref{eq:LP1bar} has an integer optimal solution. Any integer optimal solution $\xbaropt$ of \eqref{eq:LP1bar} is an equilibrium trip  allocation vector, and any optimal solution $\(\uopt, \tollopt\)$ of \eqref{eq:D1bar} is an equilibrium utility vector and an equilibrium edge price vector. The equilibrium payment vector $\popt$ is given by:  
\begin{align}\label{eq:p}
    \pmopt=\sum_{\r \in \R} \sum_{\bbar \ni \m} \xrbopt \vrmbbar - \uopt_m, \quad \forall \m \in \M.
\end{align}
% Furthermore, any market equilibrium $\(\xopt, \popt, \tollopt\)$ satisfies
% \begin{itemize}
%     \item[(i)] $\sum_{e \in \r} \toll_e^{\start+\distre*} \leq \sum_{\e \in \r'} \toll_{e}^{\start+d_{r', e}*}$ for any $\start=1, \dots, T$, and any $\r, \r' \in \R$ such that $\tr \geq d_{r'}$.
%     \item[(ii)] $\sum_{e \in \r} \toll_e^{\start+\distre*} \leq \sum_{e \in \r} \toll_e^{\start'+\distre*}$ for any $\start \geq \start'$ and any $\r \in \R$.
% \end{itemize}
\end{proposition}

%The primal-dual formulation in Proposition \ref{prop:primal_dual} enables us to convert the existence problem to the existence of integer optimal solution in \eqref{eq:LP1bar}. The connection between market equilibrium existence and integral solution exists in the primal LP is widely known in literature, see for example \citet{}. We include proof of Proposition \ref{prop:primal_dual}, we show that the four properties of market equilibrium -- individual rationality, stability, budget balance, and market clearing -- are equivalent to the feasibility constraints and the complementary slackness conditions in \eqref{eq:LP1bar} and \eqref{eq:D1bar}. Following from strong duality, a market equilibrium exists if and only if the optimality gap between the linear relaxation \eqref{eq:LP1bar} and the integer problem \eqref{eq:IP} is zero. That is, the existence of market equilibrium is equivalent to the existence of integer optimal solutions in \eqref{eq:LP1bar}. 

%Proposition \ref{prop:primal_dual} demonstrates that when a market equilibrium exists, the equilibrium trip vector $\xopt$ is also socially optimal. The paper \cite{ostrovsky2019carpooling} provided the same observation using a different approach of showing the equivalence between the quasi-outcome of equilibrium carpooling and the competitive equilibrium in a market with road capacity being divisible goods. They also provided a simple example to demonstrate that a market equilibrium may not exist due to the integer constraints on trip organization and road capacity. 

Even if we ignore the integer constraints, the linear programs \eqref{eq:LP1bar} and \eqref{eq:D1bar} cannot be directly used to compute market equilibrium  because the primal (resp. dual) program has an exponential number of variables (resp. constraints). When \eqref{eq:LP1bar} does not have an integral solution, Proposition \ref{prop:primal_dual} indicates that even if \eqref{eq:IP} is solved, there does not exist an edge price vector and a payment vector to implement the optimal trip allocation so that agents are willing to optimally share the network capacity. In Sections \ref{sec:tractable_pooling}, we provide conditions that guarantee the existence and tractability of integer solution in \eqref{eq:LP1bar}, and algorithms to compute the market equilibrium. We extend the results of the tractable case to general networks and add the temporal dimension of trip allocation in Section \ref{sec:multipop}.  %  and show that linear programs \eqref{eq:LP1bar} and \eqref{eq:D1bar} can be used to compute a market equilibrium. %Furthermore, \emph{(i)} shows that the toll price of shorter routes is no less than that of longer routes given any departure time. This is because for any agent group, taking a shorter route results in a higher trip value than taking a longer route. Similarly, \emph{(ii)} indicates that the toll price of any route with late departure time is no less than that with early departure time due to the fact that trips with early departure time incurs no higher late-arrival cost when taking the same route.  %since we obtain that the total toll prices of shorter routes  must be no less than that of the longer ones (routes with higher travel time). 

%Building on Proposition \ref{prop:primal_dual}, we next provide sufficient conditions for the existence of integer solution in \eqref{eq:LP1bar} (Sec. \ref{sec:tractable_pooling}), and a polynomial-time algorithm to compute equilibrium when it exists (Sec. \ref{sec:algorithm}). 
%%novalidate
\section{Existence and properties of market equilibrium}\label{sec:tractable_pooling}We provide sufficient conditions that guarantee market equilibrium existence in Section \ref{subsec:sufficient}, and present the equilibrium computation in Section \ref{sec:algorithm}. In Section \ref{sec:strategyproof}, we identify a market equilibrium that is equivalent to the outcome of the classical Vickrey–Clarke–Groves (VCG) mechanism. This equilibrium achieves the maximum utility of each player among all market equilibria. 
\subsection{Sufficient conditions for equilibrium existence}\label{subsec:sufficient}
%In this section, we characterize sufficient conditions on network topology and trip values under which there exists a market equilibrium. 
Before we present the results, we first introduce the definition of a series-parallel network: 

\medskip 
\begin{definition}[Series-parallel network \cite{duffin1965topology, valdes1979recognition}]\label{def:series_parallel}  A series-parallel network is a graph that can be recursively constructed from a single-edge network using a finite number of series and parallel compositions. Formally, (i) a single-edge network is a series-parallel network; (ii) Series composition merges the sink of one series-parallel sub-network to the source of another series-parallel sub-network; (iii) Parallel composition combines two series-parallel sub-networks by merging their respective sources and sinks. 
\end{definition}\citet{milchtaich2006network} showed that a network is series-parallel if and only if a Wheatstone structure as in Figure \ref{fig:network} cannot be embedded in the network.
\begin{figure}[htp]
\centering
\includegraphics[width=0.3\textwidth]{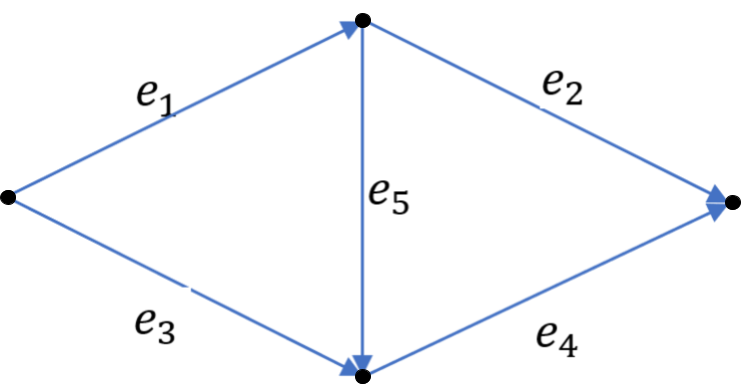}
    \caption{A Wheatstone network.}
\label{fig:network}
\end{figure}
\begin{theorem}\label{theorem:sp}
Market equilibrium $\left(\xopt, \popt, \tollopt\right)$ exists if (i) the network is series-parallel, and (ii) the capacity sharing disutility parameters are homogeneous among all agents:  
\begin{align}\label{eq:homogeneous}
     \changetrip_m(|b|) =\changetrip(|\b|), \quad \changevm_m(|\b|) = \changevm(|\b|), \quad  \forall |b|=2, \dots, A, \quad \forall \m \in \M. 
\end{align} 
\end{theorem}

Theorem \ref{theorem:sp} shows that the sufficient condition for equilibrium existence includes both the condition on network topology -- being series-parallel -- and the condition on capacity sharing disutility parameters -- being identical. We note that the homogeneous capacity sharing disutility condition still allows agents to have heterogeneous trip values -- agents can have different parameters $\alpha_m$ and $\beta_m$ in \eqref{eq:m_valuation_trip}.

 We next provide two examples illustrating that market equilibrium may fail to exist when either of the two sufficient conditions in Theorem \ref{theorem:sp} is violated. In particular, Example~\ref{ex:wheatstone} considers a Wheatstone network, which violates the series-parallel condition, while all agents have zero capacity-sharing disutility. Example~\ref{ex:carpool}, on the other hand, considers a simple network with two parallel routes and two subsets of agents that differ in their capacity-sharing disutilities, thereby violating the homogeneity condition. In both examples, market equilibrium does not exist. 

\medskip 
\begin{example}\label{ex:wheatstone}
{\normalfont Consider the Wheatstone network as in Figure \ref{fig:network}. The capacity of each edge in $\{e_1, e_2, e_3, e_4\}$ is 1, and the capacity of edge $e_5$ is 4. The time cost of each edge is given by $d_{e_1}=1$, $d_{e_2}=2$, $d_{e_3}=2$, $d_{e_4}=1$, and $d_{e_5}=0.2$.  The maximum coalition size is $A=2$. Three agents $\m=1, 2, 3$ have identical preference parameters: $\tripm=6$, $\vm=1$, $\changetrip_m(|b|)=0$ and $\changevm_m(|b|)=0$ for any $|b|=1,2$ and any $\m \in \M$. 
We define route $e_1$-$e_2$ as $r_1$, $e_1$-$e_5$-$e_4$ as $r_2$, and $e_3$-$e_4$ as $r_3$. Trip values are: $V_{r_1}(\{m\})=V_{r_3}(\{m\})=3$, and $V_{r_2}(\{m\})=3.8$ for all $m \in M$; $V_{r_1}(\{m, m'\}) = V_{r_3}(\{m, m'\})  =6$, and $V_{r_2}(\{m, m'\})=7.6$ for all $m, m' \in \M$.  The unique optimal solution of the linear program \eqref{eq:LP1bar} is $x^*_{r_1}(\{1, 2\})=x^*_{r_2}(\{2, 3\})=x^*_{r_3}(\{1, 3\})=0.5$, and $S(\xopt)=9.8$. That is, \eqref{eq:LP1bar} does not have an integer optimal solution, and market equilibrium does not exist (Proposition \ref{prop:primal_dual}).}\hfill $\square$
\end{example}

\medskip
\begin{example}\label{ex:carpool}
    {\normalfont 
Consider a network with two parallel edges $e_1, e_2$. Both edges have a capacity of 1 and a time cost of $d_{e_1} = d_{e_2} = 1$. The maximum coalition size is $A = 6$. Twelve agents participate in this market. Agents $1, 2, \ldots, 6$ have the following preference parameters: $\alpha_m = 50$, $\beta_m = 1/6$, $\changetrip_m(|b|) = 0.25(|b| - 1)$ for $ |b| \leq 5$, and $ \changetrip_m(|b|) = 0.5(|b|-1)$ for $|b| =6$. Agents $7, 8, \ldots, 12$ have parameters $\alpha_m = 100$, $\beta_m = 0.5$, $\changetrip_m(|b|) = 2(|b| - 1)$ if $|b| \leq 4$ and infinity otherwise. 
Furthermore $\changevm_m(|b|) = 0$ for any $|b| = 1, 2, 3, 4, 5, 6$, and any $m \in M$. The optimal solution to the LP-relaxation is $
    x_{e_1}(\{1, 2, 3, 4, 5, 6\}) = 0.5$, $x_{e_1}(\{9, 10, 11, 12\}) = 0.5$, $x_{e_2}(\{7, 8, 10, 12\}) = 0.5$, $x_{e_2}(\{7, 8, 9, 11\}) = 0.5$. This solution has a value of 662.5. The integer optimal solution schedules the trip $\{1, 2, 3, 4, 5, 6\}$ at time 1 on $e_1$, and $\{9, 10, 11, 12\}$ at time 1 on $e_2$; this solution has value of $621< 662.5$. This indicates that the LP relaxation does not have an integer optimal solution, and thus market equilibrium does not exist.}
    \hfill $\square$
\end{example}

\medskip 

Theorem \ref{theorem:sp} demonstrates that network topology is crucial for the stability and efficiency of capacity sharing in the market. In step 1 of the theorem proof (Lemma \ref{lemma:FF}), we show that when the network is series-parallel, an integer route flow computed by a greedy algorithm serves as the optimal trip flow, regardless of agents' preferences and how they form coalitions. 
This effectively decouples the optimal trip allocation into two independent parts: first, compute the optimal integer route flow, and then form the optimal agent coalitions that utilize the allocated route capacity according to the flow. However, the conclusion of Lemma \ref{lemma:FF} may not hold when the series-parallel network condition is violated. As shown in Example \ref{ex:wheatstone}, in a non-series-parallel network, the route flow induced by the optimal trip allocation might be fractional and cannot be obtained through a greedy algorithm.
%exists an optimal solution of \eqref{eq:LP1bar} such that the total flow of trips that take each route is integer, and such flow vector can be computed using a greedy algorithm (Algorithm \ref{alg:flow}). This implies that the route flow induced by the optimal trip allocation is integer, and can be directly computed regardless of the coalition formation, which 

Furthermore, in step 2 of the theorem proof, we show that with homogeneous capacity sharing disutilites, the optimal agent coalition formation that satisfies the flow capacity constraint in step 1 is also integer (Lemmas \ref{lemma:condition_gross} -- \ref{lemma:integer}). Again, this step may not hold when the condition of homogeneous capacinteger optimal solutionity sharing disutilites is violated. In Example \ref{ex:carpool}, the route flow induced by the optimal trip allocation is integer on the two-route parallel network, but the optimal coalition formation is fractional. The two steps conclude that \eqref{eq:LP1bar} has an integer optimal solution, and thus equilibrium exists following Proposition \ref{prop:primal_dual}. %We further remark that in Example \ref{ex:wheatstone}, the series-parallel network condition is violated, causing Lemma \ref{lemma:FF} in step 1 (resp. Lemmas \ref{lemma:condition_gross} -- \ref{lemma:integer} in step 2) no longer holds. 
%Notably, the recognition that series-parallel networks promote efficient utilization of network capacity extends to other contexts, including classical nonatomic routing games \citet{milchtaich2006network} and network design with cooperative agents \citet{epstein2007strong} \textcolor{red}{include other references}. 
% Furthermore, the condition of identical disutility for capacity sharing implies that agents with differing capacity sharing disutilities should be allocated to separate markets. 

In Section \ref{subsec:extension_1}, we show that the same result holds in the multi-period model, where agents can choose their departure time. The proof of this extension follows similar ideas as Theorem \ref{theorem:sp} with step 1 being more complex for the need of computing a dynamic optimal route flow instead of a static one, which we will go into details later. Moreover, in Section \ref{sec:extension_2}, we show that when the two conditions are violated, market equilibrium can still exist if capacity pricing is route-based rather than edge-based, with differentiated capacity prices for agents with heterogeneous disutilities. In Section \ref{sec:algorithm}, we also demonstrate that the homogeneous disutility condition is necessary for the polynomial-time computation of market equilibrium.
For the rest of this section, we present the two step proof sketch of Theorem \ref{theorem:sp}. The complete proof is included in Appendix \ref{apx:proof_A}. 

\medskip 

%\textcolor{red}{The proof of Theorem \ref{theorem:sp} is built on ideas from the theory of \emph{earlist arrival network flow problem} (\cite{skutella2009introduction, hoppe2000quickest, ruzika2011earliest}) and the \emph{combinatorial auction theory} (\cite{kelso1982job, gul1999walrasian}).} 
%Recall from Proposition \ref{prop:primal_dual} that showing the existence of market equilibrium is equivalent to proving that \eqref{eq:LP1bar} has an integer optimal solution. In step 1 of the theorem proof, we show that when network is series-parallel, there exists an optimal solution of \eqref{eq:LP1bar} such that the total flow of trips that take each route is integer, and such flow vector can be computed using a greedy algorithm (Algorithm \ref{alg:flow}). In step 2, we show that with homogeneous capacity sharing disutilites, the optimal capacity sharing coalition formation that satisfies the constructed flow constraints in step 1 is also integer.These two steps conclude that \eqref{eq:LP1bar} has an integer optimal solution, and thus equilibrium exists following Proposition \ref{prop:primal_dual}. We further remark that in Example \ref{ex:wheatstone}, the series-parallel network condition is violated, causing Lemma \ref{lemma:FF} in step 1 (resp. Lemmas \ref{lemma:condition_gross} -- \ref{lemma:integer} in step 2) no longer holds. 

% the argument  no longer holds, and (resp. Example \ref{ex:carpool}), (resp. homogeneous carpool disutility condition)
%\medskip 

\noindent\underline{Step 1}. We construct a flow capacity vector $\kopt=(\kopt_r)_{r \in R}$ using the greedy Algorithm \ref{alg:flow} that allocates edge capacity $(q_e)_{e \in E}$ to routes in increasing order of travel time. In this algorithm, we begin with computing a shortest route $\rmin$ with travel time $\tmin$, and set its capacity to be the maximum possible capacity $\kopt_{\rmin} = \min_{\e \in \rmin} \qe$. Then, we reduce the residual capacity of each edge on $\rmin$ by $\kopt_{\rmin}$, and repeat this process until there exists no route with positive residual capacity in the network. %The dynamic flow capacity vector $k^*$ is the temporally repeated flow that allocates $w^{*}_r$ capacity to route $r$ for every feasible departure time $\start=1, 2, \dots, T-\tr$. 
We denote $\Ropt= \{r \in \R|k_r^{*}>0\}$ as the set of routes with positive flow capacity.

\begin{algorithm}[htp]
\SetAlgoLined
\textbf{Initialize:} Set $\tildeE \leftarrow \E$; $\tildeq_e \leftarrow \qe, ~ \forall \e \in \tildeE$; $\kopt_r \leftarrow0, ~\forall \r \in \R$; $\(\tmin, \rmin\) \leftarrow ShortestRoute(\tildeE)$\;
\While{$\tmin < \infty$}{
$\kopt_{\rmin} \leftarrow \min_{\e \in \rmin} \tildeq_e$\;
\For{$\e \in \rmin$}{
$\tildeq_e \leftarrow \tildeq_e - \kopt_{\rmin}$\;
\If{$\tildeq_e = 0$}{
$\tildeE \leftarrow \tildeE \setminus \{\e\}$\;
}
}
$\(\tmin, \rmin\) \leftarrow ShortestRoute(\tildeE)$\;
}
\textbf{Return $\kopt$}
\caption{Greedy algorithm for computing static route capacity $\kopt$}
\label{alg:flow}
\end{algorithm}

We consider another socially optimal trip organization problem \eqref{eq:LP2k}, where trips satisfy the capacity constraints according to $\kopt$. Problem \eqref{eq:LP2k} is more restrictive than the original problem \eqref{eq:LP1bar}, as trip vectors satisfying capacity constraints in \eqref{subeq:LP2k2} must also meet the original network capacity constraint \eqref{subeq:LP12}, but not necessarily vice versa. Lemma \ref{lemma:FF} demonstrates that for series-parallel networks, an optimal solution of \eqref{eq:LP2k} also optimizes the original problem \eqref{eq:LP1bar}.

\begin{lemma}\label{lemma:FF}
If the network is series-parallel, then any optimal solution of \eqref{eq:LP2k} is an optimal solution of \eqref{eq:LP1bar}:  \begin{subequations}
        \makeatletter
        \def\@currentlabel{LP$k^{*}$}
        \makeatother
        \label{eq:LP2k}
        \renewcommand{\theequation}{LP$k^{*}$.\alph{equation}}
        \begin{align}
    \max_{\xbar} \quad &\Sx=\sum_{\b \in \B} \sum_{\r \in \R}\Vrbbar \xrbbar, \notag\\
  s.t. \quad   &\sum_{\b \ni \m}\sum_{\r \in \R} \xrbbar \leq 1, \quad \forall \m \in \M,  \label{subeq:LP2k1}\\
&\sum_{\bbar \in \Bbar} \xrbbar  \leq k^{*}_r, \quad \forall \r \in \R, \label{subeq:LP2k2}\\
 &\xrbbar \geq 0, \quad \forall \bbar \in \Bbar, \quad \forall \r \in \R. \label{subeq:LP2k3}
    \end{align}
    \end{subequations}

\end{lemma}

We prove Lemma \ref{lemma:FF} by construction. We show that for any feasible solution $\x$ of \eqref{eq:LP1bar} on a series-parallel network, we can construct another trip vector $\xhat$ satisfying $S(\xhat) \geq S(\x)$ and feasibility in \eqref{eq:LP2k}. Optimal values of \eqref{eq:LP2k} and \eqref{eq:LP1bar} are equal, making any optimal solution of \eqref{eq:LP2k} optimal in \eqref{eq:LP1bar}.

The key step of the proof is to construct such $\xhat$ by redistributing flow of agent coalitions in $x$, ensuring that agent coalitions with higher time sensitivity are prioritized for shorter routes. The series-parallel network condition is used to show that the flow $\kopt$ is a maximum flow in series-parallel networks (Lemma \ref{lemma:max_flow} in Appendix \ref{apx:proof_A}). Thus, $\xhat$ has the same total flow of each $b$ as in $x$. We also prove, using mathematical induction, that $\xhat$ has higher social welfare compared to $\x$ (i.e. $S(\xhat) \geq S(\x)$) when the network is series-parallel: If the inequality holds on any two series-parallel networks, then it also holds on the network constructed by connecting the two sub-networks in series or in parallel. 

\vspace{0.2cm}

\noindent\underline{Step 2.} In this part, we show that when agents have homogeneous capacity sharing disutilities, \eqref{eq:LP2k} has an integer optimal solution. Following from Lemma \ref{lemma:FF} in {Step} 1, we know that this solution is also an integer optimal solution of \eqref{eq:LP1bar}, and thus conclude Theorem \ref{theorem:sp}. In this step, we need to introduce the monotonicity and gross substitutes definitions. %

\begin{definition}[Monotonicity]\label{def:monotonicity}
A function $f: \B \to \mathbb{R}$ is monotone if $f(\b \cup \b') \geq f(\b), ~ \forall \b, \b' \in \B$.
\end{definition}
The value of a monotonic function $f$ increases as the set $\b$ grows.   
\begin{definition}[Gross Substitutes \cite{reijnierse2002verifying}]\label{def:gross_substitute}
A function $f: \B \to \mathbb{R}$ satisfies the gross substitutes condition if\\
(i) $\forall \b, \b'\in  \B$ such that $\b \subseteq \b'$ and any $i \in \M \setminus \b'$, $f(i|\b') \leq f(i|\b)$, where $f(i|\b) = f(\b \cup \{i\}) -f(\b)$.\\
(ii) $\forall \b \in \B$ and $i,j,k \in \M \setminus \b$, $
        f(\{i, j\}|\b) + f(k|\b) \leq \max \left\{f(i|\b) + f(\{j,k\}|\b), ~ f(j|\b) + f(\{i,k\}|\b)\right\}$, where $f(\{i, j\}|\b) = f(\b \cup \{i, j\}) -f(\b)$.
\end{definition}

%In Definition \ref{def:gross_substitute}, \emph{(i)} requires that a gross substitutes function $f$ is submodular, i.e. the marginal valuation of any $i \in \M$ given $\b$ decreases in the size of $\b$. Additionally, \emph{(ii)} ensures that the sum of marginal values of $\{i, j\}$ and $k$ is not strictly higher than that of both $i, \{j, k\}$ and $j, \{i,k\}$. 

\medskip 
We note that the trip value function $\Vrbbar$ defined on the feasible coalition set $\B$ in \eqref{eq:value_of_trip} does not satisfy the monotonicity condition because the size of the combined coalition $\bbar \cup \bbar'$ may exceed the capacity limit $A$ and the value $V_r(\bbar \cup \bbar')$ may be less than $V_r(\bbar)$ when the capacity sharing disutility is sufficiently high. We denote all coalitions (with sizes both within $A$ or larger than $A$) as $\Ball \deleq 2^M$. Then, we define the augmented value function $\barV_r: \Ball \to \mathbb{R}$, where $\barV_r(\barb)$ takes the maximum value of $V_r(\bbar)$ with coalition $\bbar \subseteq \ball$. We denote a feasible coalition in $\ball$ that achieves this maximum value as the representative coalition, $\rep(\ball)$, and denote the set of all such feasible coalitions as $H_r(\ball)$: %We denote the feasible coalition in $\ball$ that achieves this maximum value as the representative coalition $\rep(\ball)$: 
\begin{align}\label{eq:Vtilde}
    \barV_r(\barb) \deleq \max_{\bbar \subseteq \ball, ~ \bbar \in \Bbar} \Vbar_r(\bbar),  \quad H_r(\ball) \deleq \argmax_{\bbar \subseteq \ball, ~ \bbar \in \Bbar}\Vbar_r(\bbar), \quad \forall  (\ball, \r) \in \Ball \times \R. 
\end{align}

The augmented value function $\barV$ satisfies the monotonicity condition. When all agents have homogeneous capacity sharing disutilities, $\barV$ also satisfies the gross substitutes condition. 
%\footnote{We show that when agents have heterogeneous capacity sharing disutilities, the augmented value function may not be gross substitutes. Consider three agents $\m=1, 2, 3$ and a single route $\r$ with $d_r=10$. The maximum coalition size is $A=2$. The parameters of agents are $\alpha_1=\alpha_2=\alpha_3=100$, $\beta_1 =\beta_2=6$, $\beta_3=4$, $\changetrip_1(2) = \changetrip_2(2)=\changetrip_3(2)=0$, $\changevm_1(2)=\changevm_2(2)=0$, and $\changevm_3(2)=3$. We compute the value function of trips as $V(\{1\}, r) = V(\{2\}, \r)= 40$, $V(\{3\}, r)=70$, $V(\{1,2\}, r)= 80$, $V(\{1, 3\}, r)=V(\{2, 3\}, r)=70$. The augmented trip value function is given by $\bar{V}(\ball, r)=V(\b, r)$ for any $|\ball|\leq 2$, and $\bar{V}(\{1, 2, 3\}, r)=V(\{1, 2\}, r)=80$. We can check that $\bar{V}(\{1\}, r)+\bar{V}(\{2, 3\}, r) = 110$, $\bar{V}(\{2\}, r)+\bar{V}(\{1, 3\}, r)=110$, and $\bar{V}(\{3\}, r)+ \bar{V}(\{1, 2\}, r)=150$. The gross substitutes condition (ii) is violated because $\bar{V}(\{3\}, r)+ \bar{V}(\{1, 2\}, r) > \max\left\{\bar{V}(\{1\}, r)+\bar{V}(\{2, 3\}, r), \bar{V}(\{2\}, r)+\bar{V}(\{1, 3\}, r)\right\}$. We will show in Section \ref{sec:algorithm} that gross substitutes condition is not only crucial for the equilibrium existence, but also important to guarantee that equilibrium can be computed in polynomial time.}

% 

% \manxi{Check how GS is referred in literature.}
\begin{lemma}\label{lemma:condition_gross}
For any $r \in R$, the augmented value function $\bar{V}_r(\cdot)$ is monotone. Additionally, $\bar{V}_r(\cdot)$ satisfies the gross substitutes condition for all $r \in R$ if agents have homogeneous capacity sharing disutilities.
\end{lemma}

By replacing the original value function $V$ with the augmented value function $\barV$ in \eqref{eq:LP2k}, we show that the corresponding linear program has an integer optimal solution $\yopt$ when $\barV$ satisfies the monotonicity and gross substitutes conditions. Furthermore, we show that we can construct an integer optimal solution $\xopt$ of the original \eqref{eq:LP2k} by replacing the augmented coalition with the represented coalition in all trips in $\yopt$.

\begin{lemma}\label{lemma:integer}
The following linear program has an integer optimal solution $\yopt=\left(\xbarr^{*}_r(\ball)\)_{(\ball, \r) \in \Ball \times \R}$ if the augmented value function $\barV$ satisfies monotonicity and gross substitutes:
\begin{subequations}
        \makeatletter
        \def\@currentlabel{$\overline{\mathrm{LP}}\kopt$}
        \makeatother
        \label{eq:LPy}
        \renewcommand{\theequation}{$\overline{\mathrm{LP}}\kopt$.\alph{equation}}
        \begin{align}
    \max_{\xbarr} \quad &S(\xbarr)=\sum_{\ball \in \Ball}\sum_{\r \in \R} \barV_{\r}(\ball) \xbarr_{\r}(\ball), \notag\\
  s.t. \quad  & \sum_{\ball \ni \m} \sum_{\r \in \R}\xbarr_{\r}(\ball) \leq 1, \quad \forall \m \in \M,  \label{subeq:LPy1}\\
&\sum_{\ball \in \Ball}\xbarr_{\r}(\ball)  \leq \kopt_r, \quad \forall \r \in \R, \label{subeq:LPy2}\\
 &\xbarr_{r}(\ball) \geq 0, \quad \forall \ball \in \Ball, \quad \forall \r \in \R. \label{subeq:LPy3}
    \end{align}
    \end{subequations}

Furthermore, given an integer optimal solution $\yopt$, any $\xopt$ that satisfies the following constraints is an integer optimal solution of \eqref{eq:LP2k}: 
\begin{align}\label{eq:construct}
    {\sum_{\b \in H_r(\ball)} x^{*}_{r}(b)}&=\xbarr^{*}_{r}(\ball), \quad x^{*}_{r}(b) \in \{0, 1\}, \quad \forall (\b, r) \in \B \times \R, \quad \forall (\ball, r) \in \Ball \times \R.
\end{align}
\end{lemma}

To prove that \eqref{eq:LP2k} has an integer optimal solution, we view each unit capacity of route $r$ as a ``slot''. Thus, given $\kopt$, there are $|L_r| = k_{r}^{*} $ number of slots for each $r \in \Ropt$. The total number of slots is $|L| = \sum_{r \in \R} \kopt_r $. We demonstrate that the agent assignment problem is equivalent to the good allocation problem in an auxiliary economy, where agents are indivisible goods and slots are buyers. Following a similar primal and dual analysis as in Proposition \ref{prop:primal_dual}, we show that the existence of an integer solution in \eqref{eq:LPy} is equivalent to the existence of Walrasian equilibrium (\cite{kelso1982job}, see Definition \ref{def:we} in Appendix \ref{apx:proof_A}) of our constructed economy. With monotonicity and gross substitutes conditions satisfied, the Walrasian equilibrium exists, and \eqref{eq:LPy} has an integer optimal solution $\xbarr^*$. Thus, $\xopt$ in \eqref{eq:construct} is an integer optimal solution of \eqref{eq:LP2k} and, by Lemma \ref{lemma:FF}, also of \eqref{eq:LP1bar}, concluding the proof of Theorem \ref{theorem:sp}.

%%novalidate
\subsection{Computing Market Equilibrium}\label{sec:algorithm}
\textbf{Computing optimal trip organization.} We compute the optimal trip vector $\xopt$ in two steps following the proof of Theorem \ref{theorem:sp}: \emph{(Step 1)} Compute the optimal route capacity vector $k^*$ from Algorithm \ref{alg:flow}. In each iteration of Algorithm \ref{alg:flow}, the shortest route of the network is computed by Dijkstra algorithm in time $O(|\N|^2)$, where $|\N|$ is the number of nodes in the network. Moreover, since the capacity of at least one edge is completely allocated to the shortest route of every iteration, the number of iterations in Algorithm \ref{alg:flow} is less than or equal to $|\E|$. Therefore, the time complexity of step 1 is $O(|\E||\N|^2)$. 

\emph{(Step 2)} Compute $\xopt$ as an integer optimal solution of \eqref{eq:LP2k} by allocating agents to the set of slots $\L$ given by $\kopt$. Following Lemmas \ref{lemma:condition_gross} and \ref{lemma:integer}, we know that the second step of computing the integer optimal solution of \eqref{eq:LP2k} cinteger optimal solutionquivalently turned into computing the allocation $\bar{x}^*$ in a Walrasian equilibrium of the auxiliary economy with the augmented value function $\bar{V}_r$ for each $r \in R$. The actual equilibrium allocation $\xopt$ can then be recovered from \eqref{eq:construct}. 

Several approaches have been proposed for computing Walrasian equilibrium in an economy with indivisible goods including the pseudo-polynomial time algorithms (\cite{kelso1982job, ausubel2002ascending, parkes1999bundle, de2007ascending}), and the polynomial time algorithms (\cite{nisan2006communication, murota1996valuated, paes2020computing}).
All of the above methods can be used for computing the solution in step 2 of our problem. These algorithms require the knowledge of the augmented value function and the oracle that can efficiently compute the set $ \argmax_{\ball \in \Ball} \{\bar{V}_r(\ball) - \sum_{m \in \ball} u_m\}$. In our setting, since the augmented value function $ \barV_r(\ball)$ satisfies monotonicity and gross substitutes conditions (Lemma \ref{lemma:condition_gross}), we compute the set $\argmax_{\ball \in \Ball} \{\bar{V}_r(\ball) - \sum_{m \in \ball} u_m\}$ by iteratively adding agents into the solution set greedily according to their marginal contribution to the value of the objective function $\bar{V}_r(\ball) - \sum_{m \in \ball} u_m$  (\cite{gul1999walrasian}, also included in Appendix \ref{sec:review}). We note that the greedy approach can be used to compute this set if and only if $\barV_r$ satisfies the monotonicity and gross substitutes conditions (\cite{gul1999walrasian}). Therefore, the homogeneous capacity sharing disutility condition is essential for the polynomial time complexity of the algorithm. To make the paper self-included, we provide the details of our algorithm that computes $x^*$ in step 2 in Appendix \ref{ap:alg}.
Our algorithm is based on the Kelso-Crawford algorithm (\cite{kelso1982job}) with the modification that enables us to efficiently compute $\bar{V}_r(\cdot)$ and $h_r(\cdot)$ in an iterative manner only for a subset of $(\ball, r)$ when needed. Algorithm \ref{alg:allocation} is pseudo-polynomial in the number of agents $|M|$ and the number of unit capacity flow slot $|L| = \sum_{r \in R}\kopt_r$. We remark that other Walrasian equilibrium computation algorithms can also be used to compute the trip allocation in this step.

\textbf{Computing equilibrium payments and edge prices.}
Given the optimal trip vector $\xopt$, we compute the set of agent payments $\popt$ and edge prices $\tollopt$ such that $\left(\xopt, \popt, \tollopt\right)$ is a market equilibrium. Recall from Proposition \ref{prop:primal_dual}, the equilibrium utilities and edge prices $\left(\uopt, \tollopt\right)$ are optimal solutions of the dual program \eqref{eq:D1bar}. 
We can use the Ellipsoid method to compute $(u^*, \toll^*)$ given that the separation problem -- identifying a violated constraint in \eqref{eq:D1bar} for any $(u, \toll)$ -- can be solved (\cite{nisan2006communication, grotschel1993ellipsoid}). To verify constraints \eqref{subeq:D11}, we need to check whether or not $ \max_{\b \in \B} \{V_r(\b) -\sum_{m \in \b} u_m \} \leq \sum_{e \in \r} \toll_e$ is satisfied for all $\r \in \R$. We note that
$\max_{\b \in \B} \{V_r(\b) -\sum_{m \in \b} u_m \}= \max_{\ball \in \Ball} \{\bar{V}_r(\ball) -\sum_{m \in \ball} u_m \}$. Under the monotonicity and gross substitutes conditions, $\max_{\ball \in \Ball} \{\bar{V}_r(\ball) -\sum_{m \in \ball} u_m \}$ can be computed by greedily adding agents to the set $\ball$ (Algorithm \ref{alg:allocation} Line 3-29 in Appendix \ref{ap:alg}). Thus, constraints \eqref{subeq:D11} can be verified in $O(|M||R|)$ time. Additionally, 
constraints \eqref{subeq:D12} are straightforward to verify. Thus, an equilibrium utility vector $\uopt$ and an edge price vector $\tollopt$ can be computed by the ellipsoid method in polynomial time in $|M|$ and $|R|$. Based on $\xopt$ and $(\uopt, \tollopt)$, we can compute the payment vector $\popt$ using \eqref{eq:p}.

\subsection{Equivalence to VCG mechanism}\label{sec:strategyproof}
In this section, we identify a particular market equilibrium $(\xopt, \udag, \tolldag)$ that induces the same outcome as the classical VCG mechanism. We also show that $\udag$ achieves the maximum utility for all agents and $\tolldag$ charges the minimum total edge price among the set of equilibrium $(\uopt, \tollopt)$. Throughout this section, we assume that the network is series-parallel and agents have homogeneous capacity sharing disutilities, and thus market equilibrium exists following Theorem \ref{theorem:sp}. 

%In this section, we answer the second open question raised in the introduction (Section \ref{sec:introduction}) on how to ensure that agents truthfully report their preferences to the platform. From Proposition \ref{prop:primal_dual}, we know that given a socially optimal trip organization $\xopt$, any optimal solution $(\uopt, \tollopt)$ of the dual \eqref{eq:D1bar} is a vector of equilibrium utilities and toll prices, and the associated payment vector $\popt$ is given by \eqref{eq:p}. When the dual \eqref{eq:D1bar} has non-unique solutions, there exists a non-singleton set of equilibrium utility and payment vectors. 

% we define a market outcome to be strategyproof if no agent can gain higher utility by misreporting their private preference parameters to the platform. 
% \begin{definition}[Strategyproofness]\label{def:strategyproof}
% A market equilibrium $\left(\xopt, \popt, \tollopt\right)$ is strategyproof if for any agent $m$ with any reported preference parameters $\left(\alpha_{m}^{'}, \beta_{m}^{'}, \pi_{m}^{'}, \gamma_{m}^{'}, \arr_m^{'}, d_m^{'}\right)$, the utility $u^{m*'}$ induced by the corresponding market equilibrium $\left(x^{*'}, p^{*'}, \toll^{*'}\right)$ is no higher than $u^{*}_m$ with the truthfully reported preference parameters $\(\tripm, \vm, \dism, \disum, \arrm, \delaym\)$.
% \end{definition}

A VCG mechanism is defined as $(x^*, \pdag)$, where $\xopt$ is a socially optimal trip organization vector, and payment $\pdag_m$ of each agent $m \in M$ is the difference of the total trip values for all other agents given the socially optimal trip organization with and without agent $\m$: 
\begin{align}\label{eq:pmdag}
    \pmdag  = S_{-\m}(\xoptmm) - S_{-\m}(\xopt), \quad \forall \m \in \M, 
\end{align}
where $\xoptmm$ is the optimal trip vector for the trip organization problem with agent set $M \setminus \{m\}$. The optimal social welfare with $\xoptmm$ is $S_{-\m}(\xoptmm)$ given by \eqref{eq:IP}, and $S_{-\m}(\xopt)= S(\xopt) - \sum_{r \in R}\sum_{\b \ni \m} v_{m,r}(b) x_r^{*}(b)$ is the social welfare for agents $\M \setminus \{\m\}$ with the original optimal trip vector $\xopt$. Given $\xopt$ and $\pdag$, the utility of each agent $\m \in \M$ is the difference of the optimal social welfare with and without $\m$: \begin{align}\label{eq:umaxm}
    \umaxm &\stackrel{\eqref{eq:u_p}}{=} \sum_{\b \ni \m}\sum_{r \in R} v_{m,r}(b) x_r^{*}(\b)- \pmdag \stackrel{\eqref{eq:pmdag}}{=}  
    S(\xopt) - S_{-\m} (\xopt_{-m}), \quad \forall \m \in \M. 
\end{align}From the classical theory of mechanism design \cite{ausubel2006lovely}, we know that a VCG mechanism is strategyproof. That means, if there exists a market platform that centrally organizes trips based on agents' reported preference parameters, then given the socially optimal trip organization $\xopt$, and the VCG payment $\pdag$, all agents will truthfully report their preferences to the platform. To show that there exists a strategyproof market equilibrium, it suffices to demonstrate that we can find a price vector $\tolldag$ such that $\left(\xopt, \pdag, \tolldag\right)$ is a market equilibrium. Next, we show that such $\tolldag$ exists. Moreover, all agents' equilibrium utilities given by $\udag$ are the highest of all market equilibria, and the total collected edge price is the minimum. 

\begin{theorem}\label{theorem:strategyproof}
If the network is series-parallel, and agents have homogeneous capacity sharing disutilities, then a strategyproof market equilibrium $\(\xopt, \pdag, \tolldag\)$ exists, and the equilibrium utility vector is $\udag$. Moreover, given any other market equilibrium $\(\xopt, \popt, \tollopt\)$, 
\begin{align*}
    \umaxm &\geq \umopt, \quad \forall \m \in \M, \quad \text{and} \quad 
    \sum_{\e \in \E} \qe \toll^{\dagger}_e \leq   \sum_{\e \in \E} \qe \toll^*_e. 
\end{align*}
\end{theorem}

Theorem \ref{theorem:strategyproof} shows that there exists a market equilibrium that can be implemented by platforms in a centralized manner -- agents report their private preference parameters to the platform, and the platform mediates the market on the agents' behalf. Our result shows that the platform has to implement the equilibrium that maximizes agents' utilities in order to ensure that agents will not lie about their preferences.

We prove Theorem \ref{theorem:strategyproof} in two steps that can be viewed as the dual of the two steps in the proof of Theorem \ref{theorem:sp}: Firstly, we show that a utility vector $\uopt$ is an equilibrium utility (i.e. there exists a  price vector $\tollopt$ such that $(\uopt, \tollopt)$ is an optimal solution of \eqref{eq:D1bar}) if and only if there exists a vector $\lambda^{*}= (\lambda^{*}_r)_{r \in R}$ such that $(\uopt, \lambda^{*})$ is an optimal solution of \eqref{eq:D2k} -- the dual of \eqref{eq:LP2k} (Lemma \ref{lemma:utility}). 
\begin{subequations}
        \makeatletter
        \def\@currentlabel{D$k^*$}
        \makeatother
        \label{eq:D2k}
        \renewcommand{\theequation}{D$k^*$.\alph{equation}}
    \begin{align}
         \min_{\u, \lambda}  \quad & \sum_{\m \in \M} \um + \sum_{\r \in \R}  \kopt_r \lambda_r, \notag\\
   s.t. \quad  &  \sum_{\m \in \b} \um + \lambda_r \geq \Vrbbar, \quad \forall (\b, \r) \in \B \times \R, \label{subeq:D2k1}\\
    &\um \geq 0, ~\lambda_r \geq 0, \quad \forall \m \in \M,  \quad \forall \r \in \R.\label{subeq:D2k2}
    \end{align}
\end{subequations}
Here, $\lambda$ is the dual variable of constraint \eqref{subeq:LP2k2}, which can be viewed as the price for routes (instead of for edges as in $\toll$). In particular, $\lambda_r$ is the price for using a unit capacity on route $\r$. Thus, step 1 indicates that the set of agents' equilibrium utilities with edge-based  pricing in the original network is the same as the set of equilibrium utilities achieved with route-based  pricing when trips are organized according to the flow capacity $\kopt$. Secondly, we demonstrate that $\udag$ is a part of an optimal solution of \eqref{eq:D2k}, and the set of all equilibrium utility vectors is a lattice with the maximum element being $\udag$ (Lemma \ref{lemma:utility_lattice}). This step leverages the connection between the equilibrium coalition formation given $\kopt$ and the Walrasian equilibrium of the auxiliary economy, and follows the results in \cite{gul1999walrasian}.

\section{Extensions of the basic model}\label{sec:multipop}
\subsection{Extension I: Capacity sharing over time}\label{subsec:extension_1}
In this section, we extend the static capacity sharing problem to multiple time steps $t \in [T]:= \{1, 2, \dots, T\}$. Agents, after forming a coalition $\b$, choose a route $r$ and a departure time $z \in [T]$. Thus, a trip is defined as $(\b, \r, \start)$. Agent $m$'s valuation of a trip $(\b, \r, \start)$ is given by: 
$$v_{m,r}^z(b) = \tripm - \vm \tr - \changetrip_m(|\b|)- \changevm_m (|\b|)\tr - \delaym((\start+\tr-\arrm)_{+}),$$
where $\arrm$ is agent $m$'s preferred latest arriving time, $(\start+\tr-\arrm)_{+} = \max\{\start+\tr-\arrm, 0\}$ is the amount of time agent $m$ being late, and $\delaym((\start+\tr-\arrm)_{+})$ is agent $m$'s cost of delay. The function $\delaym: \mathbb{R}_{\geq 0} \to \mathbb{R}_{\geq 0}$ can be any non-decreasing function specific to agent $m$, and $\delaym(0)=0$ for all $m \in M$. The value of a trip $(b, r, z)$ is $V_r^z(b) = \sum_{m \in b} v_{m, r}^z(b)$.

Similar to that in the static model, the outcome of the market is represented by the tuple $(x, p, \tau)$. In the multi-period model, a trip allocation vector $x= (x_r^z(b))_{r \in \R, b \in B, z \in [T]}$ is feasible if the flow induced by $x$ that enters each edge $e$ at time $t$ is less than or equal to the edge capacity:
 $$\sum_{r \in R}\sum_{\b \in \B}x_r^{t - \distre}(b) \leq \qe, \quad \forall \e \in \E, \quad \forall t\in [T],$$
where $\distre$ is the time cost from the origin to the beginning of edge $e$ along the route $r$.
Moreover, the edge price vector $\tau=(\tau_{e}^t)_{e \in E, t \in [T]}$ specifies the price of using a unit capacity when entering edge $e$ at time $t$. Thus, the price for a trip $(\start, \bbar, \r)$ equals $\sum_{\e \in \r} \toll_e^{z+\distre}$. Following Proposition \ref{prop:primal_dual}, market equilibrium exists if and only if the linear relaxation of the optimal trip allocation in multi-period model has an integer optimal solution: 
\begin{subequations}\label{eq:LP1new}
    \begin{align}
         \max_{\xbar} \quad &\Sx=\sum_{\start=1}^T \sum_{r \in R}\sum_{\b \in \B}V_r^z(b) x_r^z(b), \notag\\
  s.t. \quad   &\sum_{\start=1}^T \sum_{r \in R}\sum_{\b \ni m}x_r^z(b) \leq 1, \quad \forall \m \in \M,  \label{subeq:LP11new}\\
&\sum_{r \in R}\sum_{\b \in \B}x_r^{t - \distre}(b)  \leq \qe, \quad \forall \e \in \E, \quad \forall t=1, \dots, T, \label{subeq:LP12new}\\
    &x_r^z(b) \geq 0, \quad \forall \start =1, \dots, T, \quad \forall r\in R, \quad \forall b \in B. \label{subeq:LP13new}
    \end{align}
\end{subequations}
All results in Section \ref{sec:tractable_pooling} can be extended to the multi-period model. In particular, %market equilibrium exists if and only the linear relaxation of the optimal trip allocation integer program has an optimal integer solution: 
% \begin{subequations}
%         % \makeatletter
%         % \def\@currentlabel{$\mathrm{LP}$}
%         % \makeatother
%         % \label{eq:LP1bar}
%         % \renewcommand{\theequation}{$\mathrm{LP}$.\alph{equation}}
%     \begin{align*}
%          \max_{\xbar} \quad &\Sx=\sum_{(\start, \r, \b) \in \Trip}\Vrbbar \xrbbar, \\
%   s.t. \quad   &\sum_{(\start, \r, \b) \in \{\Trip | \b \ni \m\}} \xrbbar \leq 1, \quad \forall \m \in \M, \\
% &\sum_{(\start, \r, \b) \in \{\Trip \}}x_r^{t - \distre}(b)  \leq \qe, \quad \forall \e \in \E, \quad \forall t=1, \dots, T, \\
%     &\xrbbar \geq 0, \quad \forall (\start, \r, \b) \in \Trip.
%     \end{align*}
% \end{subequations}
market equilibrium exists under the same conditions that the network is series-parallel, and all agents have homogeneous capacity sharing disutilities (Theorem \ref{theorem:sp}). 

Similar to that in the static model, we take a two-step approach in the proof. First, in series-parallel networks, we can construct a flow vector $w^*=(w^{t*}_{r})_{r \in R, t \in [T]}$ such that an optimal trip allocation restricted to allocating $w^{t*}_{r}$ capacity to each route $r$ with each departure time $t$ is also an optimal solution of the original problem (Lemma \ref{lemma:FF_dynamic}). This result is an extension of Lemma \ref{lemma:FF} in the static model. Interestingly, $w^*$ is a \emph{temporally repeated flow} of the static flow vector $k^*$ computed in Algorithm \ref{alg:flow}, i.e. $$w^{t*}_r= \kopt_r, \quad \forall r \in R, \quad \forall t=1, \dots, T- d_r.$$ 

\begin{lemma}\label{lemma:FF_dynamic}
If the network is series-parallel, then any optimal solution of \eqref{eq:LP2new} is an optimal solution of \eqref{eq:LP1new}:  
\begin{subequations}\label{eq:LP2new}
        \begin{align}
    \max_{\xbar} \quad &\Sx=\sum_{\start=1}^T \sum_{r \in R}\sum_{\b \in \B}V_r^z(b) x_r^z(b) \notag\\
  s.t. \quad   &\sum_{\start=1}^T \sum_{r \in R}\sum_{\b \ni m} x_r^z(b) \leq 1, \quad \forall \m \in \M,  \label{subeq:LP2k1new}\\
&\sum_{\b \in \B}x_r^z(b)  \leq k^{*}_r, \quad \forall \r \in \R, \quad \forall \start=1, \dots, T-\tr, \label{subeq:LP2k2new}\\
 &x_r^z(b) \geq 0, \quad \forall \bbar \in \Bbar, \quad \forall \r \in \R, \quad \forall \start =1, 2, \dots, T. \label{subeq:LP2k3new}
    \end{align}
    \end{subequations}

\end{lemma}

The proof of Lemma \ref{lemma:FF_dynamic} follows the similar idea of constructing an optimal trip allocation $\hat{x}$ that satisfies the flow capacity constraint associated with $w^*$ from an optimal fractional solution $x$ that may violate the capacity constraint associated with $w^*$. However, the construction of $\hat{x}$ is more complex than that of Lemma \ref{lemma:FF} in the static model due to the need of re-arranging trip start time in $\hat{x}$. 
%The proof uses the fact that in series-parallel network, the temporally repeated flow $w^*$ is also the earliest arrival flow that maximizes the total flow that arrive at the sink before time $T$ in the network (Lemma \ref{} following \citet{}). 
We construct such $\xhat$ by redistributing flow of agent groups in $x$, ensuring no group has later arrival time in $\xhat$ compared to $\x$ and that agent groups with higher time sensitivity are prioritized for shorter routes. The series-parallel network condition is used to show that the temporally repeated flow $w^*$ is the earliest arrival flow in series-parallel networks (Lemma \ref{lemma:earliest_arrival} in Appendix \ref{apx:multi}). Thus, $\xhat$ has the same total flow of each $b$ as in $x$, and the flow arriving before each time step $t$ in $\xhat$ is no less than that in $x$. We also prove, using mathematical induction, that $\xhat$ has higher social welfare compared to $\x$ (i.e. $S(\xhat) \geq S(\x)$) when the network is series-parallel: If the inequality holds on any two series-parallel networks, then it also holds on the network constructed by connecting the two sub-networks in series or in parallel. 
%We prove Lemma \ref{lemma:FF} by construction. We show that on a series-parallel network, for any feasible solution $\x$ of \eqref{eq:LP1bar} on a series-parallel network, we can construct another trip vector $\xhat$ satisfying $S(\xhat) \geq S(\x)$ and feasibility in \eqref{eq:LP2k}. Optimal values of \eqref{eq:LP2k} and \eqref{eq:LP1bar} are equal, making any optimal solution of \eqref{eq:LP2k} optimal in \eqref{eq:LP1bar}.

In the second step of the proof, we follow the arguments analogous to Lemma \ref{lemma:condition_gross} -- \ref{lemma:integer} to verify the gross substitutes condition and prove equilibrium existence in the multi-period model. This also indicates that the equilibrium computation follows a similar procedure to that in the static model: We first compute the temporally repeated flow $w^*$, and then compute the Walrasian equilibrium in the auxiliary market of allocating agents to the flow capacity. Furthermore, Theorem \ref{theorem:strategyproof} also holds for the multi-period model: there exists a market equilibrium that achieves the highest utility of all agents, and that equilibrium also coincides with the outcome of a VCG mechanism. The details of the proofs are included in Appendix \ref{apx:multi}.

\subsection{Extension II: General networks and agent preferences}\label{sec:extension_2}
In this section, we generalize the equilibrium existence and computation results to general networks with multiple sources and sinks and agents with heterogeneous capacity sharing disutilities. We discuss extensions on the pricing mechanism such that the equilibrium still exists in the general setting. We include an equilibrium computation algorithm in Appendix \ref{apx:multi}. Additionally, in Appendix \ref{subsec:numerical}, we demonstrate the effectiveness of this algorithm in computing equilibrium in a general setting through a numerical experiment on carpool pricing, using data collected from the California Bay Area.

For simplicity, we still consider the static model but all results can be extended to the multi-period model as in Section \ref{subsec:extension_1}. In this generalized setting, the set of all agents $M$ is partitioned into a finite number of populations $\{M_i\}_{i \in I}$, where agents in different populations are associated with different source-sink pairs and capacity sharing disutilities. For each $i \in I$, we denote the set of routes connecting the origin and destination as $R_i$. Examples \ref{ex:wheatstone}
 -- \ref{ex:carpool} demonstrate that market equilibrium may not exist in the general setting. 
 To overcome this issue, we consider (i) setting route-based pricing instead of edge-based pricing; (ii) creating separate market for each population $M_i$, i.e. agents in $M_i$ only share trips with others in the same subset, and the set of feasible trip groups of market $i$ is $B_i$. 
 %That is, the price for market $i$ to use route $r$ is $\lambda_{r,i} \geq 0$. %We note that edge-based pricing is a special case of route-based pricing since given any edge price vector $\toll$, we can equivalently obtain a route-based price vector where $\lambda^{i}_r=\sum_{\e \in \r} \toll^{i}_e$. The converse is not necessarily true in that a route-based price vector $\lambda$ may not correspond to the additive sum of edge-based prices. 
 %The next proposition shows that market equilibrium exists Following Theorem \ref{theorem:sp}, we can show that market equilibrium $(x^*, p^*, \lambda^{*})$ exists given any capacity allocation vector $q= (q_{r,i})_{r \in R_i, i\in I}$, where $q_{r,i}$ is the capacity allocated to market $i$ on route $r$.  

\medskip 
 \begin{proposition}\label{prop:extension}
     \begin{enumerate}
         \item[(i)] If the network is series-parallel and $|I|>1$, then market equilibrium $(x^*, p^*, \tau^*)$ exists, where agents from the same population form coalitions among themselves, and $\tau^*_i= (\tau^*_{e,i})_{e \in E}$ is the population-specific edge price vector for $i$. 
         \item[(ii)] If the network is not series-parallel and $|I|=1$, then market equilibrium $(x^*, p^*, \lambda^*)$ exists, where $\lambda^*= (\lambda^*_{r})_{r \in R}$ is the route-based price vector in equilibrium. 
         \item[(iii)] If the network is not series-parallel and $|I|>1$, then market equilibrium $(x^*, p^*, \lambda^*)$ exists, where agents from the same population form coalitions among themselves, and $\lambda^*= (\lambda^*_{r,i})_{r \in R, i \in I}$ is the route-based and population-specific price vector in equilibrium.
     \end{enumerate}
 \end{proposition}
\medskip

% We first show that the equilibrium computation is NP-hard (even for series-parallel networks) when agents have heterogeneous trip sharing disutilities. Then, we present the bound of integrality gap. We also provide an exact (but exponential-time) branch-and-price algorithm for computing that utilizes the insights from the setting with identical disutilities. 

In (i) and (iii), when creating separate market for each agent population, we need to determine an allocation of network capacity to each market associated to a population. Building on Proposition \ref{prop:extension}, we formulate the following integer optimization problem that computes the optimal capacity allocation: 
\begin{subequations}
        \makeatletter
        \def\@currentlabel{${\mathrm{IP}}_{\mathrm{mult}}$}
        \makeatother
        \label{IP-mult}
        \renewcommand{\theequation}{${\mathrm{IP}}_{\mathrm{mult}}$.\alph{equation}}
\begin{align}
\max_{x, f}\quad & \Sx=\sum_{r \in R}\sum_{i \in I}\sum_{\b \in \B_i}\Vrbbar \xrbbar \notag\\
s.t. \quad   &\sum_{r \in R}\sum_{\b \ni \m}\xrbbar \leq 1, && \forall i \in I, ~\forall \m \in \M_i, \label{eq:ipMult-u} \\
& \sum_{\bbar \in \Bbar_i} \xrbbar  \leq f_{r,i}, &&  \forall i \in I, ~  \forall \r \in R_i, \label{ipmult-cap} \\
& \sum_{i \in I} \sum_{r \in \{R_i | r \ni e\}} f_{r,i} \leq q_e, && \forall e \in E, \\ 
& \xrbbar \in \{0, 1\}, \quad f_{r,i} \in \mathbb{Z}_+ 
 && \forall i \in I, ~  \forall \bbar \in B_i, ~ \forall \r \in R_i.
\end{align}
\end{subequations}
In \eqref{IP-mult}, $f=(f_{r,i})_{r \in R_i, i \in I}$ is the route capacity allocation vector, where $f_{r,i}$ is the integer capacity on route $r$ that is allocated to agents in $M_i$.

%The first step of designing such a market is to select a subset of routes connecting each o-d pair, and determine the capacity of each route.
%Instead of solving for the equilibrium of the overall market, we consider the submarkets induced by the subpopulations $\{M_i\}_{i \in I}$.  

%Our results can be used to design the carpooling market when riders are only carpooled with others in the same group. The first step of designing such a market is to select a subset of routes connecting each o-d pair, and determine the capacity of each route. Then, since riders in each group have homogeneous carpool disutilities, we can organize a carpooling market for each group by implementing either edge-based toll pricing (when the selected routes form a series-parallel network) or route-based toll pricing (when routes are non series-parallel). 

%Unlike the setting with homogeneous carpool disutilities, the problem of computing an equilibrium (in which riders are only carpooled with others in the same subpopulation) is a computationally intractable task. 

The problem of computing the socially optimal $f^*$ and equilibrium trip allocation $x^*$ is NP-hard even if the network is series-parallel because \eqref{IP-mult} is a reduction from the NP-hard \emph{edge-disjoint paths} problem. %This indicates that exact polynomial-time algorithms for computing the optimal capacity allocation does not exist. 
We can develop a \emph{Branch-and-Price algorithm} to compute the market equilibrium in the general setting. Let $(\mathrm{LP}_\mathrm{mult})$ be the LP-relaxation of \eqref{IP-mult}, and let $(x^*, f^*)$ be an optimal LP-solution. 
If the capacity allocation vector $f^*$ is integral, one can efficiently compute an equilibrium for each submarket via Algorithm \ref{alg:allocation}. Otherwise, there must exist at least one $(i, r)$ such that $q^{*}_{r,i}$ is fractional. We then branch on this variable to create two subproblems by including one of the two new constraints $q_{r, i} \leq \left\lfloor q^{*}_{r,i} \right\rfloor$, $q_{r,i} \geq \left\lfloor q^{*}_{r,i} \right\rfloor$, and compute the new optimal solution associated with the LP relaxation given the added constraint. 
We note that the computation of the LP relaxation (which has exponential number of variables) builds on the fact that the trip value function in each sub-market satisfies gross substitute condition due to the identical capacity sharing disutility (Lemma \ref{lemma:condition_gross}), and thus can be solved efficiently by column generation. 
The algorithm terminates when all $q^{*}_{r,i}$ variables are integer-valued. A formal description of Algorithm \ref{alg:branchAndPrice} and implementation details are provided in Appendix \ref{apx:multi}, with numerical experiments on a California Bay Area carpooling example reported in Appendix \ref{subsec:numerical}.%We include formal description of Algorithm \ref{alg:branchAndPrice} and the implementation details in Appendix \ref{apx:multi}, and include the numerical experiment of computing market equilibrium using the example of carpooling in California bay area highway network in Appendix \ref{subsec:numerical}. 

\section*{Acknowledgement}
This work has been partly funded by AFOSR grant FA9550-19-1-0263

\bibliographystyle{ACM-Reference-Format}
\bibliography{sn-bibliography}

@article{kelso1982job,
  title={Job matching, coalition formation, and gross substitutes},
  author={Kelso Jr, Alexander S and Crawford, Vincent P},
  journal={Econometrica: Journal of the Econometric Society},
  pages={1483--1504},
  year={1982},
  publisher={JSTOR}
}

@article{faigle1993some,
  title={On some approximately balanced combinatorial cooperative games},
  author={Faigle, Ulrich and Kern, Walter},
  journal={Zeitschrift f{\"u}r Operations Research},
  volume={38},
  pages={141--152},
  year={1993},
  publisher={Springer}
}

@article{toriello2013traveling,
  title={On traveling salesman games with asymmetric costs},
  author={Toriello, Alejandro and Uhan, Nelson A},
  journal={Operations research},
  volume={61},
  number={6},
  pages={1429--1434},
  year={2013},
  publisher={INFORMS}
}

@article{blaser2008approximately,
  title={Approximately fair cost allocation in metric traveling salesman games},
  author={Bl{\"a}ser, Markus and Shankar Ram, Lakshminarayanan},
  journal={Theory of Computing Systems},
  volume={43},
  number={1},
  pages={19--37},
  year={2008},
  publisher={Springer}
}

@article{tamir1989core,
  title={On the core of a traveling salesman cost allocation game},
  author={Tamir, Arie},
  journal={Operations Research Letters},
  volume={8},
  number={1},
  pages={31--34},
  year={1989},
  publisher={Elsevier}
}

@article{faigle1998approximately,
  title={On approximately fair cost allocation in Euclidean TSP games},
  author={Faigle, Ulrich and Fekete, S{\'a}ndor P and Hochst{\"a}ttler, Winfried and Kern, Walter},
  journal={Operations-Research-Spektrum},
  volume={20},
  pages={29--37},
  year={1998},
  publisher={Springer}
}

@article{potters1992traveling,
  title={Traveling salesman games},
  author={Potters, Jos AM and Curiel, Imma J and Tijs, Stef H},
  journal={Mathematical Programming},
  volume={53},
  pages={199--211},
  year={1992},
  publisher={Springer}
}

@article{anshelevich2008price,
  title={The price of stability for network design with fair cost allocation},
  author={Anshelevich, Elliot and Dasgupta, Anirban and Kleinberg, Jon and Tardos, {\'E}va and Wexler, Tom and Roughgarden, Tim},
  journal={SIAM Journal on Computing},
  volume={38},
  number={4},
  pages={1602--1623},
  year={2008},
  publisher={SIAM}
}

@article{von2013optimal,
  title={Optimal cost sharing for resource selection games},
  author={von Falkenhausen, Philipp and Harks, Tobias},
  journal={Mathematics of Operations Research},
  volume={38},
  number={1},
  pages={184--208},
  year={2013},
  publisher={INFORMS}
}

@inproceedings{epstein2007strong,
  title={Strong equilibrium in cost sharing connection games},
  author={Epstein, Amir and Feldman, Michal and Mansour, Yishay},
  booktitle={Proceedings of the 8th ACM conference on Electronic commerce},
  pages={84--92},
  year={2007}
}

@article{hao2024price,
  title={The price of anarchy in series-parallel network congestion games},
  author={Hao, Bainian and Michini, Carla},
  journal={Mathematical Programming},
  volume={203},
  number={1},
  pages={499--529},
  year={2024},
  publisher={Springer}
}

@book{vohra2011mechanism,
  title={Mechanism design: a linear programming approach},
  author={Vohra, Rakesh V},
  volume={47},
  publisher={Cambridge University Press},
  year={2011}
}

@article{chen2010designing,
  title={Designing network protocols for good equilibria},
  author={Chen, Ho-Lin and Roughgarden, Tim and Valiant, Gregory},
  journal={SIAM Journal on Computing},
  volume={39},
  number={5},
  pages={1799--1832},
  year={2010},
  publisher={SIAM}
}

@article{aumann1959acceptable,
  title={Acceptable points in general cooperative n-person games},
  author={Aumann, Robert J},
  journal={Contributions to the Theory of Games},
  volume={4},
  number={40},
  pages={287--324},
  year={1959}
}

@article{holzman2015strong,
  title={Strong equilibrium in network congestion games: increasing versus decreasing costs},
  author={Holzman, Ron and Monderer, Dov},
  journal={International Journal of Game Theory},
  volume={44},
  pages={647--666},
  year={2015},
  publisher={Springer}
}

@article{gkatzelis2016optimal,
  title={Optimal cost-sharing in general resource selection games},
  author={Gkatzelis, Vasilis and Kollias, Konstantinos and Roughgarden, Tim},
  journal={Operations Research},
  volume={64},
  number={6},
  pages={1230--1238},
  year={2016},
  publisher={Informs}
}

@article{harks2021efficient,
  title={Efficient black-box reductions for separable cost sharing},
  author={Harks, Tobias and Hoefer, Martin and Schedel, Anja and Surek, Manuel},
  journal={Mathematics of Operations Research},
  volume={46},
  number={1},
  pages={134--158},
  year={2021},
  publisher={INFORMS}
}

@article{ferguson2023collaborative,
  title={Collaborative Decision-Making and the k-Strong Price of Anarchy in Common Interest Games},
  author={Ferguson, Bryce L and Paccagnan, Dario and Pradelski, Bary SR and Marden, Jason R},
  journal={arXiv preprint arXiv:2311.01379},
  year={2023}
}

@article{andelman2009strong,
  title={Strong price of anarchy},
  author={Andelman, Nir and Feldman, Michal and Mansour, Yishay},
  journal={Games and Economic Behavior},
  volume={65},
  number={2},
  pages={289--317},
  year={2009},
  publisher={Elsevier}
}

@article{holzman1997strong,
  title={Strong equilibrium in congestion games},
  author={Holzman, Ron and Law-Yone, Nissan},
  journal={Games and economic behavior},
  volume={21},
  number={1-2},
  pages={85--101},
  year={1997},
  publisher={Elsevier}
}

@article{paes2020computing,
  title={Computing Walrasian equilibria: Fast algorithms and structural properties},
  author={Paes Leme, Renato and Wong, Sam Chiu-wai},
  journal={Mathematical Programming},
  volume={179},
  number={1},
  pages={343--384},
  year={2020},
  publisher={Springer}
}

@inproceedings{parkes1999bundle,
  title={i Bundle: An efficient ascending price bundle auction},
  author={Parkes, David C},
  booktitle={Proceedings of the 1st ACM Conference on Electronic Commerce},
  pages={148--157},
  year={1999}
}

@article{ausubel2002ascending,
  title={Ascending auctions with package bidding},
  author={Ausubel, Lawrence M and Milgrom, Paul R},
  journal={The BE Journal of Theoretical Economics},
  volume={1},
  number={1},
  pages={20011001},
  year={2002},
  publisher={De Gruyter}
}

@article{de2007ascending,
  title={On ascending Vickrey auctions for heterogeneous objects},
  author={de Vries, Sven and Schummer, James and Vohra, Rakesh V},
  journal={Journal of Economic Theory},
  volume={132},
  number={1},
  pages={95--118},
  year={2007},
  publisher={Elsevier}
}

@article{nisan2006communication,
  title={The communication requirements of efficient allocations and supporting prices},
  author={Nisan, Noam and Segal, Ilya},
  journal={Journal of Economic Theory},
  volume={129},
  number={1},
  pages={192--224},
  year={2006},
  publisher={Elsevier}
}

@article{murota1996valuated,
  title={Valuated matroid intersection II: Algorithms},
  author={Murota, Kazuo},
  journal={SIAM Journal on Discrete Mathematics},
  volume={9},
  number={4},
  pages={562--576},
  year={1996},
  publisher={SIAM}
}

@article{murota2003application,
  title={Application of M-convex submodular flow problem to mathematical economics},
  author={Murota, Kazuo and Tamura, Akihisa},
  journal={Japan Journal of Industrial and Applied Mathematics},
  volume={20},
  number={3},
  pages={257--277},
  year={2003},
  publisher={Springer}
}

@article{harks2011worst,
  title={The worst-case efficiency of cost sharing methods in resource allocation games},
  author={Harks, Tobias and Miller, Konstantin},
  journal={Operations research},
  volume={59},
  number={6},
  pages={1491--1503},
  year={2011},
  publisher={INFORMS}
}

@inproceedings{cole2011inner,
  title={Inner product spaces for minsum coordination mechanisms},
  author={Cole, Richard and Correa, Jos{\'e} R and Gkatzelis, Vasilis and Mirrokni, Vahab and Olver, Neil},
  booktitle={Proceedings of the forty-third annual ACM symposium on theory of computing},
  pages={539--548},
  year={2011}
}

@article{immorlica2009coordination,
  title={Coordination mechanisms for selfish scheduling},
  author={Immorlica, Nicole and Li, Li Erran and Mirrokni, Vahab S and Schulz, Andreas S},
  journal={Theoretical computer science},
  volume={410},
  number={17},
  pages={1589--1598},
  year={2009},
  publisher={Elsevier}
}

@inproceedings{christodoulou2004coordination,
  title={Coordination mechanisms},
  author={Christodoulou, George and Koutsoupias, Elias and Nanavati, Akash},
  booktitle={International Colloquium on Automata, Languages, and Programming},
  pages={345--357},
  year={2004},
  organization={Springer}
}

@inproceedings{rozenfeld2006strong,
  title={Strong and correlated strong equilibria in monotone congestion games},
  author={Rozenfeld, Ola and Tennenholtz, Moshe},
  booktitle={Internet and Network Economics: Second International Workshop, WINE 2006, Patras, Greece, December 15-17, 2006. Proceedings 2},
  pages={74--86},
  year={2006},
  organization={Springer}
}

@article{hoefer2011competitive,
  title={Competitive cost sharing with economies of scale},
  author={Hoefer, Martin},
  journal={Algorithmica},
  volume={60},
  number={4},
  pages={743--765},
  year={2011},
  publisher={Springer}
}

@article{albers2009value,
  title={On the value of coordination in network design},
  author={Albers, Susanne},
  journal={SIAM Journal on Computing},
  volume={38},
  number={6},
  pages={2273--2302},
  year={2009},
  publisher={SIAM}
}

@inproceedings{kollias2011restoring,
  title={Restoring pure equilibria to weighted congestion games},
  author={Kollias, Konstantinos and Roughgarden, Tim},
  booktitle={Automata, Languages and Programming: 38th International Colloquium, ICALP 2011, Zurich, Switzerland, July 4-8, 2011, Proceedings, Part II 38},
  pages={539--551},
  year={2011},
  organization={Springer}
}

@article{fotakis2008atomic,
  title={Atomic congestion games among coalitions},
  author={Fotakis, Dimitris and Kontogiannis, Spyros and Spirakis, Paul},
  journal={ACM Transactions on Algorithms (TALG)},
  volume={4},
  number={4},
  pages={1--27},
  year={2008},
  publisher={ACM New York, NY, USA}
}

@inproceedings{chekuri2006non,
  title={Non-cooperative multicast and facility location games},
  author={Chekuri, Chandra and Chuzhoy, Julia and Lewin-Eytan, Liane and Naor, Joseph and Orda, Ariel},
  booktitle={Proceedings of the 7th ACM conference on Electronic commerce},
  pages={72--81},
  year={2006}
}

@article{johari2004efficiency,
  title={Efficiency loss in a network resource allocation game},
  author={Johari, Ramesh and Tsitsiklis, John N},
  journal={Mathematics of Operations Research},
  volume={29},
  number={3},
  pages={407--435},
  year={2004},
  publisher={INFORMS}
}

@article{grotschel1993ellipsoid,
  title={The ellipsoid method},
  author={Gr{\"o}tschel, Martin and Lov{\'a}sz, L{\'a}szl{\'o} and Schrijver, Alexander and Gr{\"o}tschel, Martin and Lov{\'a}sz, L{\'a}szl{\'o} and Schrijver, Alexander},
  journal={Geometric Algorithms and Combinatorial Optimization},
  pages={64--101},
  year={1993},
  publisher={Springer Berlin Heidelberg}
}

@article{ausubel2006lovely,
  title={The lovely but lonely Vickrey auction},
  author={Ausubel, Lawrence M and Milgrom, Paul and others},
  journal={Combinatorial auctions},
  volume={17},
  pages={22--26},
  year={2006}
}

@article{reijnierse2002verifying,
  title={Verifying gross substitutability},
  author={Reijnierse, Hans and van Gellekom, Anita and Potters, Jos AM},
  journal={Economic Theory},
  volume={20},
  number={4},
  pages={767--776},
  year={2002},
  publisher={Springer}
}

@inproceedings{ostrovsky2019carpooling,
  title={Carpooling and the economics of self-driving cars},
  author={Ostrovsky, Michael and Schwarz, Michael},
  booktitle={Proceedings of the 2019 ACM Conference on Economics and Computation},
  pages={581--582},
  year={2019}
}

@article{gul1999walrasian,
  title={Walrasian equilibrium with gross substitutes},
  author={Gul, Faruk and Stacchetti, Ennio},
  journal={Journal of Economic theory},
  volume={87},
  number={1},
  pages={95--124},
  year={1999},
  publisher={Elsevier}
}

@article{bein1985minimum,
  title={Minimum cost flow algorithms for series-parallel networks},
  author={Bein, Wolfgang W and Brucker, Peter and Tamir, Arie},
  journal={Discrete Applied Mathematics},
  volume={10},
  number={2},
  pages={117--124},
  year={1985},
  publisher={Elsevier}
}

@article{milchtaich2006network,
  title={Network topology and the efficiency of equilibrium},
  author={Milchtaich, Igal},
  journal={Games and Economic Behavior},
  volume={57},
  number={2},
  pages={321--346},
  year={2006},
  publisher={Elsevier}
}

@article{ruzika2011earliest,
  title={Earliest arrival flows on series-parallel graphs},
  author={Ruzika, Stefan and Sperber, Heike and Steiner, Mechthild},
  journal={Networks},
  volume={57},
  number={2},
  pages={169--173},
  year={2011},
  publisher={Wiley Online Library}
}

@article{duffin1965topology,
  title={Topology of series-parallel networks},
  author={Duffin, Richard J},
  journal={Journal of Mathematical Analysis and Applications},
  volume={10},
  number={2},
  pages={303--318},
  year={1965},
  publisher={Academic Press}
}

@inproceedings{valdes1979recognition,
  title={The recognition of series parallel digraphs},
  author={Valdes, Jacobo and Tarjan, Robert E and Lawler, Eugene L},
  booktitle={Proceedings of the eleventh annual ACM symposium on Theory of computing},
  pages={1--12},
  year={1979}
}

\appendix
\newpage \section{Agent allocation algorithm}\label{ap:alg}

The algorithm is a modified version of the Kelso-Crawford algorithm (\cite{kelso1982job}) for computing Walrasian equilibrium in the equivalent economy with augmented trip value functions $\bar{V}$. 

\begin{algorithm}[htp]
\LinesNumbered
\SetAlgoLined
\textbf{Initialize:} Set $u_m \leftarrow 0 ~ \forall m \in M$; $\bl \leftarrow \emptyset, ~\forall l \in L$. Set $\epsilon \in \left(0, 1/2|M|\right)$\; 
\While{TRUE}{
%$J_l \leftarrow \emptyset, ~\forall \l \in \L$\;
\For{$\l$ in $\L$}{
$J_l \leftarrow \emptyset$, $~ \hlbar \leftarrow \emptyset$, $|\hlbar| \leftarrow 0, ~ \etalowl \leftarrow 0, ~ \phil \leftarrow 0$, $\forall l \in L$\; %$(\hlbar', |\hlbar'|, \phil', \etalowl') \leftarrow (\hlbar, |\hlbar|, \phil, \etalowl)$\;
\For{$\hat{m}$ in $\bl$ in decreasing order of $\eta_{m,l}$}{
\eIf{$\eta_{\hat{m}, l} < \left(\xi_l(|\hlbar|+1) - \xi_l(|\hlbar|)\right)d_l$}{break}{$|\hlbar| \leftarrow |\hlbar|+1$, $\hlbar \leftarrow \hlbar\cup \{\hat{m}\}$, $\etalowl \leftarrow \eta_{\hat{m}, l}$}
}
$\phil \leftarrow \sum_{\m \in \hlbar} \eta_{m, l} -\xi_l(|\hlbar|)d_l - \sum_{m \in \bl} u_m$\;
\For{$\hat{j}$ in $S\setminus \bl$ in decreasing order of $\eta_{j,l}-u_j$}{ %\mathrm{sort}(S\setminus \bl, \mathrm{key} = \eta_{j,l}-u_j)$}{
%$\hat{j} \leftarrow \argmax_{j \in S\setminus \(\bl \cup J_l\)} \eta_l^j-u^j - \epsilon$\;
\uIf{$\eta_{\hat{j}, l} \geq \etalowl \geq (\xi_l(|\hlbar|+1) - \xi_l(|\hlbar|)) d_l$}{$|\hlbar'|\leftarrow |\hlbar|+1$, $\hlbar' \leftarrow \hlbar \cup \{\hat{j}\}$, $\phil' \leftarrow \phil + \eta_{\hat{j},l} - (\xi_l(|\hlbar|+1) - \xi_l(|\hlbar|)) d_l - u_{\hat{j}} - \epsilon$, $\etalowl'\leftarrow \etalowl$}
\uElseIf{$\etalowl \geq \eta_{\hat{j}, l} \geq (\xi_l(|\hlbar|+1) - \xi_l(|\hlbar|)) d_l$}{$|\hlbar'|\leftarrow |\hlbar|+1$, $\hlbar' \leftarrow \hlbar \cup \{\hat{j}\}, \etalowl' \leftarrow \eta_{\hat{j},l}$, $\phil' \leftarrow \phil + \eta_{\hat{j},l} - (\xi_l(|\hlbar|+1) - \xi_l(|\hlbar|)) d_l- u_{\hat{j}} - \epsilon$}
\uElseIf{$\eta_{\hat{j},l} \geq \etalowl$ and $(\xi_l(|\hlbar|+1) - \xi_l(|\hlbar|)) d_l \geq \etalowl$}{$\hlbar' \leftarrow \hlbar \cup \{\hat{j}\} \setminus \{\underline{l}\}, \etalowl' \leftarrow \eta_{\hat{j},l}$, $\phil' \leftarrow \phil + \eta_{\hat{j},l} - \etalowl- u_{\hat{j}} - \epsilon$}
\Else{$\phil' \leftarrow \phil- u_{\hat{j}} - \epsilon$}
%\{W_l(j|\bl \cup J_l) - u^j\}$\;
%$(|h'|, h',  W') \leftarrow \text{Algorithm 3}(\bl \cup \{\hat{j}\})$\;
\eIf{$\phi_{l}^{'} \leq \phi_l$}{break}{$(\hlbar, |\hlbar|, \phil, \etalowl) \leftarrow (\hlbar', |\hlbar'|, \phil', \etalowl')$, $J_l \leftarrow J_l \cup \{\hat{j}\}$}}}
%$\Delta_l \leftarrow , ~ \forall l$\;
\eIf{$J_l =\emptyset, ~ \forall l \in L$}{
break}{Arbitrarily pick $\hat{l}$ with $J_{\hat{l}} \neq \emptyset$\;
$\bar{b}_{\hat{l}} \leftarrow \bar{b}_{\hat{l}} \cup J_{\hat{l}}$\;
$\bar{b}_{\hat{l}} \leftarrow \bar{b}_{\hat{l}} \setminus J_{\hat{l}}, ~ \forall l \neq \hat{l}$\;
$u_m \leftarrow u_m+\epsilon, ~ \forall \m \in J_{\hat{l}}$.}
}
\textbf{Return} $\left(\bl\right)_{l \in L}$, $(\hlbar)_{l \in L}$
\caption{Allocating capacity sharing coalitions}
\label{alg:allocation}
\end{algorithm}
For each $r \in \Ropt$, we re-write the augmented trip value function \eqref{eq:Vtilde} with agent coalition $\ball \in \Ball$ and slot $l \in L_r$ (with slight abuse of notation) as follows: 
\begin{align}\label{eq:barV}
    \barV_l(\ball)%&= \sum_{\m \in \tildeb} \(\tripm - \vm \tr - d_m((\start+\tr-q_m)_{+})\) - \sum_{\m \in \tildeb} \(\pi(|\tildeb|)+\gamma(|\tildeb|)\tr\) - \(\fixedcost+ \delta \tr\)|\tildeb| \\
    &=\sum_{\m \in h_l(\ball)} \eta_{m,l} -\xi_l(|h_l(\ball)|), \quad \forall \ball \in \Ball, \quad \forall \l \in \L_r, %\quad \forall z=1, \dots, T-\tr, \quad r \in \Ropt, 
\end{align}
where 
\begin{align*}\eta_{m,l} &\deleq \tripm - \vm\tr, \\ \xi_l(|h_l(\ball)|) &\deleq \changetrip(|h_l(\ball)|)|h_l(\ball)|+ \changevm(|h_l(\ball)|)|h_l(\ball)| \tr,
\end{align*}
and $h_l(\ball) = h_r(\ball)$ is the representative agent coalition in $\ball$ in slot $l \in L_r$ as in \eqref{eq:Vtilde}.

Algorithm \ref{alg:allocation} starts with setting the utility of all agents to be zero, and the set of agents assigned to each slot $l \in L$ to be empty (Line 1). The algorithm keeps track of the following quantities: 
\begin{itemize}
\item[-] $u_m$ is the utility of each agent $m \in M$.
    \item[-] $\bl$ is the augmented agent coalition that is assigned to each slot $l \in L$.
    \item[-] $\hlbar=h_l(\bl)$ is the representative agent coalition given $\bl$ in slot $\l \in L$. $|\hlbar|$ is the size of the representative agent coalition.
    \item[-] $\phil(\ball_l) = \barV_l(\bl) - \sum_{m \in \bl} \um$ is the difference between the augmented trip value function with assigned agent coalition $\bl$ in slot $l$ and the total utility of agents in $
    \bl$. 
    \item[-] $J_l = \arg\max_{J \subseteq \M \setminus \bl} \phil(\ball \cup J) - \phil(\ball)$ is the set of agents, when added to $\bl$, maximally increases the value of $\phil$.
\end{itemize}

In each iteration of Algorithm \ref{alg:allocation}, Lines 3-29 compute the representative agent coalition $\hlbar$, and the set $J_l$ based on the current agent coalition assignment and the utility vector for each $l \in L$. In particular, the representative agent coalition $\hlbar$ is computed by selecting agents from the currently assigned augmented agent coalition $\bl$ in decreasing order of $\eta_{l, m}$ in \eqref{eq:barV}, and the last selected agent $\hat{\m}$ (i.e. the agent in $\hlbar$ with the minimum value of $\eta_{l, m}$) satisfies $\eta_{l, \mhat} \geq \xi_r(|\hlbar|) - \xi_r(|\hlbar|-1)$. That is, adding agent $\hat{\m}$ to the set $\hlbar \setminus \{\hat{\m}\}$ increases the trip value, but adding any other agents decrease the trip value, i.e. $\eta_{l, \mhat}  < \xi_r(|\hlbar|+1) - \xi_r(|\hlbar|)$ for all $\m \in \bl\setminus \hlbar$. The value of $\hlbar$, $|\hlbar|$ records the element and size of the representative agent coalition in the current round, and $\lambda_l$ records the value of $\eta_{l, \hat{m}}$. We also compute the value of the function $\phil(\ball_l) = \barV_l(\bl) - \sum_{m \in \bl} \um$ in Line 12. %Therefore, the set $\hlbar$ is computed by adding $\m \in \bl$ one by one according to decreasing order of $\eta_{l, m}$ until all $\m$ with $\eta_r^m \geq \eta_r^{\hat{\m}}$ are included. 

%  \begin{algorithm}[H]
% \textbf{Initialize:} Set $\tilde{h}\leftarrow 0$, $h_r(\bl) \leftarrow \emptyset$, $\tilde{b}_l \leftarrow \bl$\;
% \While{TRUE}{
% $\hat{m} \leftarrow \argmax_{\m \in\tilde{b}_l} \eta_l^m$\;
% \eIf{$\eta_l^{\hat{m}} < \(\thetafun_r(\tilde{h}+1) - \thetafun_r(\tilde{h})\)t_l$}{break}{$\tilde{h} \leftarrow \tilde{h}+1$, $h_r(\bl) \leftarrow h_r(\bl) \cup \{\hat{m}\}$, $\tilde{b}_l \leftarrow \tilde{b}_l \setminus \{\hat{\m}\}$}
% }
% $W_l(\bl) \leftarrow \sum_{\m \in h_r(\bl)} \eta_l^m -\theta_r(\tilde{h})$\;
% \textbf{Return $\tilde{h}$, $h_r(\bl)$ and $W_l(\bl)$}
% \caption{Computing $h_r(\bl)$ and $W_l(\bl)$}
% \label{alg:representative}
% \end{algorithm}

Furthermore, since the augmented value function $ \barV_l(\ball)$ satisfies monotonicity and gross substitutes conditions, we can compute the set $J_l = \arg\max_{J \subseteq \M \setminus \bl} \phil(\ball \cup J) - \phil(\ball)$ by iteratively adding agents not in $\bl$ into $J_l$ greedily according to their marginal contribution to the value of the function $\phi_l(\bl)$  (\cite{gul1999walrasian}). In this step, we do not re-compute the represented agent coalition or the augmented trip value function every time we add an agent $\hat{j}$ to $\bl$. Instead, we only need to compare $\hat{j}$ with $\lambda_l$ and the marginal change of $\phi_l(|\hlbar|)$ to determine if the representative agent coalition needs to include $\hat{j}$ (Lines 13-29). %We note that the greedy approach can be used to compute $J_l$ if and only if $\barV_l$ satisfies the monotonicity and gross substitute conditions (\cite{gul1999walrasian}). Therefore, these conditions are in fact \emph{necessary} to ensure the polynomial time complexity of the algorithm. We include detailed discussion in Appendix \ref{sec:review}. 
 In Lines 30-38, if there exists at least one slot $\hat{l} \in \L$ with $J_{\hat{l}} \neq \emptyset$, then we choose one such slot $\hat{l}$, and re-assign agents in coalition $J_{\hat{l}}$ to the current assignment $\hat{l}$. We increase the utility $\um$ for the re-assigned agents $\m \in J_{\hat{l}}$ by a small number $\epsilon>0$. As $u_m$ increases for each $m \in M$, the algorithm eventually terminates in finite time when $J_l = \emptyset$ for all $l \in L$. The algorithm returns the representative agent coalition of each slot $\left(\hlbar\right)_{l \in L}$ (Line 39). From Lemma \ref{lemma:kc}, we know that $\left(\hlbar\right)_{l \in L}$ is a Walrasian equilibrium good allocation in the auxiliary economy. Then, following from Lemma \ref{lemma:integer}, the corresponding trip organization vector is given by \eqref{eq:construct}. %such that $ x_r^{z*}(\b)=1$ for all $b \in \{\hlbar\}_{l \in L_r^z}$ and all $z \in 1, \dots, T-\tr$ and all $r \in \{\R|w_r^*>0\}$. The rest of the trips are not organized. 
%  \begin{subequations}\label{eq:opt_final}
% \begin{align}
%     x_r^{z*}(\b)&= \left\{\begin{array}{ll}
%     1, &\quad \text{if $b \in \{\hlbar\}_{l \in L_r^z}$}, \\
%     0, & \quad \text{otherwise}, 
%     \end{array}\right. \quad \forall z \in 1, \dots, T-\tr, \quad \forall r \in \Ropt, \\
%      x_r^{z*}(\b)&= 0, \quad \forall \b \in \B, \quad  \forall z=1, \dots, T-\tr, \quad \forall r \in \R\setminus \Ropt.
% \end{align}
% \end{subequations}

%%novalidate

\section{Review of Combinatorial Auction Theory}\label{sec:review}

Consider an economy with a finite set of indivisible goods $M$ and a finite set of buyers $L$. Each buyer $l \in L$ has a valuation function $\bar{V}_l: \bar{B} \to \mathbb{R}$, where each $\bar{b} \in \bar{B} = 2^{M}$ is a bundle of goods, and $\bar{V}_l(\bar{b})$ is buyer $l$'s valuation of $\bar{b}$. The good allocation vector in this economy is $\bar{x}= (\bar{x}_{l}(\bar{b}))_{l \in L, \bar{b} \in \bar{B}}$, where $\bar{x}_{l}(\bar{b}) =1$ if good bundle $\bar{b}$ is allocated to buyer $l$ and 0 otherwise.

\textbf{Equivalence between coalition formation and good allocation.} Our problem of forming capacity sharing coalitions without the capacity constraint can be equivalently viewed as the good allocation problem in the economy with indivisible goods. In particular, the set of agents $M$ is equivalently viewed as the set of goods $M$. The set of route slots $\L$ is viewed as the set of buyers. Then, the augmented trip value function $\bar{V}_r(\bar{b})$ is equivalent to any buyer $l \in L_{r}$'s valuation of good bundle $\bar{b}$. Each agent $m$'s utility is equivalent to the price of good $m$. 

We next define Walrasian equilibrium of the equivalent economy. 
%The payment vector in this economy is $p=(p_l)_{l \in L}$, where $p_l$ is the payment of buyer $l \in L$. The utility vector is $u=(u_l)_{l \in L}$, where 
% \begin{align*}
% u_l = \sum_{\bar{b} \in \bar{B}} \bar{x}_l(\bar{b}) \bar{V}_{l}(\bar{b}) - p_l, \quad \forall l \in L.
% \end{align*}

\begin{definition}[Walrasian equilibrium \cite{kelso1982job}]\label{def:we}
A tuple $\(\bar{x}^*, \uopt\)$ is a Walrasian equilibrium if 
\begin{itemize}
    \item[(i)] For any $\l \in \L$, $\bl \in \argmax_{\ball \in \Ball} \bar{V}_l(\ball) - \sum_{\m \in \bl} \um$, where $\bl$ is the good bundle that is allocated to $l$ given $\yopt$, i.e. $\yopt_{l}(\bl)=1$
    \item[(ii)] For any good $\m \in \M$ that is not allocated to any buyer, (i.e. $\sum_{\l \in \L} \sum_{\ball \ni \m} \yopt_l(\ball) =0$), the price $\umopt=0$. 
\end{itemize}
\end{definition}

\begin{lemma}[\cite{kelso1982job}]\label{lemma:existence_eq}
If the augmented value function $\bar{V}_r(\ball)$ satisfies the monotonicity and gross substitutes conditions for all $r \in R$, then Walrasian equilibrium exists in the equivalent economy with indivisible goods. \end{lemma}
\begin{lemma}[\cite{gul1999walrasian}]\label{lemma:lattice}
If the value function $\bar{V}$ satisfies the monotonicity and gross substitutes conditions, then the set of Walrasian equilibrium prices $U^{*}$ is a lattice and has a maximum component $\umax = \left(\umaxm\right)_{\m \in \M}$ as in \eqref{eq:umaxm}. 
\end{lemma}

\begin{lemma}[\cite{kelso1982job}]\label{eq:greedy_oracle}
Given any price vector $u$, if the value function $\bar{V}_l$ for any $l \in L$ satisfies the monotonicity and gross substitutes conditions, then $\bl \in \argmax_{\ball \in \Ball} \{\bar{V}_l(\ball) - \sum_{\m \in \ball} \um\}$ can be computed by greedy algorithm.
\end{lemma}

% \begin{algorithm}[htp]
% \SetAlgoLined
% \textbf{Initialize:} Set $\ball_l \leftarrow \emptyset$\;
% \While{True}{
% $\hat{m} \leftarrow \argmax_{m \in M \setminus \bl} \bar{V}_l(\bl \cup \{m\}) - \bar{V}_{l}(\bl) - u_m$\;
% \uIf{$\bar{V}_l(\bl \cup \{\hat{m}\}) - \bar{V}_{l}(\bl) - u_{\hat{m}}>0$}{$\bl \leftarrow \bl \cup \{\hat{m}\}$\;}
% \Else{break, return $\bl$}
% }
% \caption{Greedy algorithm for computing $\bar{b}_l$}
% \label{alg:greedy_demand}
% \end{algorithm} 

\begin{lemma}[\cite{kelso1982job}]\label{lemma:kc}
For any $\epsilon < \frac{1}{2|M|}$, if the value function $\bar{V}_l$ satisfies the monotonicity and gross substitutes conditions for all $l \in L$, then $(\ball_l)_{l \in L}$ computed by Algorithm \ref{alg:allocation} is a Walrasian equilibrium good allocation. 
\end{lemma}

%  \begin{algorithm}[htp]
% % \SetAlgoLined
% \textbf{Initialize:} Set $u_m \leftarrow 0 ~ \forall m \in M$; $\bl \leftarrow \emptyset, ~ \forall l \in L$\;
% \While{TRUE}{
% %$J_l \leftarrow \emptyset, ~\forall \l \in \L$\;
% \For{$\l \in \L$}{$J_l \leftarrow \arg\max_{J \subseteq \M\setminus \bl} \left\{\bar{V}_l(J \cup \bl) - \sum_{m \in \bl} u_m-  \sum_{m \in J} \(u_m+\epsilon\)\right\}$}
% %$\Delta_l \leftarrow , ~ \forall l$\;
% \eIf{$J_l =\emptyset, ~ \forall l \in L$}{
% break}{Arbitrarily pick $\hat{l}$ with $J_{\hat{l}} \neq \emptyset$\;
% $\bar{b}_{\hat{l}} \leftarrow \bar{b}_{\hat{l}} \cup J_{\hat{l}}$\;
% $\bar{b}_{\hat{l}} \leftarrow \bar{b}_{\hat{l}} \setminus J_{\hat{l}}, ~ \forall l \neq \hat{l}$\;
% $u_m \leftarrow u_m+\epsilon, ~ \forall \m \in J_{\hat{l}}$.}
% }
% \textbf{Return $\(\bl\)_{l \in L}$}
% \caption{Kelso-Crawford Auction \cite{kelso1982job}}
% \label{alg:KC}
% \end{algorithm}

\section{Proof of Proposition \ref{prop:primal_dual}}\label{apx:proof_A}
First, we prove that the four conditions of market equilibrium $\(\xopt, \popt, \tollopt\)$ ensure that $\xopt$ satisfies the feasibility constraints of the primal \eqref{eq:LP1bar}, $\(\uopt, \tollopt\)$ satisfies the constraints of the dual \eqref{eq:D1bar}, and $\(\xopt, \uopt, \tollopt\)$ satisfies the complementary slackness conditions. Here, the vector $\uopt$ is the utility vector computed from \eqref{eq:u_p}.  
\begin{enumerate}
\item[(i)] Feasibility constraints of \eqref{eq:LP1bar}: Since $\xbaropt$ is a feasible trip vector, $\xbaropt$ must satisfy the feasibility constraints of \eqref{eq:LP1bar}. 

\item[(ii)] Feasibility constraints of \eqref{eq:D1bar}: From the stability condition \eqref{eq:stability}, individual rationality \eqref{eq:ir}, and the fact that edge prices are non-negative, we know that $\(\uopt, \tollopt\)$ satisfies the feasibility constraints of \eqref{eq:D1bar}. 

\item[(iii)] Complementary slackness condition with respect to  \eqref{subeq:LP11}: If agent $\m$ is not assigned, then \eqref{subeq:LP11} is slack with the integer trip assignment $\xbaropt$ for some agent $\m$. The budget balanced condition \eqref{subeq:p_not_assigned} shows that $\popt_m=0$. Since agent $\m$ is not in any trip and the payment is zero, the dual variable (i.e. agent $\m$'s utility) $\umopt=0$. On the other hand, if $\umopt>0$, then agent $\m$ must be in a trip, and constraint \eqref{subeq:LP11} must be tight. Thus, we can conclude that the complementary slackness condition with respect to the primal constraint \eqref{subeq:LP11} is satisfied. 

\item[(iv)] Complementary slackness condition with respect to \eqref{subeq:LP12}: Since the mechanism is market clearing, edge price $\tollet$ is nonzero if and only if the load that enters edge $\e$ at time $t$ is below the capacity, i.e. the primal constraint \eqref{subeq:LP12} is slack for edge $\e\in \E$ and $t$. Therefore, the complementary slackness condition with respect to the primal constraint \eqref{subeq:LP12} is satisfied. 
\item[(v)] Complementary slackness condition with respect to \eqref{subeq:D11}: From \eqref{subeq:bb}, we know that for any organized trip, the corresponding dual constraint \eqref{subeq:D11} is tight. If constraint \eqref{subeq:D11} is slack for a trip $\(\bbar, \r\)$, then the budget balance constraint ensures that trip is not organized. Therefore, the complementary slackness condition with respect to the primal constraint \eqref{subeq:D11} is satisfied. 
\end{enumerate}

We can analogously show that the inverse of (i) -- (v) are also true: the feasibility constraints of \eqref{eq:LP1bar} and \eqref{eq:D1bar}, and the complementary slackness conditions ensure that $\(\xopt, \popt, \tollopt\)$ is a market equilibrium. Thus, we can conclude that $\(\xopt, \popt, \tollopt\)$ is a market equilibrium if and only if $\(\xopt, \uopt, \tollopt\)$ satisfies the feasibility constraints of \eqref{eq:LP1bar} and \eqref{eq:D1bar}, and the complementary slackness conditions.

From strong duality theory, we know that the equilibrium trip vector $\xbaropt$ must be an integer optimal solution of \eqref{eq:LP1bar}. Therefore, the existence of market equilibrium is equivalent to the existence of an integer optimal solution of \eqref{eq:LP1bar}. The optimal trip assignment is an integer optimal solution of \eqref{eq:LP1bar}, and $\left(\uopt, \tollopt\right)$ is an optimal solution of the dual problem \eqref{eq:D1bar}. The payment $\popt$ can be computed from \eqref{eq:u_p}. \QEDA

\section{Proof of Statements in Section \ref{subsec:sufficient}}\label{apx:proof_B}

%%novalidate
% \section{Proof of Section \ref{sec:tractable_pooling}.}\label{apx:proof_B}
% \vspace{0.3cm}

\begin{lemma}[\cite{bein1985minimum}]\label{lemma:max_flow}
On series-parallel networks, the flow $\kopt$ maximizes the total flow. That is, for any $x$ that satisfies \eqref{subeq:LP11} -- \eqref{subeq:LP13}, we have: 
\begin{align*}
\sum_{r \in \R} \sum_{b \in B} x_r(b) \leq \sum_{r \in \R} \kopt_r.
\end{align*}
\end{lemma}

\noindent\emph{Proof of Lemma \ref{lemma:FF}.}
Consider any (fractional) optimal solution of \eqref{eq:LP1bar}, denoted as $\xhat$. We denote $\fhat(\b) = \sum_{\r \in \R} \xhat_r(\b)$ as the flow of coalition $\b$, and $\widehat{F}= \sum_{\b \in \B} \fhat(\b)$ as the total flow. Since $\xhat$ is feasible, we know that $\widehat{F} \leq C$, where $C$ is the maximum capacity of the network. We re-write the trip valuation as follows: 
\begin{align*}
V_r(\b)= z(\b)- g(\b) \tr, \quad \forall \(\b, \r\) \in \B \times \R, \end{align*}
where $g(\b) = \sum_{\m \in \b} \vm + \sum_{\m \in \b} \changevm_m(|\b|)$, and $z(\b)= \sum_{\m \in \b} \tripm - \sum_{m \in b} \changetrip_m(|b|)$.

The set of all coalitions with positive flow in $\xhat$ is $\widehat{\B} \deleq \{b \in \B|\fhat(\b)>0\}$. We denote the number of coalitions in $\Bhat$ as $n$, and re-number these coalitions in decreasing order of $g(\b)$, i.e. 
\begin{align}\label{eq:g_ordering}
g(\b_1) \geq g(\b_2) \geq \cdots \geq g(\b_n).
\end{align}

We now construct another trip  vector $\xopt$ by the following assignment procedure: \\
\emph{Initialization:} Set route set $\tildeR = \Ropt$, initial zero assignment vector $\xrbopt \leftarrow 0$ for all $\r \in \R$ and all $\b \in \B$, and residual route capacity $\tildeq_r = \kopt_{\r}$ for all $\r \in \tildeR$.\\\
\emph{For $j=1, \dots, n$:}
\begin{itemize}
    \item[(i)] Assign coalition $\b_j$ to a route 
    $\rhat$ in $\tildeR$, which has the minimum travel time among all routes with flow less than the capacity, i.e. $\rhat \in \argmin_{\r \in \tildeR}\{\tr\}$. 
    \item[(ii)] If $\fhat(\b_j) \leq \tildeq_{\hat{r}}$, then $x^*_{\hat{r}}(\b_j) = \fhat(\b_j)$. We update $\tilde{q}_{\hat{r}} \leftarrow \tilde{q}_{\hat{r}} - \fhat(\b_j)$. 
    \item[(iii)] Otherwise, assign $\xopt_{\hat{r}}(\b_j)= \tildeq_{\hat{r}}$, set $\tilde{q}_{\hat{r}} \leftarrow 0$, $\tilde{R} \leftarrow \tilde{R} \setminus \{\hat{r}\}$, and continue to assign the remaining weight $\hat{f}(\b_j)\leftarrow \hat{f}(\b_j)- \xopt_{\hat{r}}(\b_j)$ to the next unsaturated route with the minimum cost. Repeat this process until the condition in (ii) is satisfied, i.e. the weight of $\bhat_j$ is assigned. 
\end{itemize}
From Lemma \ref{lemma:max_flow}, we know that $\hat{F} = \sum_{r \in R} \sum_{b \in B} \hat{x}_r(b) \leq \sum_{r \in R} k_r^*$ so that the algorithm terminates with the flow of all $\hat{B}$ being allocated. 
We can check that $\sum_{\b \ni \m} \sum_{r \in \R}  \xrbopt = \sum_{\b \ni \m} \fhat(\b) \leq 1$ so that \eqref{subeq:LP2k1} is satisfied. Additionally, since in the assignment procedure, the total weight assigned to route $\r$ is less than or equal to $\kopt_r$, we must have $\sum_{\b \in \B} \xrbopt \leq \kopt_r$ for all $\r \in \R$, i.e. \eqref{subeq:LP2k2} is satisfied. Thus, $\xopt$ is a feasible solution of \eqref{eq:LP2k}. 

It remains to prove that $\xopt$ is optimal for \eqref{eq:LP2k}. We prove this by showing that $S(\xopt) \geq S(\xhat)$. The objective function $S(\xopt)$ can be written as follows:
\begin{align}\label{eq:decompose_V}
S(x^*)=&\sum_{r \in \R} \sum_{\b \in \B} \Vrb \xrbopt = \sum_{r \in \R} \sum_{\b \in \B} z(\b) \xrbopt - \sum_{r \in \R} \sum_{\b \in \B} g(\b) \tr \xrbopt. \end{align}
Since the assignment procedure terminates with all coalitions in $\xhat$ being assigned, $\sum_{\r \in \R} \xrbopt=\fhat(\b) = \sum_{\r \in \R} \xhat_r(\b)$ for all $\b \in \B$. Therefore, 
\begin{align}\label{eq:equal_part}
&\sum_{r \in \R} \sum_{\b \in \B} z(\b) \xrbopt = \sum_{\b \in \B} z(\b) \fhat(\b) = \sum_{r \in \R} \sum_{\b \in \B} z(\b) \xhat_r(\b).
\end{align}
Thus, to prove $S(\xopt) \geq S(\xhat)$, it remains to show that 
\begin{align}\label{eq:final_argument}
    \sum_{r \in \R}  \sum_{\b \in \B} g(\b)\tr \xrbopt \leq \sum_{r \in \R}  \sum_{\b \in \B} g(\b)\tr \xhat_r(b).
\end{align}
{This is equivalent to proving that the following claim holds: % on any series-parallel network $G$ and  prove that \eqref{eq:final_argument} holds on any series parallel network, it is sufficient to prove the following claim: 
\begin{claim}\label{claim}
    Suppose that the network $G$ is series-parallel. For any $\xhat$ and coalition flow vector $\hat{f}$, we construct the trip allocation vector $\xopt$ following procedure (i) -- (iii). Then, the trip allocation vector $\xopt$ minimizes $\sum_{r \in \R} \sum_{\b \in \B} g(\b) \tr \xrb$ among all feasible $\x$ that induces the same flow of coalitions $\fhat$, i.e. 
\begin{align}\label{eq:induction}
\xopt \in \argmin_{\x \in \X(\fhat)}\sum_{r \in \R}  \sum_{\b \in \B} g(\b)\tr \xrb, 
\end{align}
where 
\begin{align}\label{eq:X}
\X(\fhat) \deleq \left\{\(\xrb\)_{\r \in \R, \b \in \B}\left\vert
\begin{array}{l}
\sum_{\r \in \R} \xrb= \fhat(\b), \quad \forall \b \in \B, \\
\sum_{\b \in \B}\sum_{\r \ni \e} \xrb \leq \qe, \quad \forall \e \in \E, \\
\xrb \geq 0, \quad \forall \r \in \R, \quad \forall \b \in \B 
\end{array}
\right.
\right\}.
\end{align}
\end{claim}

If Claim \ref{claim} holds, then for any $\xhat$ that is a fractional optimal solution of \eqref{eq:LP1bar}, we can compute the aggregate coalition flow vector $\fhat$ and construct $x^*$ following the assignment procedure. The constructed $x^*$ satisfies all the constraints in \eqref{eq:LP2k}, and satisfies $S(\xhat) \leq S(x^*)$. This implies that the optimal value of \eqref{eq:LP2k} is higher than that of \eqref{eq:LP1bar}. Since \eqref{eq:LP1bar} and \eqref{eq:LP2k} have the same objective function and any feasible solution of \eqref{eq:LP2k} is also a feasible solution of \eqref{eq:LP1bar}, we can conclude that any optimal solution of \eqref{eq:LP2k} is also an optimal solution of \eqref{eq:LP1bar}. 

The rest of the proof focuses on showing that Claim \ref{claim} holds by mathematical induction. If $G$ is  a single-edge network, then Claim \ref{claim} holds trivially since the assignment procedure does not change the trip allocation. We next prove that if Claim \ref{claim} holds for any two series-parallel sub-networks $\Gp$ and $\Gpp$, then Claim \ref{claim} holds for the network $\G$ that connects $\Gp$ and $\Gpp$ in series or in parallel. In particular, we analyze the cases of series connection and parallel connection separately.} %  If network $G$ has a single edge, then \eqref{eq:induction} holds trivially. Since any series-parallel network is formed by connecting single edges in series or in parallel for finite times, we prove  on series-parallel network $G$ by mathematical induction, and 
%{\color{integer optimal solution} The crux of the argument, for both cases, is the following. We first compute an allocation $x^*$ using procedure (i)-(iii) on $G$. The optimality of $x^*$ is argued by considering the restriction of $x^*$ in the individual subnetworks $\Gp, \Gpp$, and showing that this is precisely the trip vector that one would obtain by independently running procedure (i) -- (iii) on $\Gp$ and $\Gpp$ respectively (and therefore, by induction, $x^*$ must be optimal). We do so via a contradiction argument: if the restriction of $x^*$ in $\Gp$ differs from the output of procedure (i) - (iii) (say $x^1$), there is at least one coalition and two distinct routes where it differs. We show that carefully moving a fraction of this coalition flow from one route to another either contradicts the optimality of $x^*$, or contradicts the assumption that $x^1$ was returned by procedure (i) -- (iii).  }

\vspace{0.2cm}
\noindent\emph{(Case 1)} Series-parallel network $G$ is formed by connecting two series-parallel sub-networks $\Gp$ and $\Gpp$ in series. 

    We denote the sets of routes in sub-network $\Gp$ and $\Gpp$ as $\Rp$ and $\Rpp$, respectively. Since $\Gp$ and $\Gpp$ are connected in series, the set of routes in network $\G$ is $\R \deleq \Rp \times \Rpp$. We denote a route $r=r_1r_2 \in R$ as a route composed of $r_1 \in R_1$ and $r_2 \in R_2$, connected in series. For any flow vector $\fhat$, the set of trip vectors on $G$ that {satisfy the constraint in \eqref{eq:X}} is $\X(\fhat)$, and the trip  vector obtained from the above-mentioned assignment procedure based on $\fhat$ is $\xopt$.  
Since the two sub-networks are connected in series, the coalition flow of each $b \in B$ is $\fhat(b)$ in both $\Gp$ and $\Gpp$. {We denote a trip vector on sub-network $G^1$ (resp. $G^2$) as $x^1$ (resp. $x^2$). Analogously to \eqref{eq:X}, we define $\Xp(\fhat)$ (resp. $\Xpp(\fhat)$) as the set of trip vectors that induce the coalition flow vector $\hat{f}$ on sub-network $\Gp$ (resp. $\Gpp$)
    \begin{align*}
        \Xp(\fhat)&\deleq \left\{\(\xp_{\rp}(b)\)_{\rp \in \Rp, \b \in \B}\left\vert
\begin{array}{l}
\sum_{\rp \in \Rp} \xp_{\rp}(b)= \fhat(\b), \quad \forall \b \in \B, \\
\sum_{\b \in \B}\sum_{\rp \ni \e} \xp_{\rp}(b) \leq \qe, \quad \forall \e \in \E^1, \\
\xp_{\rp}(b) \geq 0, \quad \forall \rp \in \Rp, \quad \forall \b \in \B 
\end{array}
\right.
\right\}, \\
\Xpp(\fhat)& \deleq \left\{\(\xpp_{\rpp}(b)\)_{\rpp \in \Rpp, \b \in \B}\left\vert
\begin{array}{l}
\sum_{\r \in \R} \xpp_{\rpp}(b)= \fhat(\b), \quad \forall \b \in \B, \\
\sum_{\b \in \B}\sum_{\rpp \ni \e} \xpp_{\rpp}(b) \leq \qe, \quad \forall \e \in \E^2, \\
\xpp_{\rpp}(b) \geq 0, \quad \forall \rpp \in \Rpp, \quad \forall \b \in \B 
\end{array}
\right.
\right\},
    \end{align*}
where $E^1$ and $E^2$ are the edge sets of $\Gp$ and $\Gpp$, respectively. Since the two sub-networks are connected in series, for any $\x \in \X(\fhat)$, we can find $\xp = (\xp_{\rp}(\b))_{\rp \in \Rp, b \in B} \in \Xp(\fhat)$ (resp. $\xpp = (\xpp_{\rpp}(\b))_{\rpp \in \Rpp, b \in B} \in \Xpp(\fhat)$) such that $\xp_{\rp}(\b) = \sum_{\rpp \in \Rpp}\x_{\rp\rpp}(\b)$ (resp. $\xpp_{\rpp}(\b)= \sum_{\rp \in \Rp}\x_{\rp\rpp}(\b)$) for all $\b \in \B$ and all $\rp \in \Rp$ (resp. $\rpp \in \Rpp$). Therefore,} 
    \begin{align}
        &\sum_{r \in \R}  \sum_{\b \in \B} g(\b)\tr \xrb = \sum_{\rp \in \Rp} \sum_{\rpp \in \Rpp} \sum_{\b \in \B} g(\b)(d_{\rp}+ d_{\rpp}) x_{\rp\rpp}(b) \notag \\
        =& \sum_{\rp \in \Rp} \sum_{\b \in \B} g(\b)d_{\rp} \(\sum_{\rpp \in \Rpp} \x_{\rp\rpp}(\b)\)+ \sum_{\rpp \in \Rpp} \sum_{\b \in \B} g(\b)d_{\rpp} \(\sum_{\rp \in \Rp} \x_{\rp\rpp}(\b)\)\notag \\
        =&\sum_{\rp \in \Rp} \sum_{\b \in \B} g(\b)d_{\rp} \xp_{\rp}(\b)+ \sum_{\rpp \in \Rpp} \sum_{\b \in \B} g(\b)d_{\rpp} \xpp_{\rpp}(\b).\label{eq:series_1}
        \end{align}

We construct $\xoptp$ (resp. $\xoptpp$) as the trip vector derived from the assignment procedure given $\fhat$ on $\Gp$ (resp. $\Gpp$). We now argue that $\sum_{\rpp \in \Rpp}\xopt_{\rp\rpp}(\b)=\xoptp_{\rp}(\b)$ for all $\b \in \B$ and all $\rp \in \Rp$. For the sake of contradiction, assume that there exists $\b \in \B$ such that $\sum_{\rpp \in \Rpp}\xopt_{\rp\rpp}(\b) \neq \xoptp_{\rp}(\b)$ for at least one $\rp \in \Rp$. We denote $\bhat$ as one such coalition with the maximum $g(\bhat)$. Since the total flow of $\bhat$ is $\fhat(\bhat)$ in both $\xopt$ and $\xoptp$, if $\sum_{\rpp \in \Rpp}\xopt_{\rp\rpp}(\bhat) \neq \xoptp_{\rp}(\bhat)$ on one $\rp \in \Rp$, the same inequality must hold for another $\rptwo \in \Rp$. Without loss of generality, we assume that $d_{\rp}< d_{\rptwo}$. {Since any coalition $\b$ assigned before} $\bhat$ ($\{b \in B | g(\b)< g(\bhat)\}$) satisfies $\sum_{\rpp \in \Rpp}\xopt_{\rp\rpp}(\b)=\xoptp_{\rp}(\b)$ for all $\rp \in \Rp$, we know that the residual route capacities $\tildeq$ in the round of assigning $\bhat$ in procedure (i) -- (iii) satisfy $\sum_{\rpp \in \Rpp}\tildeq_{\rp\rpp}=\tildeq_{\rp}$ for all $\rp \in \Rp$. Therefore, if $\sum_{\rpp \in \Rpp}\xopt_{\rp\rpp}(\bhat)>\xoptp_{\rp}(\bhat)$, then $\xoptp$ is not obtained by procedure (i) -- (iii) on $\Gp$ because $\rp$ is the route with the minimum time cost in the round of assigning flow of $\hat{b}$ but $\xoptp$ does not assign as much flow of $\bhat$ as possible to $\rp$, and more flow of $\bhat$ should be moved from the longer route $\rptwo$ to the shorter $\rp$. We can analogously argue that if $\sum_{\rpp \in \Rpp}\xopt_{\rp\rpp}(\bhat)<\xoptp_{\rp}(\bhat)$, then $\xopt$ is not obtained from the algorithm on network $G$. In either case, we have arrived at a contradiction. We can analogously argue that $\sum_{\rp \in \Rp}\xopt_{\rp\rpp}(\b)=\xoptpp_{\rpp}(\b)$ for all $\b \in \B$ and all $\rpp \in \Rpp$. Therefore, \begin{align}
        &\sum_{r \in \R}  \sum_{\b \in \B} g(\b)\tr \xrbopt\notag \\
        =& \sum_{\rp \in \Rp} \sum_{\b \in \B} g(\b)d_{\rp} \(\sum_{\rpp \in \Rpp} \xopt_{\rp\rpp}(\b)\)+ \sum_{\rpp \in \Rpp} \sum_{\b \in \B} g(\b)d_{\rpp} \(\sum_{\rp \in \Rp} \xopt_{\rp\rpp}(\b)\)\notag\\
        =&\sum_{\rp \in \Rp} \sum_{\b \in \B} g(\b)d_{\rp} \xoptp_{\rp}(\b)+ \sum_{\rpp \in \Rpp} \sum_{\b \in \B} g(\b)d_{\rpp} \xoptpp_{\rpp}(\b).\label{eq:series_2}
        \end{align}
        
Since Claim \ref{claim} holds on both sub-networks $\Gp$ and $\Gpp$ for any flow vector $\hat{f}$, we have
\[\xoptp \in \argmin_{\x \in \Xp(\fhat)} \sum_{\rp \in \Rp}  \sum_{\b \in \B} g(\b)\trp \xp_{\rp}(\b), \quad \xoptpp \in \argmin_{\x \in \Xpp(\fhat)}\sum_{\rpp \in \Rpp}  \sum_{\b \in \B} g(\b)\trpp \xpp_{\rpp}(\b).\]
From \eqref{eq:series_1} -- \eqref{eq:series_2}, we know that Claim \ref{claim} also holds on network $\G$. 
\vspace{0.2cm}
    
\noindent\emph{(Case 2)} Series-parallel network $G$ is formed by connecting two series-parallel networks $G_1$ and $G_2$ in parallel.

Same as case 1, we denote $\Rp$ (resp. $\Rpp$) as the set of routes in $\Gp$ (resp. $\Gpp$). Then, the set of all routes in $\G$ is $\R = \Rp \cup \Rpp$. Given $\fhat$, we construct $\xopt$ from the procedure (i) -- (iii) on the entire network $\G$. We define $f^{1*}(b) = \sum_{\rp \in \Rp}  \xopt_r(b)$ and $f^{2*}(b) = \sum_{\rpp \in \Rpp}  \xopt_r(b)$, which are the aggregate flow of each $b \in B$ induced by $x^*$ on sub-network $\Gp$ and $\Gpp$, respectively. The tuple $(f^{1*}, f^{2*})$ governs how the flow $\hat{f}$ splits between the two sub-networks for each $b \in B$. We can verify that the trip vector induced by the assignment procedure given $f^{1*}$ (resp. $f^{2*}$) on sub-network $\Gp$ (resp. $\Gpp$) is $(x_{\rp}^*)_{\rp \in \Rpp}$ (resp. $(x^*_{\rpp})_{\rpp \in \Rpp}$). % is   splits the flow $\fhat$ across the two networks is one flow split induced by $x^*$, and $x^*$ is the trip allocation associated with the flow split $(f^{1*}, f^{2*})$.

% Given any $\fhat$, we compute $\xopt$ from the procedure (i) -- (iii) in network $\G$. We denote $f^{1*}(b) = \sum_{\rp \in \Rp}  \xopt_r(b)$ (resp. $f^{2*}(b) = \sum_{\rpp \in \Rpp}  \xopt_r(b)$) as the flow of each coalition $b\in B$ assigned to sub-network $\Gp$ (resp. $\Gpp$) given $\xopt$. Given the flow vector $f^{1*}$ (resp. $f^{2*}$) on sub-network $\Gp$ (resp. $\Gpp$), we follow the procedure (i) -- (iii) to compute the trip vector $\xoptp$ (resp. $\xoptpp$). Then, $\xoptp = \(\xopt_{\rp}(\b)\)_{\rp \in \Rp, \b \in \B}$, and $\xoptpp = \(\xopt_{\rpp}(\b)\)_{\rpp \in \Rpp, \b \in \B}$.  

Consider any arbitrary split of flow $\fhat$ between the two sub-networks, denoted as $\(\fhatp, \fhatpp\)$, such that $\fhatp(\b)+\fhatpp(\b)=\fhat(\b)$ for all $\b \in \B$. We define the set of feasible trip  vectors on sub-network $\Gp$ (resp. $\Gpp$) that induce the total flow $\fhatp$ (resp. $\fhatpp$) given by \eqref{eq:X} as $\Xp(\fhatp)$ (resp. $\Xpp(\fhatpp)$), i.e.
{\begin{align*}
        \Xp(\fhat)&\deleq \left\{\(\xp_{\rp}(b)\)_{\rp \in \Rp, \b \in \B}\left\vert
\begin{array}{l}
\sum_{\rp \in \Rp} \xp_{\rp}(b)= \fhat^1(\b), \quad \forall \b \in \B, \\
\sum_{\b \in \B}\sum_{\rp \ni \e} \xp_{\rp}(b) \leq \qe, \quad \forall \e \in \E^1, \\
\xp_{\rp}(b) \geq 0, \quad \forall \rp \in \Rp, \quad \forall \b \in \B 
\end{array}
\right.
\right\}, \\
\Xpp(\fhat)& \deleq \left\{\(\xpp_{\rpp}(b)\)_{\rpp \in \Rpp, \b \in \B}\left\vert
\begin{array}{l}
\sum_{\r \in \R} \xpp_{\rpp}(b)= \fhat^2(\b), \quad \forall \b \in \B, \\
\sum_{\b \in \B}\sum_{\rpp \ni \e} \xpp_{\rpp}(b) \leq \qe, \quad \forall \e \in \E^2, \\
\xpp_{\rpp}(b) \geq 0, \quad \forall \rpp \in \Rpp, \quad \forall \b \in \B 
\end{array}
\right.
\right\},
\end{align*}}where $E^1$ and $E^2$ are the edge sets of sub-network $\Gp$ and $\Gpp$, respectively. 
%Then, the set of all trip  vectors that induce $\fhat$ on network $G$ is $\X(\fhat) = \cup_{(\fhatp, \fhatpp): \fhatp+\fhatpp=\fhat } (\Xp(\fhatp), \Xpp(\fhatpp))$. 
For any flow split $\(\fhatp, \fhatpp\)$, we denote the trip  vector obtained by applying the assignment procedure (i) -- (iii) with $\fhatp$ (resp. $\fhatpp$) on sub-network $\Gp$ (resp. $\Gpp$) as $\xopthatp$ (resp. $\xopthatpp$). Since Claim \ref{claim} holds for sub-network $\Gp$ (resp. $\Gpp$) with any flow $\fhatp$ (resp. $\fhatpp$), we must have
\begin{align*}
    &\sum_{\rp \in \Rp}  \sum_{\b \in \B} g(\b)\tr \xopthatp_{\rp}(\b) + \sum_{\rpp \in \Rpp}  \sum_{\b \in \B} g(\b)\tr \xopthatpp_{\rpp}(\b) \\
    \leq &\sum_{\rp \in \Rp}  \sum_{\b \in \B} g(\b)\tr \xhatp_{\rp}(\b) +\sum_{\rpp \in \Rpp}  \sum_{\b \in \B} g(\b)\tr \xhatpp_{\rpp}(\b), \quad \forall \xhatp \in \X(\fhatp), \xhatpp \in \X(\fhatpp).
\end{align*}
{Therefore, the optimal solution of \eqref{eq:induction} for the entire network $G$ and flow vector $\fhat$ must be a trip  vector $\(\xopthatp, \xopthatpp\)$ that is constructed following assignment procedure (i) -- (iii) on each sub-network given a flow split $\(\fhatp, \fhatpp\)$. We can thus restrict our attention to finding the optimal flow split such that the trip allocation vector induced by the assignment procedure minimizes the objective function value. Recall that $x^*$ is the trip vector induced by the flow split $(f^{1*}, f^{2*})$. To prove that $x^*$ is the optimal solution (i.e. Claim \ref{claim} holds for network $\G$), it remains to show that the flow split $(f^{1*}, f^{2*})$ is the optimal flow split.

Before proceeding with the proof, we argue that in situations where there are multiple coalitions $b \in \hat{B}$ with the same $g(b)$ value (i.e. some inequalities in \eqref{eq:g_ordering} are equalities), these coalitions can be considered equivalent, and their flows can be combined. We denote a partition of $\hat{B}$ as $\hat{B}_1, \dots, \hat{B}_n$, where $g(b) = g(b')$ for any $b, b'$ in each $\hat{B}_{j}$, $j =1, \dots, n$, and $g(b) > g(b')$  for any $b \in \hat{B}_j$ and $b' \in \hat{B}_{j+1}$ for all $j=1, \dots, n-1$. We note that if $(\hat{f}^1, \hat{f}^2) \neq (f^{1*}, f^{2*})$ but 
\begin{align}\label{eq:f_equivalent}
    \sum_{b \in \hat{B}_j}\hat{f}^{1}(b) = \sum_{b \in \hat{B}_j}f^{1*}(b), \quad \sum_{b \in \hat{B}_j}\hat{f}^{2}(b) = \sum_{b \in \hat{B}_j}f^{2*}(b), \quad \forall j=1, \dots, n,
\end{align}
then the associated trip vectors $\hat{x}^*$ and $x^*$ derived from the assignment procedure satisfy: 
\begin{align*}
    \sum_{b \in \hat{B}_j} \xhat^*_r(b) = \sum_{b \in \hat{B}_j} x^*_r(b), \quad \forall r \in R,
\end{align*} 
because coalitions $\bhat$ in the same subset $\hat{B}_j$ are assigned in consecutive order. Moreover, $\hat{x}^*$ and $x^*$ incur the same cost: 
\begin{align*}
\sum_{\rp \in \Rp} \sum_{b \in B} d_{\rp} g(b) \hat{x}^{1*}_{\rp}(b) + \sum_{\rpp \in \Rpp} \sum_{b \in B} d_{\rpp} g(b) \hat{x}^{2*}_{\rpp}(b) = \sum_{r \in R} \sum_{b \in B} g(b) d_r x^*(b).
\end{align*}
This implies that interchanging the route assignment of coalitions $b \in \hat{B_j}$ with the same $g(b)$ does not change the objective function value, and the objective function value only depends on $(x_r(\hat{B}_j))_{r \in R, j=1, \dots, n}$, where $x_r(\hat{B}_j)=\sum_{b\in \hat{B}_j}x_r(b)$ is the aggregate assignment of coalitions in each $\hat{B}_j$ to each route $r \in R$. Therefore, for any flow vector $\hat{f}$ and any network $G$, we can view coalitions in each $\hat{B}_j$ as a single coalition, and combine their flows in the trip assignment process. Then, any trip assignment that leads to the aggregate assignment $(x_r(\hat{B}_j))_{r \in R, j=1, \dots, n}$ will have the same objective value. For the rest of the proof, without loss of generality, we view coalitions with the same $g(b)$ value as the same coalition, and we assume that the ordering of $g(b)$ in \eqref{eq:g_ordering} has no tie. 
%\item In each sub-network, if there are multiple routes with the same time cost, then we can equivalently view those routes as a single route and combine their capacties as the capacity of the single route. Such route combination does not affect the assignment procedure since interchanging the trip assignment across routes with the same length does not affect the objective function value. For the rest of the proof, we assume without loss of generality that routes in each sub-network have different time costs. 

}

We now proceed with the proof. For any $\(\fhatp, \fhatpp\) \neq \(f^{1*}, f^{2*}\)$, we can find a coalition $\bj$ such that $\fhatp(\bj) \neq f^{1*}(\bj)$ (henceforth $\fhatpp(\bj) \neq f^{2*}(\bj)$). We denote $\bjhat$ as one such coalition with the maximum $g(\b)$, i.e. $\fhatp(\bj)= f^{1*}(\bj)$ for any $j =1, \dots, \jhat-1$. Since coalitions $\b_1, \dots, \b_{\jhat-1}$ are assigned before coalition $\bjhat$ according to procedure (i) -- (iii), we know that $\xopthatp_{\rp}(\bj)=\xoptp_{\rp}(\bj)$ and $\xopthatpp_{\rpp}(\bj)=\xoptpp_{\rpp}(\bj)$ for all $\rp \in \Rp$, all $\rpp \in \Rpp$ and all $j=1, \dots, \jhat-1$. 

Since $\fhatp(\bjhat) \neq f^{1*}(\bjhat)$, the trip vector associated with $\bjhat$ in $\xopthatp$ and $\xopthatpp$ must be different from that in $\xoptp$ and $\xoptpp$. Without loss of generality, we assume that $\fhatp(\bjhat)> f^{1*}(\bjhat)$ and $\fhatpp(\bjhat)< f^{2*}(\bjhat)$. We note the following three facts: 
\begin{itemize}
    \item[(1)] The residual capacity after assignment of $\b_1, \dots, \b_{\jhat-1}$  is the same for the flow splits  $(\hat{f}^1, \hat{f}^2)$ and $(f^{1*}, f^{2*})$. 
    \item[(2)] When constructing $\xopt$, the assignment procedure (i)–(iii) assigns as much flow of coalition $\bjhat$ as possible to routes with the minimum travel time cost among all routes in $\R \cup \Rpp$ that remain unsaturated after assigning coalitions $b_1, \dots, b_{\jhat-1}$. 
    \item[(3)] When constructing $\xopthatp$ (resp. $\xopthatpp$), the assignment procedure (i)–(iii) assigns as much flow of coalition $\bjhat$ as possible to routes with the minimum travel time cost among all routes in $\Rp$  (resp. $\Rpp$) that remain unsaturated after assigning coalitions $b_1, \dots, b_{\jhat-1}$. 
\end{itemize}
Due to these three facts, the route set in $\Rpp$ where $b_{\hat{j}}$ is assigned to given $f^{2*}$ is a superset of the one with $\hat{f}^2$ since $\fhatpp(\bjhat)< f^{2*}(\bjhat)$. Similarly, the route set in $\Rp$ where $b_{\hat{j}}$ is assigned to given $f^{1*}$ is a subset of the one with $\hat{f}^1$ since $\fhatp(\bjhat)> f^{1*}(\bjhat)$. Moreover, $\xopthatpp_{\rpp}(\bjhat) \leq \xopt_{\rpp}(\bjhat)$ for all $\rpp \in \Rpp$, $\xopthatp_{\rp}(\bjhat) \geq \xopt_{\rp}(\bjhat)$ for all $\rp \in \Rpp$, and there must exist non-empty subsets $\hat{R}^1 \subseteq \Rp$ and $\hat{R}^2 \subseteq \Rpp$ such that $\xopthatpp_{\hat{r}^2}(\bjhat) < \xopt_{\hat{r}^2}(\bjhat)$ and $\xopthatp_{\hat{r}^1}(\bjhat) > \xopt_{\hat{r}^1}(\bjhat)$ for all $\hat{r}^1 \in \hat{R}^1$ and all $\hat{r}^2 \in \hat{R}^2$. Due to facts (2) and (3), we have $d_{\rpphat} \leq d_{\rphat}$ for all pairs of $\hat{r}^1 \in \hat{R}^1$ and $\hat{r}^2 \in \hat{R}^2$. We further have two cases: 
\begin{itemize}
    \item[(a)] There exist one $\hat{r}^1 \in \hat{R}^1$ and one $\hat{r}^2 \in \hat{R}^2$ such that $d_{\rpphat} < d_{\rphat}$. In this case, if route $\rpphat$ is unsaturated given $\xopthatpp$, then we decrease $\xopthatp_{\rphat}(\bjhat)$ and increase $\xopthatpp_{\rpphat}(\bjhat)$ for a small positive number $\epsilon>0$. We can check that the objective function of \eqref{eq:induction} is reduced by $\epsilon (d_{\rphat}- d_{\rpphat})g(\bjhat) > 0$. Therefore, the $\epsilon$-perturbation strictly reduces the objective function value, and thus $\hat{x}^*$ with flow split $(\fhat^1, \fhat^2)$ cannot be an optimal solution. 
    
On the other hand, if route $\rpphat$ is saturated, there must exist another coalition with a lower $g(b)$ value than $\b_{\jhat}$, which is allocated to $\rpphat$ with positive flow. Without loss of generality, we assume that this coalition is $b_{\hat{j}+1}$. We decrease $\xopthatp_{\rphat}(\bjhat)$ and $\xopthatpp_{\rpphat}(\b_{\jhat+1})$ by $\epsilon>0$, increase $\xopthatp_{\rphat}(\b_{\jhat+1})$ and $\xopthatpp_{\rpphat}(\b_{\jhat})$ by $\epsilon$ (i.e. exchange a small fraction of coalition $\bjhat$ on $\rphat$ with coalition $\b_{\jhat+1}$ on $\rpphat$). Note that $g(\bjhat)> g(\b_{\jhat+1})$ and $d_{\rphat}> d_{\rpphat}$. We can thus check that the objective function of \eqref{eq:induction} is reduced by $\epsilon (d_{\rphat} - d_{\rpphat}) (g(\bjhat)- g(\b_{\jhat+1}))> 0$. Therefore, we have again found an adjustment of trip  vector $(\xopthatp, \xopthatpp)$ that reduces the objective function of \eqref{eq:induction}. Hence, for any flow split $\(\fhatp, \fhatpp\) \neq \(f^{1*}, f^{2*}\)$, the associated trip  vector $\(\xopthatp, \xopthatpp\)$ is not the optimal solution of \eqref{eq:induction}. The optimal solution of \eqref{eq:induction} must be $\xopt$, and therefore Claim \ref{claim} holds for network $\G$. 

    \item[(b)] For all $\hat{r}^1 \in \hat{R}^1$ and $\hat{r}^2 \in \hat{R}^2$, $d_{\rpphat} = d_{\rphat}$. In this case, redistributing flow of any coalition assigned to any routes in $\hat{R}^1 \cup \hat{R}^2$ does not change the objective function value. Therefore, we can redistribute the flow of $b_{\hat{j}}$ allocated to $\hat{R}^1 \cup \hat{R}^2$ in $\hat{x}^*$ in the same way as that in $x^*$, and arbitrarily redistribute the flow of any other coalitions assigned to $\hat{R}^1 \cup \hat{R}^2$ to the remaining capacities of routes in $\hat{R}^1 \cup \hat{R}^2$. Such redistribution leads to another trip vector $\hat{x}^{*'}$ and flow split vector $(\hat{f}^{1'}, \hat{f}^{2'})$ that does not change the objective function value and satisfies 
    \begin{align*}
        \hat{x}^{*'}_r(b_j) &= x^*_r(b_j), \quad \forall j=1, \dots, \hat{j}, \quad \forall r \in R,\\
        \hat{f}^{1'}(b_j) &= f^{1*}(b_j), \quad \hat{f}^{2'}(b_j) = f^{2*}(b_j), \quad \forall j=1, \dots, \hat{j}. 
    \end{align*}
    We can then compare $(\hat{f}^{1'}, \hat{f}^{2'})$ with $(f^{1*}, f^{2*})$ and again find the coalition with the smallest index such that the flow split is different, i.e. $\jhat'= \min\{j=1, \dots, n| \hat{f}^{1'}(b_{j}) \neq f^{1*}(b_{j}) \}$. We have $\jhat'> \hat{j}$. We repeat the process until either we reach $\hat{j'}= n+1$ (i.e. we have changed the flow split $(\hat{f}^{1}, \hat{f}^{2})$ to be equal to $(f^{1*}, f^{2*})$ and $\xhat^* = x^*$ without changing the objective function value) or we have found a coalition such that the $\epsilon$-perturbation described in (a) strictly reduces the objective function value. In either scenario, we have shown that the flow split $(f^{1*}, f^{2*})$ is optimal, and $\xopt$ is an optimal solution of \eqref{eq:induction}. Therefore Claim \ref{claim} holds for network $\G$. 
\end{itemize}

We have shown from cases 1 and 2 that if Claim \ref{claim} holds on any two series-parallel networks $\Gp$ and $\Gpp$, then Claim \ref{claim} holds on the network $\G$ that is constructed by merging $\Gp$ and $\Gpp$ in series or in parallel. Moreover, since Claim \ref{claim} holds trivially when the network $\G$ has a single edge, and any series-parallel network is formed by iteratively connecting series-parallel sub-networks in series or in parallel, we can conclude that Claim \ref{claim} holds on any series-parallel network. 

From \eqref{eq:decompose_V}, \eqref{eq:equal_part} and \eqref{eq:induction}, we can conclude that $S(\xopt) \geq S(\xhatopt)$. Hence, $\xopt$ must be an optimal solution for \eqref{eq:LP2k} on any series-parallel network.  
\QEDA

\color{black}

\medskip

\noindent\emph{Proof of Lemma \ref{lemma:condition_gross}.} The augmented value function satisfies monotonicity condition since for any $\ball \subseteq \ball'$, we have: 
\begin{align*}
    \barV_r(\barb)=  \max_{\bbar \subseteq \ball, ~ \bbar \in \Bbar} \Vbar_r(\bbar) \leq  \max_{\bbar \subseteq \ball', ~ \bbar \in \Bbar} \Vbar_r(\bbar) = \barV_r(\ball').
\end{align*}
We next prove that $\barVrz$ satisfies gross substitutes condition. 
Since all agents have homogeneous disutility of capacity sharing, we can simplify the trip value function $\barVrz(\ball)$ as follows: 
\begin{align*}
    \barVrz(\ball) = \sum_{\m \in \tildeb} \etamrz - \theta_r(|\tildeb|), 
\end{align*}
where 
\begin{align*}
    \etamrz &\deleq \tripm - \vm\tr, \\
    \theta_r(|\tildeb|)&\deleq  \changetrip(|\tildeb|)|\tildeb|+ \changevm(|\tildeb|)|\tildeb| \tr 
\end{align*}

Before proving that the augmented trip value function $\barVrz(\ball)$ satisfies (a) and (b) in Definition \ref{def:gross_substitute}, we first provide the following statements that will be used later: 

\emph{(i)} The function $\thetafun(|\rep(\ball)|)$ is non-decreasing in $|\rep(\ball)|$ because the marginal disutility of capacity sharing is non-decreasing in the coalition size. 

\emph{(ii)} The representative agent coalition for any trip can be constructed by selecting agents from $\ball$ in decreasing order of $\etamrz$. The last selected agent $\hat{m}$ (i.e. the agent in $\rep(\ball)$ with the minimum value of $\etamrz$) satisfies: 
\begin{align}\label{eq:ell_one}
    \eta_{\hat{m}, r} \geq \thetafun(|\rep(\ball)|) - \thetafun(|\rep(\ball)|-1).
\end{align}
That is, adding agent $\hat{m}$ to the set $\rep(\ball) \setminus \{\hat{m}\}$ increases the trip valuation. Additionally, 
\begin{align}\label{eq:ell_two}
    \etamrz < \thetafun(|\rep(\ball)|+1) - \thetafun(|\rep(\ball)|), \quad \forall \m \in \ball \setminus \rep(\ball).
\end{align}
Then, adding any agent in $\ball \setminus \rep(\ball)$ to $\rep(\ball)$ no longer increases the trip valuation.

\emph{(iii)} $|\rep(\ball')| \geq |\rep(\ball)|$ for any two agent coalitions $\ball', \ball \in \B$ such that $\ball' \supseteq \ball$. \\
\emph{Proof of (iii).} Assume for the sake of contradiction that $|\rep(\ball')|< |\rep(\ball)|$. Consider the agent $\hat{m} \in \arg\min_{\m \in \rep(\ball)} \eta_{m,r}$. The value $\eta_{\hat{m}, r}$ satisfies  \eqref{eq:ell_one}. Since $|\rep(\ball')|< |\rep(\ball)|$, $\ball' \supseteq \ball$, and we know that agents in the representative agent coalition $\rep(\ball')$ are the ones with $|\rep(\ball')|$ highest $\etamrz$ in $\ball'$, we must have $\hat{m} \notin \rep(\ball')$. From \eqref{eq:ell_two}, we know that $\eta_{\hat{m}, r} < \thetafun(|\rep(\ball')|+1) - \thetafun(|\rep(\ball')|)$. Since the marginal disutility of capacity sharing is non-decreasing in the agent coalition size, we can check that $\thetafun(|\rep(\ball)|+1) - \thetafun(|\rep(\ball)|)$ is non-decreasing in $|\rep(\ball)|$. Since $|\rep(\ball')|< |\rep(\ball)|$, we have $|\rep(\ball')|\leq |\rep(\ball)|-1$. Therefore, \[\eta_{\hat{m}, r} < \thetafun(|\rep(\ball')|+1) - \thetafun(|\rep(\ball')|) \leq \thetafun(|\rep(\ball)|) - \thetafun(|\rep(\ball)|-1),\] which contradicts \eqref{eq:ell_one} and the fact that $\hat{m} \in \tildeb$. Hence, $|\rep(\ball')| \geq |\rep(\ball)|$. 

We now prove that $\barVrz$ satisfies \emph{(i)} in Definition \ref{def:gross_substitute}.
For any $\ball, \ball' \subseteq \M$ and $\ball \subseteq \ball'$, consider two cases: \\
\emph{Case 1:} $i \notin \rep(\{i\} \cup\ball')$. In this case, $\rep(\ball' \cup i) = \rep(\ball')$, and $ \barVrz(i|\ball')= \barVrz(\ball' \cup i) -  \barVrz(\ball')=0$. Since $ \barVrz$ satisfies monotonicity condition, we have $ \barVrz(i|\ball) \geq 0 $. Therefore, $ \barVrz(i|\ball) \geq  \barVrz(i|\ball')$.

\noindent\emph{Case 2:} $i \in \rep(\{i\} \cup\ball')$. We argue that $i \in \rep(\{i\} \cup\ball)$. From 
\eqref{eq:ell_one}, $\eta_{i,r} \geq \thetafun(|\rep(\ball')|) - \thetafun(|\rep(\ball')|-1)$. Since $\ball' \supseteq \ball$, we know from (iii) that $|\rep(\ball')| \geq |\rep(\ball)|$. Hence, $\eta_{i,r} \geq\thetafun(|\rep(\ball)|) - \thetafun(|\rep(\ball)|-1)$, and thus $i \in \rep(\{i\} \cup\ball)$. 

We define $\hat{m}' \deleq \arg\min_{m \in \rep(\ball')} \eta_{m,r}$ and $\hat{m} \deleq \arg\min_{m \in \rep(\ball)} \eta_{m,r}$. We also consider two thresholds $\mu' = \thetafun(|\rep(\ball')|+1) - \thetafun(|\rep(\ball')|)$, and $\mu = \thetafun(|\rep(\ball)|+1) - \thetafun(|\rep(\ball)|)$.  
Since $\ball' \supseteq \ball$, from (iii), we have $|\rep(\ball')| \geq |\rep(\ball)|$ and thus $\mu' \geq \mu$. 
We further consider four sub-cases:

\emph{(2-1)} $\eta_{\hat{m}',r} \geq \mu'$ and $\eta_{\hat{m}, r} \geq \mu$. From \eqref{eq:ell_one} and \eqref{eq:ell_two}, $\rep(\{i\} \cup\ball') = \rep(\ball') \cup \{i\}$ and $\rep(\{i\} \cup\ball) = \rep(\ball) \cup \{i\}$. The marginal value of $i$ is $\barVrz(i |\ball')=\eta_{i,r} - \mu'$, and $\barVrz(i|\ball)= \eta_{i,r}- \mu$. Since $\mu' \geq \mu$, $\barVrz(i|\ball') \leq \barVrz(i|\ball)$. 
    
\emph{(2-2)} $\eta_{\hat{m}',r}< \mu'$ and $\eta_{\hat{m}, r} \geq \mu$. Since $i \in \rep(\{i\} \cup\ball')$ in \emph{Case 2}, we know from \eqref{eq:ell_one} and \eqref{eq:ell_two} that $\rep(\{i\} \cup\ball')= \rep(\ball') \setminus \{\hat{m}'\}\cup\{i\}$ and $\rep(\{i\} \cup\ball) = \rep(\ball) \cup \{i\}$. Therefore, $\barVrz(i|\ball') = \eta_{i,r} -\eta_{\hat{m}',r}$ and $\barVrz(i|\ball) = \eta_{i,r}- \mu$. We argue in this case, we must have $|\rep(\ball')|>|\rep(\ball)|$. Assume for the sake of contradiction that $|\rep(\ball')|=|\rep(\ball)|$, then $\mu' = \mu$ and $\eta_{\hat{m}',r} \geq \eta_{\hat{m}, r}$ because $\ball' \supseteq \ball$. However, this contradicts the assumption of this sub-case that $\eta_{\hat{m}',r} < \mu' = \mu \leq {\eta_{\hat{m}, r}}$. Hence, we must have $|\rep(\ball')| \geq |\rep(\ball)|+1$. Then, from \eqref{eq:ell_one}, we have $\eta_{\hat{m}',r} \geq \thetafun(|\rep(\ball')|) - \thetafun(|\rep(\ball')|-1) \geq \mu$. Hence, $\barVrz(i|\ball')  \leq \barVrz(i|\ball)$. 

\emph{(2-3)} $\eta_{\hat{m}',r}\geq \mu'$ and $\eta_{\hat{m}, r} < \mu$. From \eqref{eq:ell_one} and \eqref{eq:ell_two}, $\rep(i \cup\ball') = \rep(\ball') \cup \{i\}$ and $\rep(\{i\} \cup\ball) = \rep(\ball) \setminus \{\hat{m}'\}  \cup \{i\}$. Therefore, $\barVrz(i|\ball') = \eta_{i,r} -\mu'$ and $\barVrz(i|\ball) = \eta_{i,r}- \eta_{\hat{m}, r}$. Since $\mu' \geq \mu \geq \eta_{\hat{m}, r}$, we know that $\barVrz(i|\ball')  \leq \barVrz(i|\ball)$. 

\emph{(2-4)} $\eta_{\hat{m}',r} < \mu'$ and $\eta_{\hat{m}, r} < \mu$. From \eqref{eq:ell_one} and \eqref{eq:ell_two},  $\rep(\{i\} \cup\ball') = \rep(\ball') \setminus \{\hat{m}'\}\cup \{i\}$, and $\rep(\{i\} \cup\ball) = \rep(\ball) \setminus \{\hat{m}\}\cup \{i\}$. Therefore, $\barVrz(i|\ball') = \eta_{i,r} -\eta_{\hat{m}',r}$ and $\barVrz(i|\ball) = \eta_{i,r}- \eta_{\hat{m}, r}$. If $|\rep(\ball')| = |\rep(\ball)|$, then we must have $\eta_{\hat{m}',r} \geq \eta_{\hat{m}, r}$, and hence $\barVrz(i|\ball') \leq \barVrz(i|\ball)$. On the other hand, if $|\rep(\ball')| \geq |\rep(\ball)|+1$, then from \eqref{eq:ell_one} we have $\eta_{\hat{m}, r} \geq \thetafun(|\rep(\ball')|) - \thetafun(|\rep(\ball')|-1) \geq \mu > \eta_{\hat{m}, r}$. Therefore, we can also conclude that $\barVrz(i|\ball') \leq \barVrz(i|\ball)$.

From all four subcases, we can conclude that in case 2, $\barVrz(i|\ball) \geq \barVrz(i|\ball')$.  

We now prove that $ \barVrz$ satisfies condition (\emph{ii}) of Definition \ref{def:gross_substitute} by contradiction. Assume for the sake of contradiction that Definition \ref{def:gross_substitute} (ii) is not satisfied. Then, there must exist a coalition $\ball \in \Ball$, and $i, j,k \in M\setminus \ball$ such that: 
\begin{subequations}
\begin{align}
    &\barVrz(\{i, j\}|\ball)+ \barVrz(k|\ball) > \barVrz(i|\ball) + \barVrz(\{j, k\}|\ball),\quad \Rightarrow \quad \barVrz(j|\{i\} \cup \ball) > \barVrz(j|\{k\} \cup \ball), \label{subeq:one}\\
    &\barVrz(\{i, j\}|\ball)+ \barVrz(k|\ball) > \barVrz(j|\ball) + \barVrz(\{i,k\}|\ball), \quad \Rightarrow \quad \barVrz(i|\{j\}\cup \ball) > \barVrz(i|\{k\} \cup \ball),\label{subeq:two}
\end{align}
\end{subequations}
where $\barVrz(\{i, j\}|\ball)= \barVrz(\ball \cup \{i, j\})- \barVrz(\ball)$ for any $i, j \in M$ and $\ball \in \Ball$, $\barVrz(k|\ball) = \barVrz(\{k\} \cup \ball) - \barVrz(\ball)$ for any $k \in M$ and $\ball \in \Ball$, and $\barVrz(i|\{j\}\cup \ball) = \barVrz(\{i, j\} \cup \ball) - \barVrz(\{j\} \cup \ball)$ for any $i, j \in M$ and $\ball \in \Ball$. 

We consider the following four cases: 

\emph{Case A:} $\rep\(\ball \cup \{i, j\}\) = \rep\(\ball \cup \{i\}\) \cup \{j\}$ and $\rep\(\ball \cup \{j, k\}\) = \rep\(\ball \cup \{k\}\) \cup \{j\}$. In this case, if $|\rep\(\ball \cup \{i\}\)| \geq  |\rep\(\ball \cup \{k\}\)|$, then $\barVrz(j|\{i\} \cup \ball) \leq \barVrz(j|\{k\} \cup \ball)$, which contradicts \eqref{subeq:one}. On the other hand, if $|\rep\(\ball \cup \{i\}\)| <  |\rep\(\ball \cup \{k\}\)|$, then we must have $\rep\(\ball \cup \{i\}\)=\rep(\ball)$ and $\rep\(\ball \cup \{k\}\) = \rep(\ball) \cup \{k\}$. Therefore, $\barVrz(i|\{j\} \cup \ball)=0$, and \eqref{subeq:two} cannot hold. We thus obtain the contradiction.  

\emph{Case B:} $|\rep\(\ball \cup \{i, j\}\)| = |\rep\(\ball \cup \{i\}\)|$ and $|\rep\(\ball \cup \{j, k\}\)| = |\rep\(\ball \cup \{k\}\)|$. We further consider the following four sub-cases: 

\emph{(B-1).} $\rep\(\ball \cup \{i, j\}\) = \rep\(\ball \cup \{i\}\) $ and $\rep\(\ball \cup \{j, k\}\) = \rep\(\ball \cup \{k\}\) $. In this case, $\barVrz(j|\{i\} \cup \ball) = \barVrz(j|\{k\} \cup \ball)=0$. Hence, we arrive at a contradiction against \eqref{subeq:one}. 

\emph{(B-2).} $\rep\(\ball \cup \{i, j\}\) \neq \rep\(\ball \cup \{i\}\) $ and $\rep\(\ball \cup \{j, k\}\) = \rep\(\ball \cup \{k\}\) $. In this case, when $j$ is added to the set $\ball \cup \{i\}$, $j$ replaces an agent, denoted as $\hat{m} \in \ball \cup \{i\}$. Since $\hat{m}$ is replaced, we must have $\eta_{\hat{m}, r} \leq \eta_{m, r}$ for any $\m \in \rep(\ball \cup\{j\})$. If $\hat{m} = i$, then $\rep(\ball \cup\{i, j\}) = \rep(\ball \cup \{j\})$. Hence, $\barVrz(i|\{j\} \cup  \ball)=0$, and we arrive at a contradiction with \eqref{subeq:two}. On the other hand, if $\hat{m} \neq i$, then $\hat{m}$ is an agent in coalition $\ball$. This implies that $\hat{m} \in \ball$ should be replaced by $j$ when $j$ is added to the set $\{k\} \cup \ball$, which contradicts the assumption of this case that $\rep\(\ball \cup \{j, k\}\) = \rep\(\ball \cup \{k\}\) $. 
        
\emph{(B-3).} $\rep\(\ball \cup \{i, j\}\) = \rep\(\ball \cup \{i\}\) $ and $\rep\(\ball \cup \{j, k\}\) \neq \rep\(\ball \cup \{k\}\)$. Analogous to case \emph{B-2}, we know that $\rep\(\ball \cup \{j, k\}\) = \rep\(\ball \cup \{j\}\)$ and $\eta_{j,r} \geq \eta_{k,r}$.  Moreover, since $\rep\(\ball \cup \{i, j\}\) = \rep\(\ball \cup \{i\}\) $, we must have $\eta_{j, r} \leq \eta_{i,r}$. Therefore, $\barVrz(\ball \cup \{i, j\}) = \barVrz(\ball \cup \{i\})$, and $\barVrz(i|\{j\} \cup \ball)= \barVrz(\ball\cup \{i\}) - \barVrz(\ball\cup \{j\})$. Since $\eta_{j, r} \leq \eta_{i,r}$ and $\eta_{j, r} \geq \eta_{k,r}$, we know that $\barVrz(i|\{k\} \cup \ball) =\barVrz(\ball\cup \{i\}) - \barVrz(\ball \cup \{k\}) \geq \barVrz(\ball\cup \{i\}) - \barVrz(\ball\cup \{j\}) = \barVrz(i|\{j\} \cup \ball)$, which contradicts \eqref{subeq:two}. 

\emph{(B-4).} $\rep\(\ball \cup \{i, j\}\) \neq \rep\(\ball \cup \{i\}\) $ and $\rep\(\ball \cup \{j, k\}\) \neq \rep\(\ball \cup \{k\}\) $. In this case, if $\rep\(\ball \cup \{i, j\}\) = \rep\(\ball \cup \{j\}\)$, then $\barVrz(i|\{j\} \cup \ball)= \barVrz(\{i, j\}\cup \ball)- \barVrz(\{j\} \cup \ball) = \barVrz(\{j\} \cup \ball) - \barVrz(\{j\} \cup  \ball)=0$, which contradicts \eqref{subeq:two}. On the other hand, if $\rep\(\ball \cup \{i, j\}\) \neq \rep\(\ball \cup \{j\}\)$, then one agent $\hat{m} \in \ball$ must be replaced by $j$ when $j$ is added into the set $\ball \cup \{i\}$, i.e. $\rep\(\ball \cup \{i, j\}\) = \rep\(\ball\setminus \{\hat{m}\} \cup \{i, j\}\) $. Hence, $\eta_{\hat{m}, r} \leq \eta_{i,r}$ and $\eta_{\hat{m}, r} \leq \eta_{j, r}$. If $\eta_{\hat{m}, r} \leq  \eta_{k, r}$, then under the assumption that $|\rep\(\ball \cup \{j, k\}\)|=|\rep\(\ball \cup \{k\}\)|$ and $\rep\(\ball \cup \{j, k\}\) \neq \rep\(\ball \cup \{k\}\)$, we must have $\rep\(\ball \cup \{j, k\}\)= \rep\(\ball\setminus \{\hat{m}\} \cup \{j, k\}\)$. Then, we can check that $\barVrz(j|i, b) = \barVrz(j|k, b)$, which contradicts \eqref{subeq:one}. 

        On the other hand, if $\eta_{\hat{m}, r} > \eta_{k, r}$, then $\rep\(\ball \cup \{j, k\}\)= \rep\(\ball \cup \{j\}\)$. In this case, $\barVrz(i|\{j\} \cup \ball)$ is the change of trip value by replacing $\hat{m}$ with $i$, and $\barVrz(i|\{k\} \cup \ball)$ is the change of trip value by replacing $k$ with $i$. Since $\eta_{k, r} < \eta_{\hat{m}, r}$, we must have $\barVrz(i|\{j\} \cup \ball) < \barVrz(i|\{k\} \cup \ball)$, which contradicts \eqref{subeq:two}.

\emph{Case C:} $\rep\(\ball \cup \{i, j\}\) = \rep\(\ball \cup \{i\}\) \cup \{j\}$ and $|\rep\(\ball \cup \{j, k\}\)| = |\rep\(\ball \cup \{k\}\)|$. We further consider the following sub-cases: 

\emph{(C-1).}  $\rep\(\ball \cup \{j, k\}\) = \rep\(\ball \cup \{k\}\)$. In this case, $\eta_{j, r} \leq \eta_{m,r}$ for all $m \in \rep(\ball \cup \{k\})$, and $\eta_{j, r} < \thetafun(|\rep(\ball \cup \{k\})+1|) - \thetafun(|\rep(\ball \cup \{k\})|)$. Since $\rep\(\ball \cup \{i, j\}\) = \rep\(\ball \cup \{i\}\) \cup \{j\}$, we know that $\eta_{j, r} \geq \thetafun(|\rep(\ball \cup \{i\})+1|) - \thetafun(|\rep(\ball \cup \{i\})|)$. Since disutility of capacity sharing is non-decreasing in agent coalition size, for $\eta_{j, r}$ to satisfy both inequalities, we must have $|\rep(\ball \cup \{i\})| < |\rep(\ball \cup \{k\})|$. Then, we must have $\rep(\ball\cup \{i\}) = \rep(\ball)$ and $\rep(\ball\cup \{k\}) = \rep(\ball) \cup \{k\}$. Therefore,  $\barVrz(\{i, j\} \cup \ball) = \barVrz(\{j\} \cup \ball)$ and $\barVrz(\{i, k\} \cup \ball) = \barVrz(\{k\} \cup \ball)$. Hence, $\barVrz(i|\{j\} \cup \ball) = \barVrz(i|\{k\} \cup  \ball)=0$, which contradicts \eqref{subeq:two}. 

\emph{(C-2).}  $\rep\(\ball \cup \{j, k\}\) \neq \rep\(\ball \cup \{k\}\)$. Since $|\rep\(\ball \cup \{j, k\}\)| = |\rep\(\ball \cup \{k\}\)|$, $j$ replaces an agent $\hat{m}$ in $\ball \cup \{k\}$, and $\eta_{\hat{m}, r}\leq \eta_{m,r}$ for all $\m \in \ball \cup{k}$. If $\hat{m} = k$, then $\rep\(\ball \cup \{j, k\}\) = \rep\(\ball \cup \{j\}\)$. Therefore, $\barVrz(j|\{i\} \cup  \ball) = \eta_{j, r} - \(\thetafun(|\rep(\ball \cup \{i\})|+1) - \thetafun(|\rep(\ball \cup \{i\})|)\)$ and $\barVrz(j|\{k\} \cup  \ball) = \eta_{j, r} - {\eta_{k,r}}$. If $\eta_{k, r} \leq \thetafun(|\rep(\ball \cup \{i\})|+1) - \thetafun(|\rep(\ball \cup \{i\})|)$, then \eqref{subeq:one} is contradicted. Thus, $\eta_{k, r} > \thetafun(|\rep(\ball \cup \{i\})|+1) - \thetafun(|\rep(\ball \cup \{i\})|)$. Since $k$ is replaced by $j$ when $j$ is added to $\ball \cup \{k\}$, we must have $\eta_{k, r} < \thetafun(|\rep(\ball \cup \{j\})|+1) - \thetafun(|\rep(\ball \cup \{j\})|)$. For $\eta_{k, r}$ to satisfy both inequalities, we must have $|\rep(\ball \cup \{j\})| > |\rep(\ball \cup \{i\})|$. Hence, $\rep(\ball \cup \{j\}) = \rep(\ball) \cup \{j\}$ and $\rep(\ball \cup \{i\}) = \rep(\ball)$. Then, $\barVrz(i|\{j\} \cup \ball) = \barVrz(\ball\cup \{i, j\}) -  \barVrz(\ball\cup \{ j\}) = 0$, which contradicts \eqref{subeq:two}. 

On the other hand, if $\hat{m} \in \ball$, then we know from \eqref{eq:ell_two} that $\eta_{\hat{m}, r} < \thetafun(|\rep\(\ball \cup \{k\}\)|+1) - \thetafun(|\rep\(\ball \cup \{k\}\)|)$. Additionally, since $\rep\(\ball \cup \{i, j\}\) = \rep\(\ball \cup \{i\}\) \cup \{j\}$, we know from \eqref{eq:ell_one} that $\eta_{\hat{m}, r} \geq \thetafun(|\rep\(\ball \cup \{i\}\)|+1) - \thetafun(|\rep\(\ball \cup \{i\}\)|)$. If $\eta_{\hat{m}, r}$ satisfies both inequalities, then we must have $|\rep\(\ball \cup \{i\}\)|< |\rep\(\ball \cup \{k\}\)|$. Therefore, $\rep\(\ball \cup \{i\}\)= \rep(\ball)$. Then, $\barVrz(i|\{j\} \cup \ball)=0$, which contradicts \eqref{subeq:two}.

\emph{Case D:} $|\rep\(\ball \cup \{i, j\}\)| = |\rep\(\ball \cup \{i\}\)|$ and $\rep\(\ball \cup \{j, k\}\) = \rep\(\ball \cup \{k\}\) \cup \{j\}$. We further consider the following sub-cases: 

\emph{(D-1).}  $\rep\(\ball \cup \{i, j\}\) = \rep\(\ball \cup \{i\}\)$. In this case, analogous to \emph{(C-1)}, we know that $|\rep(\ball \cup \{k\})| < |\rep(\ball \cup \{i\})|$. Therefore, $\rep(\ball\cup \{k\})=\rep(\ball)$ and $\rep(\ball \cup \{i\})=\rep(\ball) \cup \{i\}$. Therefore, $\eta_{k, r} < \eta_{i,r}$. Additionally, since $\rep\(\ball \cup \{i, j\}\) = \rep\(\ball \cup \{i\}\)$, $\eta_{j, r} < \eta_{i,r}$. Then, $\barVrz(i|\{j\} \cup  \ball)= \barVrz(\{i\} \cup \ball) - \barVrz(\{j\} \cup \ball)$ and $\barVrz(i|\{k\} \cup \ball) = \barVrz(\{i\} \cup \ball)- \barVrz(\ball)$. Since $ \barVrz$ is monotonic, $\barVrz(\{j\} \cup \ball) \geq \barVrz(\ball)$ so that $\barVrz(i|\{j\} \cup \ball) \leq \barVrz(i|\{k\} \cup \ball)$, which contradicts \eqref{subeq:two}. 

\emph{(D-2).} $\rep\(\ball \cup \{i, j\}\) \neq \rep\(\ball \cup \{i\}\)$. Since $|\rep\(\ball \cup \{i, j\}\)| = |\rep\(\ball \cup \{i\}\)|$, $j$ replaces the agent $\hat{m} \in \ball \cup \{i\}$ such that $\eta_{\hat{m}, r} \leq \eta_{m, r}$ for all $\m \in \rep(\ball \cup \{i\})$. If $\hat{m} = i$, then analogous to case \emph{C-2}, we know that if \eqref{subeq:two} is satisfied, then $|\rep(\ball \cup \{j\})|< |\rep(\ball \cup \{k\})|$. Hence, $\rep(\ball \cup \{j\}) = \rep(\ball)$ and $V(j |\{i\} \cup \ball) =0$, which contradicts \eqref{subeq:one}. 

On the other hand, if $\hat{m} \in \ball$, then again analogous to case \emph{C-2}, we know that $|\rep\(\ball \cup \{k\}\)|< |\rep\(\ball \cup \{i\}\)|$. Therefore, $\rep\(\ball \cup \{k\}\)=\rep(\ball)$, and $\rep\(\ball \cup \{i\}\)= \rep(\ball)\cup \{i\}$. Then, $\barVrz(j|\{i\} \cup \ball) = \barVrz(\ball \setminus \{\hat{m}\} \cup \{i, j\}) - \barVrz(\{i\} \cup \ball)$, and $\barVrz(j|\{k\} \cup \ball) = \barVrz(\ball\cup \{j\}) - \barVrz(\ball)$. Since $\hat{m} \neq i$, $\barVrz(i|\{j\} \cup \ball)=\barVrz(\ball \setminus \{\hat{m}\} \cup \{i, j\}) - \barVrz(\{j\} \cup \ball) = \eta_{i, r}- \eta_{\hat{m}, r}$. Additionally, since $\rep(\{i\} \cup \ball) = \rep(\ball) \cup \{i\}$, $\barVrz(i|\{k\} \cup \ball)=\barVrz(\{i\} \cup  \ball) - \barVrz(\ball) = \eta_{i, r} -(\thetafun(|\rep(\ball)|+1)- \thetafun(|\rep(\ball)|))$. Since $\rep(\ball \cup \{i\})= \rep(\ball) \cup \{i\}$ and $\hat{m} \in \ball$, we know from \eqref{eq:ell_one} that $\eta_{\hat{m}, r}\geq  \thetafun(|\rep(\ball)|+1)- \thetafun(|\rep(\ball)|)$. Therefore, $\barVrz(i|\{j\} \cup \ball) \leq \barVrz(i|\{k\} \cup \ball)$, which contradicts \eqref{subeq:two}.

From all above four cases, we can conclude that condition \emph{(ii)} of Definition \ref{def:gross_substitute} is satisfied. We can thus conclude that $ \barVrz$ satisfies gross substitutes condition.
\QEDA

\medskip 
\noindent\emph{Proof of Lemma \ref{lemma:integer}.} For any route $r \in \Ropt = \{r|\kopt_r>0\}$, we denote $L_r$ as the set of slots that correspond using route $r$. We denote $L=\cup_{r \in \Ropt} L_r$. Given any  integer optimal solution $\xopt$ of \eqref{eq:LPy}, we denote $\{\ball_l\}_{l \in L_r}$ as the set of augmented  coalitions in $\Ball$ such that $\bar{x}_r^{*}(\ball)=1$. In particular, if the number of agent coalitions that take route $r$ is less than $k_r^{*}$, then $\ball_l=\emptyset$ for some of $l \in L_r$. 

We first show that $\yopt$ is an integer optimal solution of \eqref{eq:LPy} if and only if $((\ball_l)_{l \in L}, \uopt)$ is a Walrasian equilibrium of the equivalent economy with good set $M$ and buyer set $L$. Here, $\uopt$ is the optimal dual variable in the dual program of \eqref{eq:LPy} associated with constraint \eqref{subeq:LPy1}. If $((\ball_l)_{l \in L}, \uopt)$ is a Walrasian equilibrium of the equivalent economy, then the associated $\yopt$ must be a feasible solution of \eqref{eq:LPy}. We define $\lambda_r^{*} = \bar{V}_{l}(\ball_l) - \sum_{m \in \ball_l} u_m^{*} = \max_{\ball \in \Ball} \{\bar{V}_{l}(\ball) - \sum_{m \in \ball} u_m^{*}\}$ for $l \in L_r$, where the second equality follows from the definition of Walrasian equilibrium. We can check that $(\uopt, \tollopt)$ satisfies the dual constraints of \eqref{eq:LPy}, and $(\xopt, \uopt, \tollopt)$ satisfies the complementary slackness conditions associated with all the primal and dual constraints. Thus, $\xopt$ is an integer optimal solution of \eqref{eq:LPy}. On the other hand, we can analogously argue that if $\xopt$ is an integer optimal solution of \eqref{eq:LPy}, then the associated $(\ball_l)_{l \in L}$ and the dual optimal solution $\uopt$ is a Walrasian equilibrium in the equivalent economy. 

From Lemma \ref{lemma:existence_eq}, we know that when the augmented value function $\bar{V}_{l} =\bar{V}_{r}$ satisfies the monotonicity and gross substitutes conditions, Walrasian equilibrium exists in the equivalent economy. As a result, we know that the associated $\yopt$ is an integer optimal solution in \eqref{eq:LPy}. 

Finally, in \eqref{eq:construct}, we select one representative agent coalition $h_r(\ball)$ for each $\ball$ that is assigned to $r$ as the true feasible agent coalition that takes route $r$. Such $\xopt$ achieves the same social welfare as that in $\yopt$, and thus is an optimal solution of \eqref{eq:LP2k}. \QEDA

\section{Proof of Statements in Section \ref{sec:strategyproof}}

We define $\Uopt \deleq \left\{\u|\exists \toll \text{ such that } (\u, \toll) \text{ is optimal solution of \eqref{eq:D1bar}}\right\}$ as the equilibrium utility set. 
\begin{lemma}\label{lemma:utility}
If the network is series-parallel, a utility vector $\uopt \in \Uopt$ if and only if there exists a vector $\lambda^*$ such that $\left(\uopt, \lambda^*\right)$ is an optimal solution of \eqref{eq:D2k}.
\end{lemma}

\noindent\emph{Proof of Lemma \ref{lemma:utility}.} We first show that for any optimal utility vector $\uopt \in \Uopt$, there exists a vector $\lambda^*$ such that $\(\uopt, \lambda^*\)$ is an optimal solution of \eqref{eq:D2k}. Since $\uopt \in \Uopt$, there must exist an edge price vector $\tollopt$ such that $\(\uopt, \tollopt\)$ is an optimal solution of \eqref{eq:D1bar}. Consider $\lambda^* = \(\lambda^*_r\)_{\r \in \R}$ as follows: 
\begin{align}\label{eq:lambdaopt}
\lambda^*_r = \sum_{\e \in \r} \tollopt_e, \quad \forall \r \in \R. 
\end{align}
Since $\(\uopt, \tollopt\)$ is feasible in \eqref{eq:D1bar}, we can check that $\(\uopt, \lambda^*\)$ is also a feasible solution of \eqref{eq:D2k}. Moreover, since $\(\xopt, \uopt, \tollopt\)$ satisfies complementary slackness conditions with respect to \eqref{eq:LP1bar} and \eqref{eq:D1bar}, $\(\xopt, \uopt, \lambda^*\)$ also satisfies complementary slackness conditions with respect to \eqref{eq:LP2k} and \eqref{eq:D2k}. Therefore, $\(\umopt, \lambda^*\)$ is an optimal solution of \eqref{eq:D2k}.

We next show that for any optimal solution $\(\uopt, \lambda^*\)$ of \eqref{eq:D2k}, we can find an edge price vector $\tollopt$ such that $\(\uopt, \tollopt\)$ is an optimal solution of \eqref{eq:D1bar} (i.e. $\uopt \in \Uopt$). We prove this argument by mathematical induction. To begin with, if the network only has a single edge $\E=\{\e\}$, then for any optimal solution $\(\uopt, \lambopt\)$, we can check that $\(\uopt, \tollopt\)$ where $\tollopt_e=\lambopt_e$ is an optimal solution of \eqref{eq:LP1bar}. We now prove that if this argument holds on two series-parallel networks $\Gp$ and $\Gpp$, then it also holds on the network constructed by connecting $\Gp$ and $\Gpp$ in parallel or in series. We prove the case of parallel connection and series connection separately as follows: 

\medskip
\noindent\emph{(Case 1).} The network $\G$ is constructed by connecting $\Gp$ and $\Gpp$ in parallel. In each network $\Gi$ ($i=1, 2$), we define $\Ei$ as the set of edges, $\Ri$ as the set of routes. We also define $\kopti$ as the optimal route capacity vector computed from Algorithm \ref{alg:flow} in $\Gi$, and $\Ropti = \{r\in \Ri|\kopti_r>0\}$ as the set of routes with positive capacity in $\kopti$. Since $\Gp$ and $\Gpp$ are connected in parallel, we have $\Ep\cup \Epp=\E$, $\Rp\cup \Rpp=\R$, $\kopt=\(k^{1*}, k^{2*}\)$, and $\Ropt=\Roptp\cup \Roptpp$. 

For each $i=1, 2$, we consider the sub-problem, where agents organize trips on the sub-network $\Gi$. For any $\(\xopt, \uopt, \lambopt\)$ on the original network $\G$, we define the trip vector $\xopti= \(\xopt_{\ri}(\b)\)_{\ri \in \Ri, \b \in \B}$ and the route price vector $\lambopti= \(\tollopt_{\ri}\)_{\ri \in \Ri}$ for the sub-network $\Gi$. We can check that the vector $\xopti$ is a feasible solution of \eqref{eq:LP2k} for the subproblem, where the route set $\Ropt$ in the original problem \eqref{eq:LP2k} is replaced by $\Ropti$, and $\kopt$ is replaced by $\kopti$, and the vector $\(\uopt, \lambopti\)$ is a feasible solution of \eqref{eq:LP2k}. Additionally, since the original optimal solutions $\xopt$ and $\(\uopt, \lambopt\)$ satisfy the complementary slackness conditions of constraints \eqref{subeq:LP11}-\eqref{subeq:LP12} and \eqref{subeq:D11} for all $\m \in \M$ and all $\r \in \Ropt= \Roptp\cup \Roptpp$, we know that $\xopti$ and $\(\uopt, \lambopti\)$ must also satisfy the complementary slackness conditions of these constraints in each subproblem. Therefore, $\xopti$ is an integer optimal solution of \eqref{eq:LP2k} and $\(\uopt, \lambopti\)$ is an optimal solution of \eqref{eq:D2k} in the subproblem $\probi$.

From our assumption of mathematical induction, there exists an edge price vector $\tollopti=\(\tollopt_e\)_{\e \in \Ei}$ such that $\(\uopt, \tollopti\)$ is an optimal solution of \eqref{eq:D1bar} in each subproblem $i$ with sub-network $\Gi$. Thus, $\(\uopt, \tollopti\)$ satisfies the feasibility constraints in \eqref{eq:D1bar} of each subproblem $i$, and $\xopti$ and $\(\uopt, \tollopti\)$ satisfy the complementary slackness conditions with respect to constraints \eqref{subeq:LP11} for each $\m \in \M$, \eqref{subeq:LP12} for each $\e \in \Ei$, \eqref{subeq:D11} for each $\ri \in \Ri$. Consider the edge price vector $\tollopt=\(\tolloptp, \tolloptpp\)$. Since $\R=\Rp \cup \Rpp$ and $\E=\Ep\cup \Epp$, $\(\uopt, \tollopt\)$ must be feasible in \eqref{eq:D1bar} on the original network, and $\xopt$, $\(\uopt, \tollopt\)$ must satisfy the complementary slackness conditions with respect to constraints \eqref{subeq:LP11} -- \eqref{subeq:LP12}, and \eqref{subeq:D11}. Therefore, we can conclude that for any optimal solution $\(\uopt, \tollopt\)$ of \eqref{eq:D2k}, there exists an edge price vector $\tollopt$ such that $\(\uopt, \tollopt\)$ is an optimal solution of \eqref{eq:D1bar} in network $\G$.

\medskip
\noindent\emph{(Case 2).} The network $\G$ is constructed by connecting $\Gp$ and $\Gpp$ in series. Same as that in case 1, we define $\Ei$ as the set of edges in the sub-network $\Gi$ ($i=1, 2$), and $\Ri$ as the set of routes. Since $\Gp$ and $\Gpp$ are connected in series, we have $\E = \Ep \cup \Epp$, and $\R=\Rp \times \Rpp$.

We define a sub-trip $
(\b, \ri)$ as the trip in the sub-network $\Gi$ where agent coalition $\b$ takes route $\ri \in \Ri$. Analogous to the value of trip defined in \eqref{eq:value_of_trip}, the value of each sub-trip $\(\b, \ri\) \in \B \times \Ri$ is defined as: 
\begin{align}\label{eq:Virb}
    V_{\ri}^i(\b) = \sum_{\m \in \bbar} \tripmi - |b|\changetrip^i(|b|) - \sum_{\m \in \b} \vm d_{\ri} - |b|\changevm(|b|) d_{\ri}, ~ \forall \bbar \in \B, ~\forall \ri \in \Ri, ~ \forall i=1, 2, 
\end{align}where $\tripmi$ can be any number in $[0, \tripm]$ as long as  $\tripmp+\tripmpp=\tripm$, and $\changetrip^i(|b|)$ can be any number in $[0, \changetrip(|b|)]$ as long as  $\changetrip^1(|b|)+\changetrip^2(|b|)=\changetrip(|b|)$. We can check that $V_{\rp}^1(\b)+ V_{\rpp}^2(\b)= V_{\rp\rpp}(\b)$ is the value of the entire trip $\(\b, \rp\rpp\)$ of the original network.

We denote the trip organization vector on $\Gi$ as $\xxi=\(\xxi_{\ri}(\b)\)_{\ri \in \Ri, \b \in \B}$, where $\xxi_{\ri}(\b)=1$ if the sub-trip $\(\b, \ri\)$ is organized in $\Gi$, and 0 otherwise. The optimal trip organization problem \eqref{eq:LP1bar} can be equivalently presented by $\(\xp, \xpp\)$ as follows:  
\begin{subequations}\label{eq:LP1-r}
    \begin{align}
         \max_{\xp, \xpp} \quad &S(\xp, \xpp) = \sum_{\bbar \in \Bbar}\sum_{\rp \in \Rp} V^1_{\rp}(\b) \xp_{\rp}(\b)+  \sum_{\bbar \in \Bbar}\sum_{\rpp \in \Rpp} V^2_{\rpp}(\b) \xpp_{\rpp}(\b) \notag\\
   s.t. \quad  &  \sum_{\ri \in \Ri}\sum_{\b\ni \m} \xxi_{\ri}(\b) \leq 1, \quad \forall \m \in \M, \quad \forall i=1, 2 \label{subeq:LP11-r}\\
 \quad &\sum_{\ri \ni \e} \sum_{\bbar \in \Bbar} \xxi_{\ri}(\b)  \leq \qe, \quad \forall \e \in \Ei, \quad \forall i=1, 2 \label{subeq:LP12-r}\\
 & \sum_{\rp \in \Rp}\xp_{\rp}(\b)= \sum_{\rpp \in \Rpp}\xp_{\rpp}(\b), \quad \forall \b \in \B, \label{subeq:LP1-extra}\\
    &\xxi_{\ri}(\b) \geq 0, \quad \forall \bbar \in \Bbar, \quad \forall \ri \in \Ri, \quad \forall i=1, 2,\label{subeq:LP13-extra}
    \end{align}
\end{subequations}
where \eqref{subeq:LP11-r} and \eqref{subeq:LP12-r} are the constraints of $\xxi$ in the trip organization sub-problem on $\Gi$. The constraint \eqref{subeq:LP1-extra} ensures that any agent coalition that takes a route in $\Gp$ (resp. $\Gpp$) must also takes a route in $\Gpp$ (resp. $\Gp$) to complete a trip in the original network $\G$.

%That, the total value of all organized trips given $\xp, \xpp$ is equivalent to that given the corresponding original trip vector $\x$, and $S(\xp, \xpp)$ dose not depend on how $\tripmi$ is split between $\tripmp$ and $\tripmpp$. Therefore,  

We denote $\kopti$ as the optimal capacity vector of sub-network $\Gi$ computed from Algorithm \ref{alg:flow}. Since Algorithm \ref{alg:flow} allocates capacity on routes in increasing order of their travel time, and the total travel time of each route is $d_{\rp\rpp}=d_{\rp}+d_{\rpp}$, we know that $k^{1*}_{\rp}= \sum_{\rpp \in \Rpp} \kopt_{\rp\rpp}$ for all $\rp \in \Rp$ and $k^{2*}_{\rpp}=\sum_{\rp \in \Rp} \kopt_{\rp\rpp}$ for all $\rpp \in \Rpp$. Analogous to the proof of Lemma \ref{lemma:FF}, any integer optimal solution of the following linear program is an optimal solution of \eqref{eq:LP1-r}: 
\begin{subequations}\label{eq:LPk-r}
        \begin{align}
         \max_{\xp, \xpp} \quad &S(\xp, \xpp) = \sum_{\bbar \in \Bbar}\sum_{\rp \in \Rp} V^1_{\rp}(\b) \xp_{\rp}(\b)+  \sum_{\bbar \in \Bbar}\sum_{\rpp \in \Rpp} V^2_{\rpp}(\b) \xpp_{\rpp}(\b) \notag\\
   s.t. \quad  &  \sum_{\ri \in \Ri}\sum_{\b\ni \m} \xxi_{\ri}(\b) \leq 1, \quad \forall \m \in \M, \quad \forall i=1, 2, \label{subeq:LPk1-r}\\
 \quad & \sum_{\b \in \B} \xxi_{\ri}(\b)  \leq \kopti_{\ri}, \quad \forall \ri \in \Ri, \quad \forall i=1, 2,\label{subeq:LPk2r}\\
  & \sum_{\rp \in \Rp}\xp_{\rp}(\b)= \sum_{\rpp \in \Rpp}\xp_{\rpp}(\b), \quad \forall \b \in \B, \label{subeq:LPk-extra}\\
    &\xxi_{\ri}(\b) \geq 0,  \quad \forall \bbar \in \Bbar, \quad \forall \ri \in \Ri, \quad \forall i=1, 2.\label{subeq:LPk3-r}
    \end{align}
    \end{subequations}
    
    We note that a trip $\(\b, \rp\rpp\)$ is organized if and only if both $\xp_{\rp}(\b)=1$ and $\xpp_{\rpp}(\b)=1$. Thus, any $\(\xp, \xpp\)$ is feasible in \eqref{eq:LP1-r} (resp. \eqref{eq:LPk-r}) if and only if there exists a feasible $\x$ in \eqref{eq:LP1bar} (resp. \eqref{eq:LP2k}) such that $\xp_{\rp}(\b)=\sum_{\rpp \in \Rpp}\x_{\rp\rpp}(\b)$ and $\xpp_{\rpp}(\b)=\sum_{\rp \in \Rp}\x_{\rp\rpp}(\b)$. Moreover, the value of the objective function $S\(\xp, \xpp\)$ equals to $S(\x)$ with the corresponding $\x$: 
{\small \begin{align*}
    &S(\xp, \xpp) \\
    \stackrel{\eqref{eq:Virb}}
    {=} &\sum_{\b \in \B} \sum_{\rp \in \Rp}\xp_{\rp}(\b)\( \sum_{\m \in \b}\alpha^{1}_m - |b|\changetrip^1(|b|)\) -\sum_{\b \in \B} \sum_{\rp \in \Rp} \(\sum_{\m \in \b} \vm d_{\rp} + |\b|\changevm(|\b|)d_{\rp}\)\xp_{\rp}(\b) \\
    &+\sum_{\b \in \B} \sum_{\rpp \in \Rpp}\xpp_{\rpp}(\b)\( \sum_{\m \in \b}\alpha^{2}_m - |b|\changetrip^2(|b|)\) -\sum_{\b \in \B} \sum_{\rpp \in \Rpp} \(\sum_{\m \in \b} \vm d_{\rpp} + |\b|\gamma(|\b|)d_{\rpp}\)\xpp_{\rpp}(\b) \\
    \stackrel{\eqref{subeq:LP1-extra}}{=}& \sum_{\b \in \B} \sum_{\r \in \R} \xrb \(\sum_{\m \in \b} \tripm - |b|\changetrip(|b|)\) - \sum_{\b \in \B} \sum_{\rp \in \Rp} \(\sum_{\m \in \b} \vm d_{\rp} + |\b|\changevm(|\b|)d_{\rp}\)\sum_{\rpp \in \Rpp}\x_{\rp\rpp}(\b) \\
    &-  \sum_{\b \in \B} \sum_{\rpp \in \Rpp} \(\sum_{\m \in \b} \vm d_{\rpp} + |\b|\changevm(|\b|)d_{\rpp}\)\sum_{\rp \in \Rp}\x_{\rp\rpp}(\b) \\
    =& \sum_{\b \in \B} \sum_{\r \in \R}\xrb \(\sum_{\m \in \b} \tripm - \vm d_{\r} - |b| \changetrip(|b|)- |\b|\changevm(|\b|)d_{\r}\) = S(\x).
\end{align*}}
Therefore, given any optimal solution $\xopt$ of \eqref{eq:LP2k}, $\(\xoptp, \xoptpp\)$, where $\xoptp_{\rp}(\b)=\sum_{\rpp \in \Rpp}\xopt_{\rp\rpp}(\b)$ and $\xoptpp_{\rpp}(\b)=\sum_{\rp \in \Rp}\xopt_{\rp\rpp}(\b)$, is an integer optimal solution of \eqref{eq:LPk-r}. Additionally, $\(\xoptp, \xoptpp\)$ is also an optimal solution of 
\eqref{eq:LP1-r}. Hence, the optimal values of \eqref{eq:LP1bar}, \eqref{eq:LP1-r}, \eqref{eq:LP2k} and \eqref{eq:LPk-r} are the same.

    We introduce the dual variables $\ui=\(\ui_m\)_{\m \in \M}$ for constraints \eqref{subeq:LP11-r}, $\tolli=\(\tolli_e\)_{\e \in \Ei}$ for \eqref{subeq:LP12-r} of each $i=1, 2$, and $\chi=\(\chi(\b)\)_{\b \in \B}$ for \eqref{subeq:LP1-extra}. Then, the dual program of \eqref{eq:LP1-r} can be written as follows: 
    \begin{subequations}\label{eq:D-r}
\begin{align}
    \min_{\up, \upp, \tollp, \tollpp, \chi} \quad &U= \sum_{\m \in \M} \up_m+ \sum_{\m \in \M} \upp_m + \sum_{\e \in \Ep} \qe \tollp_e+ \sum_{\e \in \Epp} \qe \tollpp_e\notag \\
s.t. \quad & \sum_{\m \in \bbar} \up_m +\sum_{\e \in \rp} \tollp_e + \chi(\b)\geq V^1_{\rp}(\b), \quad \forall \bbar \in \Bbar, \quad \forall \rp \in \Rp,\label{subeq:D111}\\
& \sum_{\m \in \bbar} \upp_m +\sum_{\e \in \rpp} \tollpp_e - \chi(\b)\geq V^2_{\rpp}(\b), \quad \forall \bbar \in \Bbar, \quad \forall \rpp \in \Rpp,\label{subeq:D112}\\
        & \ui_m, ~ \tolli_e \geq 0, \quad \forall \m \in \M, \quad \forall \e \in \E, \quad i=1, 2.\label{subeq:D12-r}
\end{align}
\end{subequations}

    Similarly, we obtain the dual program of \eqref{eq:LPk-r} with the same dual variables except for the route price vector $\lambi=\(\lambi_{\ri}\)_{\ri \in \Ropti}$ for \eqref{subeq:LPk2r}: 
    \begin{subequations}\label{eq:Dk-r}
    \begin{align}
    \min_{\up, \upp, \lambp, \lambpp, \chi} \quad &U= \sum_{\m \in \M} \up_m+ \sum_{\m \in \M} \upp_m + \sum_{\rp \in \Roptp} k^{1*}_{\rp} \lambp_{\rp}+ \sum_{\rpp \in \Roptpp} k^{2*}_{\rpp} \lambpp_{\rpp}\notag \\
s.t. \quad & \sum_{\m \in \bbar} \up_m +\lambp_{\rp} + \chi(\b)\geq V^1_{\rp}(\b), \quad \forall \bbar \in \Bbar, \quad \forall \rp \in \Roptp,\label{subeq:Dk11}\\
& \sum_{\m \in \bbar} \upp_m +\lambpp_{\rpp} - \chi(\b)\geq V^2_{\rpp}(\b), \quad \forall \bbar \in \Bbar, \quad \forall \rpp \in \Roptpp,\label{subeq:Dk12}\\
        & \ui_m, ~ \lambi_{\ri} \geq 0, \quad \forall \m \in \M, \quad \forall \ri \in \Ropti, \quad i=1, 2. \label{subeq:Dk-extra}
\end{align}
\end{subequations}
From strong duality, we know that the optimal value of \eqref{eq:Dk-r} (resp. \eqref{eq:D2k}) is the same as that of \eqref{eq:LPk-r} (resp. \eqref{eq:LP2k}). Since the optimal values of \eqref{eq:LP2k} and \eqref{eq:LPk-r} are identical, we know that the optimal values of \eqref{eq:Dk-r} must be equal to that of \eqref{eq:D2k}. Additionally, we can check that for any feasible solution $\(\up, \upp, \lambp, \lambpp, \chi\)$ of \eqref{eq:Dk-r} must correspond to a feasible solution $\(\u, \lambda\)$ of \eqref{eq:D2k} such that $\um=\up_m+\upp_m$ and $\lambda_{\rp\rpp}=\lambp_{\rp}+\lambpp_{\rpp}$. Then, for each $\(\uopt, \lambopt\)$, we consider the optimal solution $\(\uoptp, \uoptpp, \lamboptp, \lamboptpp, \chi^{*}\)$ of \eqref{eq:Dk-r}, and define $\tilde{V}^1_{\rp}(\b)= V^1_{\rp}(\b)-\chi^*(\b)$, $\tilde{V}^2_{\rpp}(\b)= V^2_{\rpp}(\b)+\chi^*(\b)$ for each $\rp \in \Rp$, $\rpp \in \Rpp$ and $\b \in \B$. Then, for each $i=1, 2$, $\(\uopti, \lambopti\)$ is an optimal solution of the following linear program: %equals to that of \eqref{eq:LPk-r}, the optimal value of \eqref{} We can check that any vector $\(\uopt, \lambopt\)$ is an optimal solution of \eqref{eq:D2k} if and only if there exists a optimal solution $\(\uoptp, \uoptpp, \lamboptp, \lamboptpp, \chi^{*}\)$ of \eqref{eq:Dk-r} such that $\uoptp_m+\uoptpp_m=\umopt$ for all $\m \in \M$, and $\lamboptp_{\rp}+\lamboptpp_{\rpp}=\lambopt_{\rp\rpp}$ for all $\rp\rpp \in \Roptp \times \Roptpp$. 
\begin{subequations}\label{eq:Dki}
    \begin{align}
    \min_{\ui, \lambi} \quad &U^i= \sum_{\m \in \M} \ui_m+ \sum_{\ri \in \Ropti} \kopti_{\ri} \lambi_{\ri}\notag \\
s.t. \quad & \sum_{\m \in \bbar} \ui_m +\lambi_{\ri} \geq \tilde{V}^i_{\ri}(\b), \quad \forall \bbar \in \Bbar, \quad \forall \ri \in \Ropti, \label{subeq:Dki1}\\
         & \ui_m, ~ \lambi_{\ri} \geq 0, \quad \forall \m \in \M, \quad \forall \ri \in \Ropti. \label{subeq:Dki2}
\end{align}
\end{subequations}
% and  $\(\uoptpp, \lamboptpp\)$ is an optimal solution of the following linear program: 
% \begin{subequations}\label{eq:Dkpp}
%     \begin{align}
%     \min_{\upp, \lambpp} \quad &U=  \sum_{\m \in \M} \upp_m + \sum_{\rpp \in \Roptpp} k^{2*}_{\rpp} \lambpp_{\rpp}\notag \\
% & \sum_{\m \in \bbar} \upp_m +\lambpp_{\rpp}  \geq \hat{V}_{\rpp}(\b), \quad \forall \bbar \in \Bbar, \quad \forall \rpp \in \Rpp, \label{subeq:Dkpp1}\\
%         &  \upp_m, ~\lambpp_e \geq 0, \quad \forall \m \in \M, \quad \forall \rpp \in \Roptpp. \label{subeq:Dkpp2}
% \end{align}
% \end{subequations}
From the assumption of the mathematical induction, there exists an edge price vector $\tollopti$ such that $\(\uopti, \tollopti\)$ is an optimal dual solution of the trip organization problem on the sub-network given $\tilde{V}^i$ value function for each $i=1, 2$: 
\begin{equation}\label{eq:Dp}
    \begin{split}
    \min_{\ui, \lambi} \quad &U= \sum_{\m \in \M} \ui_m+ \sum_{\e \in \Ei} \qe \tolli_{\e} \\
s.t. \quad & \sum_{\m \in \bbar} \ui_m +\sum_{\e \in \ri} \tolli_e \geq \tilde{V}^i_{\ri}(\b), \quad \forall \bbar \in \Bbar, \quad \forall \ri \in \Ri,\\
         & \ui_m, ~ \tolli_{\e} \geq 0, \quad \forall \m \in \M, \quad \forall \e \in \Ei.
\end{split}
\end{equation}
Since the objective function \eqref{eq:D-r} is the sum of the objective functions in \eqref{eq:Dp} for $i=1, 2$, and the constraints are the combination of the constraints in the two linear programs, we know that $\(\uoptp, \uoptpp, \tolloptp, \tolloptpp, \chi^*\)$ must be an optimal solution of \eqref{eq:D-r}. We consider the edge price vector $\tollopt = \(\tolloptp, \tolloptpp\)$. Since $\(\uoptp, \uoptpp, \tolloptp, \tolloptpp, \chi^*\)$ satisfies constraints \eqref{subeq:D111} and \eqref{subeq:D112} and $\umopt=\uoptp_m+\uoptpp_m$ for all $\m \in \M$, $\(\uopt, \tollopt\)$ is a feasible solution of \eqref{eq:D1bar} on the original network $\G$. Furthermore, since $\(\uopt, \tollopt\)$ achieves the same objective value as the optimal solution $\(\uoptp, \uoptpp, \tolloptp, \tolloptpp, \chi^*\)$ in \eqref{eq:Dk-r}, $\(\uopt, \tollopt\)$ must be an optimal solution of \eqref{eq:D1bar} on the network $\G$.

Finally, we conclude from cases 1 and 2 that in any series-parallel network, for any optimal solution $\(\uopt, \lambopt\)$ of \eqref{eq:D2k}, there must exist an edge price vector $\tollopt$ such that $\(\uopt, \tollopt\)$ is an optimal solution of \eqref{eq:D1bar}.  \QEDA

%We can check that for any optimal solution $\(\uopt, \tollopt\)$ of \eqref{eq:D1bar}, the vector $\(\uopt, \lambda^*\)$ -- where $\lambda^{*}_r = \sum_{\e \in \r} \toll_e^{*}$ for each $\r \in \R$ -- must also be optimal in \eqref{eq:D2k}. That is, any $\uopt \in \Uopt$ is also an optimal utility vector in \eqref{eq:D2k}. On the other hand, we show that any optimal utility vector of \eqref{eq:D2k} is also an equilibrium utility vector in $\Uopt$ in that we can find an edge price vector $\tollopt$ such that $\(\uopt, \tollopt\)$ is an optimal solution of \eqref{eq:D1bar}. We prove this argument by mathematical induction using the series-parallel network condition: If such an equilibrium edge price vector exists on two series parallel networks, then the combined edge price vector is also an equilibrium edge price vector when the two networks are connected in series or in parallel. 

Lemma \ref{lemma:utility} enables us to characterize the agents' equilibrium utility set $\Uopt$ using the dual program \eqref{eq:D2k} associated with the optimal trip organization problem under the capacity constraint with $\kopt$. The following lemma further shows that $\Uopt$ is a lattice with $\udag$ being the maximum element. 
\begin{lemma}\label{lemma:utility_lattice}
If the network is series-parallel, and agents have homogeneous carpool disutilities, then the set $\Uopt$ is a complete lattice with $\udag \in \Uopt$, and $\udag_m \geq \uopt_m$ for any $\uopt \in \Uopt$ and any $m \in M$. 
\end{lemma}

%The proof of Lemma \ref{lemma:utility_lattice}, utilizes the equivalence between the trip organization problem with the augmented trip value function $\barV$ and the economy constructed in Sec. \ref{sec:tractable_pooling}. In particular, we can show that $\Uopt$ characterized by \eqref{eq:D2k} is identical to set of good prices in Walrasian equilibrium of the economy. From Lemma \ref{lemma:condition_gross}, we know that $\barV$ satisfy monotonicity and gross substitutes conditions when agents have homogeneous carpool disutilities. Building on the theory of Walrasian equilibrium (\cite{gul1999walrasian}, also included in Lemma \ref{lemma:lattice} in Appendix \ref{sec:review}), we can show that $\Uopt$ is a lattice, and $\udag$ is the maximum element in $\Uopt$. 

% \begin{lemma}\label{lemma:utility_walrasian}
% Any $\uopt$ is an optimal utility vector in \eqref{eq:D2k} if and only if there exists a good assignment vector $\yall^*$ in economy $\econ$ such that $\(\yall^*, \uopt\)$ is a Walrasian equilibrium.
% \end{lemma}

\medskip 
\noindent\emph{Proof of Lemma \ref{lemma:utility_lattice}.}
We first prove that a utility vector $\uopt \in \Uopt$ if and only if $\uopt$ is an optimal utility vector of the following linear program: 
\begin{subequations}
        \makeatletter
        \def\@currentlabel{$\bar{\mathrm{D}}k^*$}
        \makeatother
        \label{eq:Dbark}
        \renewcommand{\theequation}{$\bar{\mathrm{D}}k^*$.\alph{equation}}
    \begin{align}
         \min_{\u, \lambda}  \quad & \sum_{\m \in \M} \um + \sum_{\r \in \R}  \kopt_r \lambda_r, \notag\\
   s.t. \quad  &  \sum_{\m \in \ball} \um + \lambda_r \geq \bar{V}_r(\ball), \quad \forall (\ball, \r) \in \Ball \times \R, \label{subeq:Dbark1}\\
    &\um \geq 0, ~\lambda_r \geq 0, \quad \forall \m \in \M, \quad \forall \r \in \R.\label{subeq:Dbark2}
    \end{align}
\end{subequations} We note that $\(\uopt, \lambdaopt\)$ is a feasible solution of \eqref{eq:Dbark} since for any $(\ball, \r) \in \Ball \times \R$, 
\begin{align*}
    \sum_{\m \in \ball} \um + \lambda_r \geq \sum_{\m \in \rep(\ball)} \um + \lambda_r \geq V_r(\rep(\ball)) =  \bar{V}_r(\ball),
\end{align*}
where $\rep(\ball)$ is a representative agent coalition of $\ball$ given $\r$. 

Note that any $\(\uopt, \lambopt\)$ is an optimal solution of \eqref{eq:D2k} if and only if there exists an optimal solution $\xopt$ of \eqref{eq:LP2k} such that $(\xopt, \uopt, \lambopt)$ satisfies the primal feasibility, dual feasibility, and complementary slackness conditions corresponding to \eqref{eq:LP2k} and \eqref{eq:D2k}. Given such $\xopt$, we can construct $\xbarr^*$ such that for any $r \in \R$, $\xbarr_r^{*}(\ball)=x_r^{*}(\ball)$ for all $\ball \in \B$, and $\xbarr_r^{*}(\ball)=0$ for all $\ball \in \Ball \setminus \B$. Such $\xbarr^*$ is an optimal solution of \eqref{eq:LPy} since it achieves the same total social value as $\xopt$. It remains to show that $(\xbarr^*, \uopt, \lambopt)$ satisfies the complementary slackness conditions associated with \eqref{subeq:LPy1}, \eqref{subeq:LPy2}, and \eqref{subeq:Dbark1}: 
\begin{itemize}
    \item[(1)] We note that $\sum_{r \in R}\sum_{\ball \ni m}\xbarr_r^{*}(\ball) = \sum_{r \in R} \sum_{\ball \ni m}\x_r^{*}(\b)$ for any $\m \in \M$. From the complementary slackness condition associated with \eqref{subeq:LPy1}, we know that $$\uopt_m \cdot \(1- \sum_{r \in R}\sum_{\ball \ni m}\xbarr_r^{*}(\ball)\) = \uopt_m \cdot \(1- \sum_{r \in R}\sum_{\ball \ni m}\x_r^{*}(\b)\)=0, \quad \forall m \in M.$$ Thus, the complementary slackness condition associated with \eqref{subeq:LP2k1} is satisfied. 
    \item[(2)] We note that $\sum_{\ball \in \Ball}\xbarr_r^{*}(\ball) = \sum_{\b \in \B}\x_r^{*}(\ball)$. From the complementary slackness condition associated with \eqref{subeq:LPy2}, $\lambopt_r \cdot \(\kopt_r - \sum_{\ball \in \Ball}\xbarr_r^{*}(\ball)\) = \lambopt_r \cdot \(\kopt_r - \sum_{\b \in \B}x_r^{*}(\ball)\)=0$. Therefore, the complementary slackness condition associated with \eqref{subeq:LP2k2} is satisfied.   
    \item[(3)] Since $\xbarr^{*}$ is 0 for any $\ball \in \Ball \setminus \B$, we only need to check that $(\xbar^*, \uopt, \lambopt)$ satisfies the complementary slackness conditions associated with \eqref{subeq:Dbark1} for $(\ball, \r) \in \Ball \times \R$. Since $\xbarr^{*}_r(\ball) = x^{*}_r(\ball)$ for all $\ball \in \B$, such complementary slackness conditions directly follow from that with respect to $\(\xopt, \uopt, \lambopt\)$.
\end{itemize}

We next show that any $\(\uopt, \lambopt\)$ that is an optimal solution of \eqref{eq:Dbark} is also an optimal solution of \eqref{eq:D2k}. We note that $\xbarr^*$ constructed from the optimal solution $\xopt$ in \eqref{eq:LP1bar} is also an optimal solution of \eqref{eq:LPy}, and thus $\(\xbarr^*, \uopt, \lambopt\)$ satisfies the complementary slackness conditions with respect to \eqref{subeq:LPy1}, \eqref{subeq:LPy2}, and \eqref{subeq:Dbark1}. From the complementary slackness condition associated with \eqref{subeq:LPy1}, we know that for any agent that is not assigned to a trip in $\xbarr^*$, the utility is zero. Since the agent coalitions assigned in $\xopt$ are the same as those in $\xbarr^*$, we know that any agent that is not assigned to trips in $\xopt$ also has zero utility. We have $\sum_{\m \in \ball} \uopt_m = \sum_{\m \in \rep(\ball)} \uopt_m$ for all $\ball \in \Ball$, where $\rep(\ball)$ is the representative agent coalition that is organized given $\ball$. Thus, $(\uopt, \lambopt)$ is a feasible solution of \eqref{eq:D2k}.

Following the analogous argument as in (1) -- (3), we can prove that $(\xopt, \uopt, \lambopt)$ also satisfies the complementary slackness conditions associated with \eqref{subeq:LPy1}, \eqref{subeq:LPy2}, and \eqref{subeq:Dbark1}, where $\xopt$ is the optimal solution in \eqref{eq:LPy} that is used to construct $\xbarr^*$.

Finally, following the proof of Lemma \ref{lemma:integer}, we know that $\Uopt$ is the set of equilibrium prices of goods in the equivalent economy. From Lemma \ref{lemma:lattice}, we know that the set of Walrasian equilibrium price is a lattice, and the maximum element is $\udag$. Consequently, we can conclude that the set $\Uopt$ is a lattice, and $\udag$ is the maximum element. Since the optimal value of the objective function for all $(\uopt, \tollopt)$ equals to $S(\xopt)$, we can conclude that the total edge price given by $\tolldag$ is no higher than that of any other equilibrium. \QEDA

Combining Lemmas \ref{lemma:utility} and \ref{lemma:utility_lattice}, we know that $(\xopt, \pdag, \tolldag)$ is a market equilibrium. We conclude Theorem \ref{theorem:strategyproof} by noticing that $(\xopt, \pdag)$, where $\pdag$ is the VCG payment as in \eqref{eq:pmdag}, implements the same outcome as a VCG mechanism, and thus $(\xopt, \pdag, \tolldag)$ must be strategyproof.

\section{Supplementary material for Section \ref{sec:multipop}}\label{apx:multi}
\begin{lemma}[\cite{ruzika2011earliest}]\label{lemma:earliest_arrival}
On series-parallel networks, the temporally repeated flow $w^*$ maximizes the total flow that arrives on or before $t$ for every $t=1, 2, \dots, T$. That is, for any feasible $x$, we have: 
\begin{align*}
\sum_{r \in \R} \sum_{\start=1}^{t-\tr} \sum_{b \in B} x_r^{\start}(b) \leq \sum_{r \in \R} \kopt_r \max\{0, t-\tr\}, \quad \forall t=1, \dots, T.
\end{align*}
\end{lemma}

\noindent\emph{Proof of Lemma \ref{lemma:FF_dynamic}.}
Consider any (fractional) optimal solution of \eqref{eq:LP1bar}, denoted as $\xhat$. For any time step $t$, we denote $\fhat^t(\b) = \sum_{\r \in \R} \xhat^{t-\tr}_r(\b)$ as the flow of coalition $\b$ that arrives at the destination at time $t$. We denote $\widehat{F}^t= \sum_{j=1}^{t}\sum_{\b \in \B} \fhat^j(\b)$ as the total flows that arrive at the destination on or before time step $t$. Since $\xhat$ is feasible and the network is series-parallel, we know from Lemma \ref{lemma:earliest_arrival} that 
\begin{align}\label{eq:max_capacity}
    \widehat{F}^t \leq \sum_{r \in R} k_r \cdot \max\{0, t-\tr\}, \quad \forall t\in [T]. 
\end{align}

We denote the set of all coalitions with positive flow in $\xhat$ as $\widehat{\B} \deleq \{\bhat \in \B|\sum_{t=1}^{T}\fhat^t(\bhat)>0\}$. For each $\bhat \in \Bhat$, we re-write the trip value function in \eqref{eq:value_of_trip} as follows: 
\begin{align*}
V_r^{\start}(\bhat)= w(\bhat)- g(\bhat) \tr -\sum_{\m \in \b} \delaym((\start+\tr-\arrm)_{+}), \quad \forall \(\bhat, \r\) \in \Bhat \times \R, \quad \forall \start \in [T-\tr], \end{align*}
where $w(\bhat)= \sum_{\m \in \bhat} \(\tripm - \Delta\alpha_m(|\hat{b}|)\)$, and $g(\bhat) = \sum_{\m \in \bhat} \(\vm  + \Delta \beta_m(|\hat{\b}|) \)$ is the sensitivity with respect to travel time cost. We denote the number of agent coalitions in $\Bhat$ as $n$, and re-number these agent coalitions in decreasing order of $g(\bhat)$, i.e.
\begin{align*}
g(\bhat_1) \geq g(\bhat_2) \geq \cdots \geq g(\bhat_n).
\end{align*}

%, and 

We now construct another trip vector $\xopt$ by the following procedure: \\
Initial zero assignment vector $\xrzbopt \leftarrow 0$ for all $\start=1, 2, \dots, T$, $\r \in \R$ and all $\b \in \B$.\\
Initial residual capacity of arriving on or before each time $t$: $\Delta^t= \sum_{r \in R} \kopt_r \cdot \max\{0, t-\tr\}-\widehat{F}^t$. \\
Initial residual capacity of taking route $r$ to arrive at time $t$: $\Lambda_r^t = \kopt_r$.\\
\emph{For $i=1, \dots, n$:}\\
\emph{For $t=1, \dots, T$:} Re-assign the flow $\hat{f}^t(\hat{b}_{i}) = \sum_{r \in R} \hat{x}_r^{t-\tr}(\hat{b}_i)$ of coalition $\hat{b}_i$ that arrive at time $t$.
\begin{itemize}
    \item[(a)] Determine the assignable arrival time step set: $\hat{T}= \{t' < t| \Delta^{t'} >0\} \cup \{t\}$. 
    \item[(b)] Determine assignable route set $\hat{R} = \{r \in R|\sum_{t \in \hat{T}}\Lambda_r^t >0\}$
    \item[(c)] Assign agent coalition $\bhat_i$ to a trip that takes route $\rhat$ and starts at $\hat{z} = \hat{t} - d_{\rhat}$, where $(\rhat, \hat{t})$ satisfies: 
    \begin{align*}
        \rhat = \arg\min_{\r \in \hat{R}}\{\tr\}, \quad  \hat{t}= \max_{t} \{t \in \hat{T}|\Lambda_{\rhat}^t >0\}.
    \end{align*}
     If $\hat{t}= t$ and $\Lambda_{\rhat}^{t} \geq \hat{f}^t(\bhat_i)$, then $x^{\hat{\start}*}_{\rhat}(\bhat_i)= \hat{f}^t(\bhat_i)$. Re-calculate $\Lambda_{\hat{r}}^{t} \leftarrow \Lambda_{\hat{r}}^{t}- x^{\hat{z}*}_{\rhat}(\bhat_i)$.\\
     If $\hat{t}= t$ and $\Lambda_{\rhat}^{t} < \hat{f}^t(\bhat_i)$, then $x^{\hat{\start}*}_{\rhat}(\bhat_i) =\Lambda_{\rhat}^{t}$. Re-calculate $\Lambda_{\hat{r}}^{t} \leftarrow 0$. Repeat (a) - (c). \\
     If $\hat{t}< t$ and $\min\{\Lambda_{\rhat}^{\hat{t}}, \Delta^{\hat{t}}\} \geq \hat{f}^t(\bhat_i)$, then $x^{\hat{\start}*}_{\rhat}(\bhat_i)= \hat{f}^t(\bhat_i)$. Re-calculate $\Lambda_{\hat{r}}^{\hat{t}} \leftarrow \Lambda_{\hat{r}}^{\hat{t}}- x^{\hat{z}*}_{\rhat}(\bhat_i)$, $\Delta^{j} \leftarrow \Delta^{j} - x^{\hat{z}*}_{\rhat}(\bhat_i)$ for $j=\hat{t}, \dots, t-1$.\\
     If $\hat{t}< t$ and $\min\{\Lambda_{\rhat}^{\hat{t}}, \Delta^{\hat{t}}\} < \hat{f}^t(\bhat_i)$, then $x^{\hat{\start}*}_{\rhat}(\bhat_i) =\min\{\Lambda_{\rhat}^{\hat{t}}, \Delta^{\hat{t}}\}$. Re-calculate $\Lambda_{\hat{r}}^{\hat{t}} \leftarrow \Lambda_{\hat{r}}^{\hat{t}}- x^{\hat{z}*}_{\hat{r}}(\bhat_i)$, $\Delta^{j} \leftarrow \Delta^{j} - x^{\hat{z}*}_{\rhat}(\bhat_i)$ for $j = \hat{t}, \dots, t-1$. Repeat (a) - (c).
     % $\Delta^{\hat{z}+t_{\hat{r}}}\leftarrow \Delta^{\hat{z}+t_{\hat{r}}} - x^{\hat{\start}*}_{\rhat}(\bhat_i)$,  If $\Delta^{\hat{z}+t_{\hat{r}}}=0$, then $\hat{T} \leftarrow \hat{T}\setminus \{\hat{z}+t_{\hat{r}}\}$. 
     % If $\Lambda_{\rhat}^{\hat{\start}+t_{\rhat}} < \hat{f}^t(\bhat_i)$, then $x^{\hat{\start}*}_{\rhat}(\bhat_i) =\Lambda_{\rhat}^{\hat{\start}+t_{\rhat}}$, and $\hat{f}^t(\bhat_i) \leftarrow \hat{f}^t(\bhat_i) - x^{\hat{\start}*}_{\rhat}(\bhat_i)$. Repeat (a)-(c) until  $\hat{f}^t(\bhat_i)=0$. 
\end{itemize}

The re-assignment proceeds to re-assign the flow of $\hat{b} \in \Bhat$ in decreasing order of their sensitivity with respect to the travel time cost. For each $\hat{b}_i$, the procedure re-assigns the flow of $\hat{b}_i$ that arrives at time $t$ in $\hat{x}$ to a time step before $t$ or at $t$. In particular, the flow $\hat{f}^t(\hat{b}_i)$ can be re-assigned to arrive at a time step $\hat{t}<t$ only if there is positive residual capacity $\Delta^{\hat{t}}$. Additionally, the re-assignment prioritizes to assign $\hat{f}^t(\hat{b}_i)$ to the route with the minimum travel time cost among all routes that have residual capacity. After assigning $\hat{f}^t(\hat{b}_i)$, the residual arrival capacity $\Delta$ and the residual route capacity $\Lambda$ are re-calculated. 

We now check that the constructed trip assignment vector is a feasible solution of \eqref{eq:LP2k}. Since we only re-assigned trips with positive weight in $\hat{x}$, we know that $\sum_{\r \in \R}\sum_{z=1}^{T} \sum_{\bbar \ni \m} x_r^{\start*}(b)\leq 1$, and thus $x^*$ satisfies \eqref{subeq:LP2k1}. Additionally, we note that in all steps of assignment, the total flow of trips that use each $r$ and starts at time $\start$ is less than the capacity in the temporal repeated flow $k_r^*$. Thus, we have $\sum_{\bbar \in \Bbar} x_r^{\start*}(b)  \leq \koptrz$ for all $\r \in \R$ and for all $\start \in T$. Therefore, \eqref{subeq:LP2k2} is satisfied. Thus, the constructed $x^*$ is a feasible solution of \eqref{eq:LP2k}. %We now check that the constructed trip assignment vector is a feasible solution of \eqref{eq:LP2k}. Since we only re-assigned trips with positive weight in $\hat{x}$, we know that $\sum_{\r \in \R}\sum_{z=1}^{T} \sum_{\bbar \ni \m} x_r^{\start*}(b)\leq 1$, and thus $x^*$ satisfies \eqref{subeq:LP2k1}. We note that the total flow that takes each route $r$ on each time step $t$ given $\xopt$ does not exceed $\kopt_r$ since the re-assignment procedure only sends flow when there is residual route capacity. Therefore, $\xopt$ satisfies \eqref{subeq:LP2k2}, and is a feasible solution of \eqref{eq:LP2k}. 

%and all $r \in R$ $\sum_{\b \ni \m} \sum_{r \in \R}  \xrbopt = \sum_{\b \ni \m} \fhat(\b) \leq 1$ so that \eqref{subeq:LP2k1} is satisfied. Additionally, since in the assignment procedure, the total weight assigned to route $\r$ is less than or equal to $\kopt_r$, we must have $\sum_{\b \in \B} \xrbopt \leq \kopt_r$ for all $\r \in \R$, i.e. \eqref{subeq:LP2k2} is satisfied. Thus, $\xopt$ is a feasible solution of \eqref{eq:LP2k}. 

It remains to prove that $\xopt$ is optimal for \eqref{eq:LP2k}. We prove this by showing that $S(\xopt) \geq S(\xhat)$. The objective function $S(\xopt)$ can be written as follows:
\begin{align}\label{eq:decompose_V_new}
\sum_{\bbar \in \Bbar}\sum_{\start=1}^{T}\sum_{\r \in \R}  V_r^z(b) \xrzbopt &= \sum_{\bbar \in \Bbar}\sum_{\start=1}^{T}\sum_{\r \in \R} w(\b) \xrzbopt - \sum_{\bbar \in \Bbar}\sum_{\start=1}^{T}\sum_{\r \in \R} g(\b)\tr \xrzbopt \notag \\
& \quad - \sum_{\bbar \in \Bbar}\sum_{\start=1}^{T}\sum_{\r \in \R}\left(\sum_{\m \in \b} \delaym((\start+\tr-\arrm)_{+})\right) \xrzbopt. \end{align}
We note that the flow $\hat{f}^t(\hat{b})$ for each $t$ and $\hat{b}$ is re-assigned with the same weight to arrive on or before $t$ since the total flow that arrive on or before $t$ (i.e. $\hat{F}^t$) is no higher than $\sum_{r \in R} \kopt_r\max\{0, t-\tr\}$ (i.e. $\kopt$ is the earliest arrival flow), and a flow that arrive later than $t$ is re-assigned to arrive before $t$ only if the residual arrival capacity $\Delta^t>0$. Such re-assignment of arrival time does not occupy capacity for the flow that previously arrive on or before $t$ in $\hat{x}$. As a result, the re-assignment process must terminate with all agent coalitions being assigned with the same weight as in $\xhat$, i.e. 
$\sum_{\r \in \R}\sum_{\start=1}^{T}\xrzbopt=  \sum_{\r \in \R} \sum_{\start=1}^{T}\sum_{\r \in \R}\xhat^z_r(\b)$ for all $\b \in \B$. Therefore, 
\begin{align}\label{eq:equal_part_new}
&\sum_{\bbar \in \Bbar}\sum_{\start=1}^{T}\sum_{\r \in \R} w(\b) \xrzbopt  =  \sum_{\bbar \in \Bbar}\sum_{\start=1}^{T}\sum_{\r \in \R} w(\b) \xhat^z_r(b). 
\end{align}Additionally, in the reassignment process (a), we note that all agent coalitions $\hat{b}_1, \dots, \hat{b}_n$ with positive weight in $\hat{x}$ are assigned to arrive at a time in $\xopt$ that is no later than the arrival time in $\hat{x}$. Therefore, we know that 
\begin{align*}
     \sum_{\bbar \in \Bbar}\sum_{\start=1}^{T}\sum_{\r \in \R}\left(\sum_{\m \in \b} \delaym((\start+\tr-\arrm)_{+})\right) \xrzbopt  \leq  \sum_{\bbar \in \Bbar}\sum_{\start=1}^{T}\sum_{\r \in \R}\left(\sum_{\m \in \b} \delaym((\start+\tr-\arrm)_{+})\right) \hat{x}_r^z(b). 
\end{align*}

%We note that since $\sum_{\r \in \R} \kopt_r = C$ and $\sum_{\r \in \R} \sum_{\b \in \B} \xhat^*_r(\b) \leq C$, the algorithm must terminate with all coalitions in $\Bhat$ being assigned. Therefore, 

To prove $S(\xopt) \geq S(\xhat^*)$, we only need to show that $$\sum_{r \in \R} \sum_{\start=1}^{T} \sum_{\b \in \B} g(\b)\tr \xrzbopt \leq \sum_{r \in \R} \sum_{\start=1}^{T} \sum_{\b \in \B} g(\b)\tr \xhat_r^{z}(b)$$ %Since the value of $\sum_{r \in \R} \sum_{\start=1}^{T} \sum_{\b \in \B} g(\b)\tr x_r^z(b)$ does not depend on the departure time of each trip, we can equivalently write the inequality as follows: 
% \begin{align}\label{eq:goal}
%     \sum_{r \in \R}  \sum_{\b \in \B} g(\b)\tr \lrbopt \leq \sum_{r \in \R} \sum_{\b \in \B} g(\b)\tr \lhat_r(b)
% \end{align}
% where $\lrbopt= \sum_{z=1}^T \xrzbopt$, and $\lhat_r(b)=\sum_{z=1}^T\xhat^z_r(b)$ for each $r \in \R$ and $b \in B$. 
To do this, we next show that $\xopt$ is an optimal solution of the following problem: 
\begin{equation}\label{eq:induction_new}
\begin{split}
&\xopt \in \arg\min_{x \in \X(\fhat)}\sum_{r \in \R} \sum_{\start=1}^{T} \sum_{\b \in \B} g(\b)\tr x_r^z(b), \\
s.t. \quad 
&\X(\fhat) \deleq \left\{x\left\vert
\begin{array}{l}
\sum_{\r \in \R}\sum_{j=1}^{t} x_r^{j-d_r}(b) \geq  \sum_{j=1}^{t}\fhat^j(\b), \quad \forall \b \in \B, \quad \forall t=1, \dots, T-1, \\
\sum_{\r \in \R}\sum_{\start=1}^{T} x_r^z(b)= \sum_{j=1}^{T}\fhat^j(\b), \quad \forall \b \in \B, \\
\sum_{\b \in \B}\sum_{\r \ni \e} x_r^{t-\distre}(b)\leq \qe, \quad \forall \e \in \E, \quad \forall t=1, \dots, T,\\
x_r^z(b) \geq 0, \quad \forall \r \in \R, \quad \forall \b \in \B, \quad \forall \start=1, \dots, T. 
\end{array}
\right.
\right\}.
\end{split}
\end{equation}
That is, we will show that $\xopt$ derived from the re-assignment of flow $\hat{f}$ induced by $\hat{x}$ minimizes the value of 
\[\sum_{r \in \R} \sum_{\start=1}^{T} \sum_{\b \in \B} g(\b)\tr x_r^z(b)\]across all trip organization vectors that allocate the same weight of each $\b$ as in $\hat{x}$ and do not increase the arrival time of any flow of $\b$. In particular, $X(\hat{f})$ characterizes the set of all such $x$: the first constraint ensures that the flow of each $b$ given $x$ arrives at the destination at time on or before $t$ is no less than that in $\hat{x}$; The second constraint ensures that the total flow of each $b$ is assigned with the same weight in $x$ as that in $\hat{x}$; The third constraint ensures that $x$ satisfies the edge capacity constraint in all time steps; and the last constraint ensures the non-negativity of $x$. The intuition of \eqref{eq:induction_new} is that in the re-assignment procedure, agent coalition $\hat{b}$ with higher sensitivity with respect to travel time cost is assigned first, and prioritized to take shorter routes.

We prove \eqref{eq:induction_new} by mathematical induction. To begin with, \eqref{eq:induction_new} holds trivially on any single-link network since no-reassignment is needed with a single route. We next prove that if \eqref{eq:induction_new} holds on two series-parallel sub-networks $\Gp$ and $\Gpp$, then \eqref{eq:induction_new} holds on the network $\G$ that connects $\Gp$ and $\Gpp$ in series or in parallel. In particular, we analyze the cases of series connection and parallel connection separately: %  If network $G$ has a single edge, then \eqref{eq:induction_new} holds trivially. Since any series-parallel network is formed by connecting single edges in series or in parallel for finite times, we prove  on series-parallel network $G$ by mathematical induction, and 

\vspace{0.2cm}
\noindent\emph{(Case 1)} Series-parallel network $G$ is formed by connecting two series-parallel sub-networks $\Gp$ and $\Gpp$ in series. We denote the set of routes in sub-network $\Gp$ and $\Gpp$ as $\Rp$ and $\Rpp$, respectively. Since $\Gp$ and $\Gpp$ are connected in series, the set of routes in network $\G$ is $\R \deleq \Rp \times \Rpp$. %For any flow vector $\fhat$, we define the set of trip vectors on $G$ that satisfy the constraint in \eqref{eq:induction_new} as $\X(\fhat)$. We also define the trip  vector that is obtained from the above-mentioned procedure based on $\fhat$ as $\xopt$.  

We define the set of feasible trip vectors in the sub-network $\Gpp$ that can induce a flow satisfying the arrival time constraint as follows: 
{\small \begin{align*}
    X^2(\fhat)  
    \deleq \left\{x\left\vert
\begin{array}{l}
\sum_{\rpp \in \Rpp} \sum_{\b \in \B}x_{\rpp}^{\zpp}(b) 
 \leq \sum_{\rp \in \Rp} \max\{0, \zpp-d_{\rp}\} \left(\sum_{\rpp \in \Rpp} k_{\rp\rpp}^*\right), ~\forall \zpp\in [T], \\ \sum_{\rpp \in \Rpp}\sum_{\zpp=1}^{T} x_{\rpp}^{\zpp}(\b) = \sum_{j=1}^{T}\fhat^j(\b), \quad \forall \b \in \B, \\
\sum_{\rpp \in \Rpp}\sum_{j=1}^{t} x_{\rpp}^{j-d_{r^2}}(b) \geq  \sum_{j=1}^{t}\fhat^j(\b), \quad \forall \b \in \B, \quad \forall t \in [T-1], \\
\sum_{\b \in \B}\sum_{\rpp \ni \e} x_{\rpp}^{t-d_{\rpp,e}}(b) \leq \qe, \quad \forall \e \in \E^2, \quad \forall t\in [T],\\
x_{\rpp}^{\zpp}(b) \geq 0, \quad \forall \rpp \in \Rpp, \quad \forall \b \in \B, \quad \forall \zpp\in [T]. 
\end{array}
\right.
\right\},
\end{align*}}
where the first constraint ensures that the flow departing from the origin of $\Gpp$ at any $\zpp$ does not exceed the maximum flow that can arrive at the destination of $\Gp$ before $\zpp$. %From Lemma \ref{lemma:earliest_arrival}, we know that the maximum amount of flow that can reach the destination of $\Gp$ (i.e. the origin of $\Gpp$) at or before each time step $t$ equals to $\sum_{\rp \in \{R^{1}|t_{\rp} \leq t\}} (t-t_{\rp})(\sum_{\rpp \in \Rpp}k_{\rp\rpp}^*)$. 

Given $\hat{x}$ (and the induced flow vector $\hat{f}$), the trip vector obtained from the re-assignment procedure restricted to the sub-network $\Gpp$ is $\x^{2*}=(\x^{\zpp*}_{\rpp}(b))_{\rpp \in \Rpp, b \in B, \zpp\in [T]}$, where 
\begin{align*}
 \x^{\zpp*}_{\rpp}(b) = \sum_{\rp \in \Rp} x^{\zpp- d_{\rp}*}_{\rp\rpp}(b), \quad \forall \rpp \in \Rpp, \quad \forall \zpp\in [T], \quad \forall b \in B.
\end{align*}
We can check that $\x^{2*} \in X^2(\fhat)$. Therefore, according to the induction assumption, we have 
\begin{align}\label{eq:comparison_1}
\sum_{\rpp \in \Rpp} \sum_{\zpp=1}^{T} \sum_{\b \in \B} g(\b)d_{\rp} x_{\rpp}^{\zpp*}(b) \leq \sum_{\rpp \in \Rpp} \sum_{\zpp=1}^{T} \sum_{\b \in \B} g(\b)d_{\rpp} x_{\rpp}^{\zpp}(b), \quad \forall \xpp \in X^2(\fhat). 
\end{align}
Additionally, given any $x^{2*}$, the set of feasible trip vector restricted to the sub-network $\Gp$ is given by 
\begin{align*}
    X^1(\fhat)  \deleq \left\{x^1\left\vert
\begin{array}{l}
\sum_{\rp \in \Rp}\sum_{\zp=1}^{T} x_{\rp}^{\zp}(\b) = \sum_{j=1}^{T}\sum_{\rpp \in \Rpp}x^{j*}_{\rpp}(\b), \quad \forall \b \in \B, \\
\sum_{\rp \in \Rp}\sum_{j=1}^{t} x_{\rp}^{j-d_{\rp}}(b) \geq  \sum_{j=1}^{t}\sum_{\rpp \in \Rpp}x^{j*}_{\rpp}(\b), \quad \forall \b \in \B, \quad \forall t \in [T-1], \\
\sum_{\b \in \B}\sum_{\rp \ni \e} x_{\rp}^{t-d_{\rp,e}}(b) \leq \qe, \quad \forall \e \in \E^1, \quad \forall t\in [T],\\
x_{\rp}^{\zp}(b) \geq 0, \quad \forall \rp \in \Rp, \quad \forall \b \in \B, \quad \forall \zp\in [T]. 
\end{array}
\right.
\right\}
\end{align*}
We consider $x^{1*}= (\x^{\zp*}_{\rp}(b))_{\rp \in \Rp, b \in B, \zp \in [T]}$, where 
\begin{align*}
    \x^{\zp*}_{\rp}(b)= \sum_{\rpp \in \Rpp} x^{\zp*}_{\rp\rpp}(b), \quad \forall \rp \in \Rp, \quad \forall \zp \in [T], \quad \forall b \in B.
\end{align*}
Analogous to our argument on $\Gpp$, we can check that $\x^{1*} \in X^1(\fhat)$. Again from the induction assumption, $\x^{1*}$ satisfies 
\begin{align}\label{eq:comparison_2}
\sum_{\rp \in \Rp} \sum_{\zp=1}^{T} \sum_{\b \in \B} g(\b)d_{\rp} x_{\rp}^{\zp*}(b) \leq \sum_{\rp \in \Rp} \sum_{\zp=1}^{T} \sum_{\b \in \B} g(\b)d_{\rp} x_{\rp}^{\zp}(b), \quad \forall \xp \in X^1(\fhat). 
\end{align}
From \eqref{eq:comparison_1} and \eqref{eq:comparison_2}, we obtain that 
\begin{align*}
    &\sum_{r \in \R} \sum_{z=1}^{T} \sum_{\b \in \B} g(\b)\tr x_r^z(b) \\
    =& \sum_{\rp \in \Rp}  \sum_{z=1}^{T}  \sum_{\b \in \B} g(\b)d_{\rp} \(\sum_{\rpp \in \Rpp} \x^z_{\rp\rpp}(\b)\)+ \sum_{\rpp \in \Rpp}\sum_{\start=1}^{T} \sum_{\b \in \B} g(\b)d_{\rpp} \(\sum_{\rp \in \Rp} \x^z_{\rp\rpp}(\b)\)\notag \\
        =&\sum_{\rp \in \Rp} \sum_{\zp=1}^{T}\sum_{\b \in \B} g(\b)d_{\rp} x^{\zp}_{\rp}(\b)+ \sum_{\rpp \in \Rpp} \sum_{\zpp=1}^{T} \sum_{\b \in \B} g(\b)d_{\rpp} x^{\zpp}_{\rpp}(\b) \\
        \geq & \sum_{\rp \in \Rp} \sum_{\zp=1}^{T}\sum_{\b \in \B} g(\b)d_{\rp} x^{\zp*}_{\rp}(\b)+ \sum_{\rpp \in \Rpp} \sum_{\zpp=1}^{T} \sum_{\b \in \B} g(\b)d_{\rpp} x^{\zpp*}_{\rpp}(\b) \\
        =& \sum_{r \in \R} \sum_{z=1}^{T} \sum_{\b \in \B} g(\b)\tr x_r^{\start*}(b).
\end{align*}
Thus, we have proved that \eqref{eq:induction_new} holds on $G$ when $\Gp$ and $\Gpp$ are connected in series.

\noindent\emph{(Case 2)} {Series-parallel network} $G$ is formed by connecting two series-parallel networks $G_1$ and $G_2$ in parallel. Same as case 1, we denote $\Rp$ (resp. $\Rpp$) as the set of routes in $\Gp$ (resp. $\Gpp$). Then, the set of all routes in $\G$ is $\R = \Rp \cup \Rpp$. 

Given any $\fhat$, we compute $\xopt$ from the re-assignment procedure in network $\G$. We denote $\foptp(b) = \sum_{\rp \in \Rp}  x_{\rp}^{t-\trp*}(b)$ (resp. $\foptpp(b) = \sum_{\rpp \in \Rpp}  x_{\rpp}^{t-\trpp*}(b)$) as the total flow of agent coalition $b$ that arrives at the destination at time $t$ using routes in the sub-network $\Gp$ (resp. $\Gpp$) given the organization vector $\xopt$. We now denote $\xoptp$ (resp. $\xoptpp$) as the trip vector $\xopt$ restricted on sub-network $\Gp$ (resp. $\Gpp$), i.e. $\xoptp = \(\x^{z*}_{\rp}(\b)\)_{\rp \in \Rp, \b \in \B, z=1, \dots, T}$ (resp. $\xoptpp = \(\x^{z*}_{\rpp}(\b)\)_{\rpp \in \Rpp, \b \in \B, z=1, \dots, T}$). We can check that $\xoptp$ (resp. $\xoptpp$) is the trip  vector obtained by the re-assignment procedure given the total flow $\foptp$ (resp. $\foptpp$) on network $\Gp$ (resp. $\Gpp$). 

Consider any arbitrary split of the total flow $\fhat$ to the two sub-networks, denoted as $\(\fhatp, \fhatpp\)$, such that $\fhat^{t,1}(\b)+\fhat^{t,2}(\b)=\fhat^t(\b)$ for all $\b \in \B$ and all $t=1, 2, \dots, T$. Given $\fhatp$ (resp. $\fhatpp$), we denote the trip  vector obtained by the re-assignment procedure on sub-network $\Gp$ (resp. $\Gpp$) as $\xopthatp$ (resp. $\xopthatpp$). We also define the set of feasible trip  vectors on sub-network $\Gp$ (resp. $\Gpp$) that induce the total flow $\fhatp$ (resp. $\fhatpp$) given by \eqref{eq:induction_new} as $\Xp(\fhatp)$ (resp. $\Xpp(\fhatpp)$). Then, the set of all trip  vectors that induce $\fhat$ on network $G$ is $\X(\fhat) = \cup_{\(\fhatp, \fhatpp\)} (\Xp(\fhatp), \Xpp(\fhatpp))$. 

Under our assumption that \eqref{eq:induction_new} holds on sub-network $\Gp$ and $\Gpp$ with any total flow, we know that given any flow split $\(\fhatp, \fhatpp\)$, 
\begin{align*}
&\sum_{r \in \R} \sum_{\start=1}^{T} \sum_{\b \in \B} g(\b)\trp \hat{x}_{\rp}^{z*}(b)+ \sum_{r \in \R} \sum_{\start=1}^{T} \sum_{\b \in \B} g(\b)\trpp \hat{x}_{\rpp}^{z*}(b) \\
\leq &\sum_{r \in \R} \sum_{\start=1}^{T} \sum_{\b \in \B} g(\b)\trp \hat{x}_{\rp}^{z}(b)+ \sum_{r \in \R} \sum_{\start=1}^{T} \sum_{\b \in \B} g(\b)\trpp \hat{x}_{\rpp}^{z}(b), 
   % \sum_{\rp \in \Rp}  \sum_{\b \in \B} g(\b)\tr \xopthatp_{\rp}(\b) + \sum_{\rpp \in \Rpp}  \sum_{\b \in \B} g(\b)\tr \xopthatpp_{\rpp}(\b) \leq \sum_{\rp \in \Rp}  \sum_{\b \in \B} g(\b)\tr \xhatp_{\rp}(\b) &+\sum_{\rpp \in \Rpp}  \sum_{\b \in \B} g(\b)\tr \xhatpp_{\rpp}(\b), \\
    \quad \quad \forall \xhatp \in \X(\fhatp), \quad \forall \xhatpp \in \X(\fhatpp).
\end{align*}
Therefore, the optimal solution of \eqref{eq:induction_new} must be a trip  vector $\(\xopthatp, \xopthatpp\)$ associated with a flow split $\(\fhatp, \fhatpp\)$. It thus remains to prove that any $\(\xopthatp, \xopthatpp\)$ associated with flow split $\(\fhatp, \fhatpp\) \neq \(f^{1*}, f^{2*}\)$ cannot be an optimal solution (i.e. can be improved by re-arranging flows). 

For any $\(\fhatp, \fhatpp\) \neq\(f^{1*}, f^{2*}\)$, we can find a coalition $\bj$ and a time $t$ such that $\fhat^{t,1}(\bj) \neq f^{t,1*}(\bj)$ (henceforth $\fhat^{t,2}(\bj) \neq f^{t,2*}(\bj)$). We denote $\bjhat$ as one such coalition with the maximum $g(\b)$, and $\that$ as the minimum of such time step, i.e. $\fhat^{t,i}(\bj)= f^{t,i*}(\bj)$ for any $i=1, 2$, any $j= 1, \dots, \jhat-1$ and any $t=1, \dots, T$. Additionally, $\fhat^{t,i}(b_{\hat{j}})= f^{t,i*}(b_{\hat{j}})$ for $i=1, 2$, and $t=1, \dots \that-1$. Since coalitions $\b_1, \dots, \b_{\jhat-1}$ are assigned before coalition $\bjhat$, we know that $\xhat^{z*}_{\rp}(\bj)=x^{z*}_{\rp}(\bj)$ and $\xhat^{z*}_{\rpp}(\bj)=x^{z*}_{\rpp}(\bj)$ for all $\rp \in \Rp$, all $\rpp \in \Rpp$ and all $j=1, \dots, \jhat-1$. %Additionally, Since $\fhatp(\bjhat) \neq \foptp(\bjhat)$, the trip  vector in $\xopthatp$ and $\xopthatpp$ must be different from that in $\xopt$.

Without loss of generality, we assume that $\fhat^{\that,1}(\bjhat)> f^{\that,1*}(\bjhat)$ and $\fhat^{\that,2}(\bjhat)< f^{\that,2*}(\bjhat)$. Then, there must exist routes $\rphat \in \Rp$ and $\rpphat \in \Rpp$, and departure time $\zp, \zpp$ such that $d_{\rphat} + \zp \leq \that$, $d_{\rpphat} + \zpp \leq \that$, $\xhat^{\zp*}_{\rphat}(\bjhat) > \x^{\zp*}_{\rphat}(\bjhat)$ and $\xhat^{\zpp*}_{\rpphat}(\bjhat) < \x^{\zpp*}_{\rpphat}(\bjhat)$. Moreover, since $\xopt$ assigns coalition $\bjhat$ to routes with the minimum travel time cost that are unsaturated after assigning coalitions $b_1, \dots, \b_{\jhat-1}$ with all arrival time $t$ and coalition $\bjhat$ with arrival time earlier than $\that$, we have $d_{\rpphat} < d_{\rphat}$. If the total flow on route $\rpphat$ with departure time $\zpp$ is less than $k_{\rpp}^{*}$ (unsaturated) given $\xopthatpp$, then we decrease $\xhat^{\zp*}_{\rphat}(\bjhat)$ and increase $\xhat^{\zpp*}_{\rpphat}(\bjhat)$ by a small positive number $\epsilon>0$. We can check that the objective function of \eqref{eq:induction_new} is reduced by $\epsilon (d_{\rphat}- d_{\rpphat})\epsilon g(\bjhat) > 0$. On the other hand, if route $\rpphat$ with departure time $\zpp$ is saturated, then another coalition $\b_{\jhat'}$ with $\jhat' > \jhat$ must be assigned to $\rpphat$ with departure time $\zpp$. Then, we decrease $\x^{\zp*}_{\rphat}(\bjhat)$ and $\x^{\zpp*}_{\rpphat}(\b_{\jhat'})$ by $\epsilon>0$, increases $\x^{\zp*}_{\rphat}(\b_{\jhat'})$ and $\x^{\zpp*}_{\rpphat}(\b_{\jhat})$ by $\epsilon$ (i.e. exchange a small fraction of coalition $\bjhat$ with coalition $\b_{\jhat'}$). Note that $g(\bjhat)> g(\b_{\jhat'})$ and $d_{\rphat}> d_{\rpphat}$. We can thus check that the objective function of \eqref{eq:induction_new} is reduced by $\epsilon (d_{\rphat} g(\bjhat)- d_{\rpphat}g(\b_{\jhat+1}))\epsilon > 0$. Therefore, we have found an adjustment of trip vector $\(\xopthatp, \xopthatpp\)$ that reduces the objective function of \eqref{eq:induction_new}. Hence, for any flow split $\(\fhatp, \fhatpp\) \neq \(\foptp, \foptpp\)$, the associated trip  vector $\(\xopthatp, \xopthatpp\)$ is not the optimal solution of \eqref{eq:induction_new}. The optimal solution of \eqref{eq:induction_new} must be constructed by the re-assignment procedure with flow split $\(\foptp, \foptpp\)$, i.e. must be $\xopt$. 

We have shown from cases 1 and 2 that if the solutions derived from the re-assignment procedure minimizes \eqref{eq:induction_new} on the two series-parallel sub-networks, then $\xopt$ derived from the re-assignment procedure must also minimize \eqref{eq:induction_new} on the connected series-parallel network. Moreover, since \eqref{eq:induction_new} is minimized trivially when the network is a single edge, and any series-parallel network is formed by connecting series-parallel sub-networks in series or parallel, we can conclude that $\xopt$ obtained from the re-assignment procedure minimizes the objective function in \eqref{eq:induction_new} for any flow vector $\fhat$ on any series-parallel network.  

From \eqref{eq:decompose_V_new}, \eqref{eq:equal_part_new} and \eqref{eq:induction_new}, we can conclude that $S(\xopt) \geq S(\xhat)$. Since $\xopt$ is a feasible solution of \eqref{eq:LP2k}, the optimal value of \eqref{eq:LP2k} must be no less than that of \eqref{eq:LP1bar}. On the other hand, since the constraints in \eqref{eq:LP1bar} are less restrictive than that in \eqref{eq:LP2k}, the optimal value of \eqref{eq:LP2k} is no higher than that of \eqref{eq:LP1bar}. Therefore, the optimal value of \eqref{eq:LP2k} equals to that of \eqref{eq:LP1bar}, and any optimal solution of \eqref{eq:LP2k} must also be an optimal solution of \eqref{eq:LP1bar}. \QEDA

\medskip

\noindent\emph{Proof of Proposition \ref{prop:extension}.} We prove (i) -- (iii) in sequence. 
\begin{itemize}
    \item[(i)] For each $e \in E$, we allocate edge capacity $q_{e,i}$ to each population $i \in I$ such that $\sum_{i \in I} q_{e, i}=q_e$. We consider the separate market for population $i$ with edge capacity vector $q_i = (q_{e, i})_{e \in E}$. From Theorem \ref{theorem:sp}, equilibrium exists in this sub-market since the network is series-parallel and agents have homogeneous capacity sharing disutilities. {In equilibrium, agents only form coalitions} within each sub-market, and edge prices are population-specific. 
    \item[(ii)] Consider any feasible flow vector $f$ in the network, we create ``slot" set $L$ that includes $|f_r|$ number of slots for each path $r \in \hat{R}:= \{R|f_r>0\}$ that is the support set of flow vector $f$. Since $|I|=1$, all agents have homogeneous capacity sharing disutilities. Following from Lemmas \ref{lemma:condition_gross}-- \ref{lemma:integer}, the following linear program associated with the selected flow vector $f$ has an integer optimal solution
    \begin{subequations}
        \begin{align*}
    \max_{\xbarr} \quad &S(\xbarr)=\sum_{\ball \in \Ball}\sum_{\r \in \R} \barV_{\r}(\ball) \xbarr_{\r}(\ball), \\
  s.t. \quad  & \sum_{\ball \ni \m} \sum_{\r \in \R}\xbarr_{\r}(\ball) \leq 1, \quad \forall \m \in \M,  \\
&\sum_{\ball \in \Ball}\xbarr_{\r}(\ball)  \leq f_r, \quad \forall \r \in \R,\\
 &\xbarr_{r}(\ball) \geq 0, \quad \forall \ball \in \Ball, \quad \forall \r \in \R.
    \end{align*}
    \end{subequations}
    and $\bar{x}^*$ can be converted to the original trip allocation vector $x^*$ using \eqref{eq:construct}. The dual optimal solution of the linear program $\lambda^*= (\lambda^*_r)_{r \in \hat{R}}$ provides equilibrium prices for routes in the support set $\hat{R}$, and $u^*= (u^*_m)_{m \in M}$ is the utility of all agents. Following Proposition \ref{prop:primal_dual}, $(x^*, p^*, \lambda^*)$ is a market equilibrium when restricted to the sub-network determined by flow $f$. By setting $\lambda^*_{r} = \max_{(b, r) \in \B \times \R} V_{r}(b)$ (a sufficiently large number), we can further verify that no subset of agents will deviate to take routes not in the support set. 
    \item[(iii)] The argument holds from combining (i) -- (ii). 
\end{itemize}\hfill $\square$

We now provide a formal description of the Branch-and-Price algorithm described in Section \ref{sec:multipop}.

\begin{algorithm}[ht]
\caption{Branch and price algorithm for solving \eqref{IP-mult}}\label{alg:branchAndPrice}
    Compute $x^*, q^*$ as an optimal solution to the LP-relaxation of \eqref{IP-mult}. \label{line:LPrelax} \\
    \eIf{$q^*$ is integral}{
    \For{$i \in I$}{
        Compute the optimal trip organization of submarket $i$ with respect to capacity $q^*$ using Algorithm \ref{alg:allocation}. 
        }
        \Return{the optimal trip organization}
    }{
    Choose an arbitrary $i \in I, r \in R_i$ such that $q^{*i}_r$ is fractional\\ 
    Recursively call Algorithm \ref{alg:branchAndPrice} to compute      $x^{(1)}$ as the integer optimal solution when constraint $q^{i}_r \leq \left\lfloor q^{*i}_r \right\rfloor$ is added, and 
        $x^{(2)}$ as the integer optimal solution when constraint $q^{i}_r \geq \left\lceil q^{*i}_r \right\rceil$ is added. \label{line:recCall2}\\    
        \Return{$\arg\max\{ S(x^{(1)}), S(x^{(2)})\}$} 
         }
    
\end{algorithm}

\FloatBarrier
 %Algorithm \ref{alg:branchAndPrice}more detailed explanation regarding the implementation of Algorithm \ref{alg:branchAndPrice}. When branching on the $q^{i, z}_r$ variables, we 

The branch and price algorithm start by computing an optimal solution of the linear relaxation of \eqref{IP-mult} (Line 1). If the optimal solution has an integral capacity allocation vector $q^*$, then we know that by using Algorithm \ref{alg:allocation}, we can compute the integral equilibrium trip organization vector and the associated edge prices and payments for each submarket (Line 3-6). If there exists at least one $(i, r)$ such that  $q^{*i}_r$ (the capacity allocated to population $i$ on route $r$) is fractional, we branch on the variable $q^{i}_r$ to create two sub-problems, where either $q^{i}_r \leq \lfloor q^{*i}_r \rfloor $ or $q^{i}_r \geq \lceil q^{*i}_r \rceil$ is added as an additional constraint. We resolve the linear relaxation associated with each subproblem, and continue to add additional constraints until we obtain an integer solution (Line 8-10). In our implementation, we also incorporate a pruning step -- if the optimal value of the linear relaxation of a subproblem is smaller than the best integer solution that has been found, then we stop branching on that subproblem.

% That is, we recursively call Algorithm \ref{alg:branchAndPrice} to compute the optimal integer trip organization vector when $q^{i, z}_r \leq \lfloor q^{*i, z}_r \rfloor $, and when $q^{i, z}_r \geq \lceil q^{*i, z}_r \rceil$. 

The key step of Algorithm \ref{alg:branchAndPrice} is to repeatedly compute the linear relaxation of the integer trip organization problem with additional constraints on $q$ that has been added in the branching process. Any such linear program has exponential number of variables -- the trip vector $x$. We compute the optimal (fractional) solution using column generation method. That is, we start by solving a restricted linear program that only includes a (small) subset of trips. Whether or not such solution is an optimal solution of the original LP can be verified by the dual feasibility. Recall that a violated dual constraint can be found in polynomial time using the greedy algorithm as in (Line 5-14 in Algorithm \ref{alg:allocation}) due to the property that the augmented trip valuation function in each sub-market satisfies the gross substitutes condition. If we detect a violated constraint, we add the corresponding trip variable and resolve the primal problem; otherwise, we terminate with an optimal solution of the original LP.

\section{Application for carpooling and toll pricing in San Francisco Bay Area}\label{subsec:numerical}

We apply our algorithm to the problem of designing the optimal carpooling and toll pricing for the highway network in the San Francisco Bay Area. In this problem, the edge price is the toll price of using that highway segment in a time period, and a shared trip is a carpool trip, where a group of travelers $b$ decides the departure time $z$ and route $r$ that connects their origin and destination.

\begin{figure}[ht]
  \begin{center}
\includegraphics[width=0.45\textwidth]{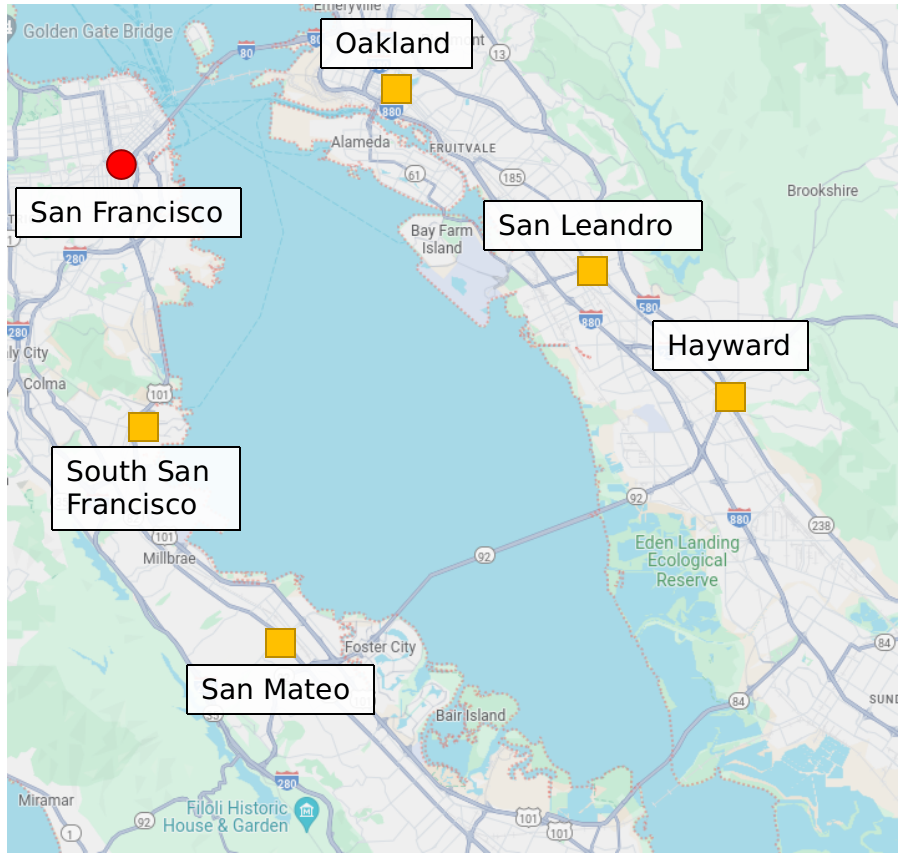}
\end{center}
    \caption{Bay area network.}
    \label{fig:SFNetwork}
\end{figure}

\noindent\textbf{Network}. We consider a network with six cities San Francisco, Oakland, San Leandro, Hayward, San Mateo, and South San Francisco, and major highways that connect them (Figure \ref{fig:SFNetwork}). San Francisco city (SF) is the common destination, and the remaining five nodes are the origins. %Notice that, as we only consider one direction of traffic on each edge, this network is series-parallel. 
We calibrate the capacity of each edge based on the traffic flow data collected from the highway sensors provided by the California Department of Transportation (https://pems.dot.ca.gov/). We consider each time interval to be 5 minutes, and the entire time period is $T = 60$ minutes between 8am and 9am on workdays. %We restricted our attention to agents who traveled to San Francisco from one of Oakland, Hayward, San Leandro, San Mateo, and South San Francisco. 

% \usepackage{booktabs}

%You can write a sentence saying that the problem parameters are calibrated using data sets (list census, caltrans, and safegraph). Write a paragraph and a table in the appendix to provide the details.

\medskip 
\noindent\textbf{Populations}. For each origin-destination pair $(o_i, \mathrm{SF})$, where $o_i$ is a city in Oakland, San Leandro, Hayward, San Mateo, and South San Francisco, agents traveling from $o_i$ to $\mathrm{SF}$ are divided into three populations, each with high (H), medium (M) and low (L) value of time, and disutilities of trip sharing. We estimate the number of agents in each population and their preference parameters based on the driving commuter population size, and their income distributions using the data collected from Safegraph (https://www.safegraph.com/), and US Census of Bureau (https://www.census.gov/).

The parameters of populations L, M, H are presented in the table below. 
\begin{table}[ht]
\begin{tabular}{@{}llll@{}}
\toprule
Parameter & Low    & Medium     &  \\ \midrule
$\alpha_m$   & Uniformly distributed on $[30, 70]$  & Uniformly distributed on $[80, 120]$  &  \\
$\beta_m$   & $\frac{10}{60}$   & $\frac{30}{60}$   &  \\
$\theta_m$   &  Uniformly distributed on $[40, 60]$  &  Uniformly distributed on $[40, 60]$ & \\
$\ell_m$      & $\delaym((\start+\tr-\arrm)_{+}) = (\start+\tr-\arrm)_{+}$            & $\delaym((\start+\tr-\arrm)_{+}) = (\start+\tr-\arrm)_{+}$ & \\
$\pi_m(|b|) + \gamma_m(|b|) d_r$           & \begin{tabular}[c]{@{}l@{}}$\pi_L(|b|) = \begin{cases} 0.25(|b|-1), & |b| \leq 5 \\ 0.5(|b|-1), & 5 < |b| \leq 10 \\ \infty, & |b| > 10 \end{cases}$\\ $\gamma_m(|b|) = 0$\end{tabular} & \begin{tabular}[c]{@{}l@{}}$\pi_M(|b|) = \begin{cases} 2(|b|-1), & |b| \leq 3 \\ 4(|b|-1), & |b| =4 \\ \infty, &  |b| > 4\end{cases}$\\ $\gamma_m(|b|) = 0$\end{tabular} &  \\ \bottomrule
\end{tabular}
\end{table}

% \usepackage{booktabs}
\begin{table}[ht]
\centering 
\begin{tabular}{@{}lll@{}}
\toprule
Parameter          & High    &  \\ \midrule
$\alpha_m$          & Uniformly distributed on $[180, 220]$     &  \\
$\beta_m$                & $\frac{90}{60}$     &  \\
$\theta_m$        &  Uniformly distributed on $[40, 60]$    &  \\
$\ell_m$                 &  $\delaym((\start+\tr-\arrm)_{+}) = (\start+\tr-\arrm)_{+}$    &  \\
$\pi_m(|b|) + \gamma_m(|b|) d_r$      & \begin{tabular}[c]{@{}l@{}}$\pi_H(|b|) = \begin{cases} 4(|b|-1), & |b| \leq 2 \\ 8(|b|-1), & |b| = 3 \\ \infty, & |b| > 3\end{cases}$\\ $\gamma_m(|b|) = 0$\end{tabular} &  \\ \bottomrule
\end{tabular}
\end{table}

\FloatBarrier
Five origin-destination pairs are considered in this instance, namely (Oakland, San Francisco), (South San Francisco, San Francisco), (Hayward, San Francisco), (San Mateo, San Francisco), and (San Leandro, San Francisco). We summarize the demand distribution, carpool sizes, and payments in the following figures:


% \usepackage{booktabs}
\begin{table}[h]
\centering 
\begin{tabular}{@{}llccc@{}}
\toprule
 ~ &Source              & Population L & Population M & Population H ~ \\ \midrule
 &Oakland             & 55           & 50           & 57           \\
&San Leandro         & 43           & 39           & 30           \\
&Hayward             & 20           & 22           & 12           \\
&South San Francisco & 18           & 19           & 23           \\
&San Mateo           & 13           & 13           & 31           \\
 \bottomrule
\end{tabular}
\caption{Distribution of Demand}
\end{table}
% \usepackage{booktabs}

\medskip

\noindent\textbf{Results}. We compute the optimal capacity allocation, the equilibrium carpool sizes, payments and toll prices using Algorithm \ref{alg:branchAndPrice}. We summarize our observations below:

\noindent\emph{1. Carpooling sizes}: In the optimal solution, the L populations from all origins form carpools of size 3 or 4, the M populations form carpools of size 1 or 2, and 95\% of population H does not carpool. \begin{table}[ht]
\centering
\begin{tabular}{@{}cccc@{}}
\toprule
Carpool Size & \% of Population L & \% of Population M & \% of Population H \\ \midrule
1            & 29.5\%                    & 30.3\%                    & 94.5\%                    \\
2            & 8.2\%                     & 51.3\%                    & 5.5\%                     \\
3            & 50.8\%                    & 18.4\%                    & 0\%                       \\
4            & 11.5\%                    & 0\%                       & 0\%                       \\ \bottomrule
\end{tabular}
\caption{Distribution of Carpool Sizes}\label{table:carpoolDistr}
\end{table}
\FloatBarrier

\noindent\emph{2. Payments}: We compute the equilibrium payment vector $p^*$ as in \eqref{eq:p}. The average payment-per-agent for the L populations is \$2, for the M populations is \$4, and for the H populations is around \$5. 
The following table presents some summary statistics of the payment distribution. 
\begin{table}[ht]
\centering
\begin{tabular}{@{}lccc@{}}
\toprule
Payment (\$)     & \multicolumn{1}{l}{Population L} & \multicolumn{1}{l}{Population M} & \multicolumn{1}{l}{Population H} \\ \midrule
Minimum          & 1                                & 1                                & 1                                \\
25th percentile  & 1.75                             & 2                                & 5                                \\
Median           & 2                                & 3                                & 9                                \\
75th percentile  & 2                                & 3                                & 9                                \\
Maximum          & 2                                & 4                                & 9                                \\
\textbf{Average} & \textbf{1.81}                    & \textbf{2.64}                    & \textbf{7.38}                    \\ \bottomrule
\end{tabular}
\caption{Distribution of Payments}\label{table:payments}
\end{table}
%  roughly corresponding to different income levels. The expected value of a trip, value for time, and carpool disutility depend on the income level of the agents, but remain constant for agents within the same subpopulation. 

\noindent\emph{3. Tolls}: The median toll per route is \$3 for the L populations, \$4 for M populations, and \$8 for the H populations. We plot the dynamic toll prices of all routes in Figure \ref{fig:tollDistr}.

    \begin{figure}[ht]
    \centering
    \includegraphics[width=\textwidth]{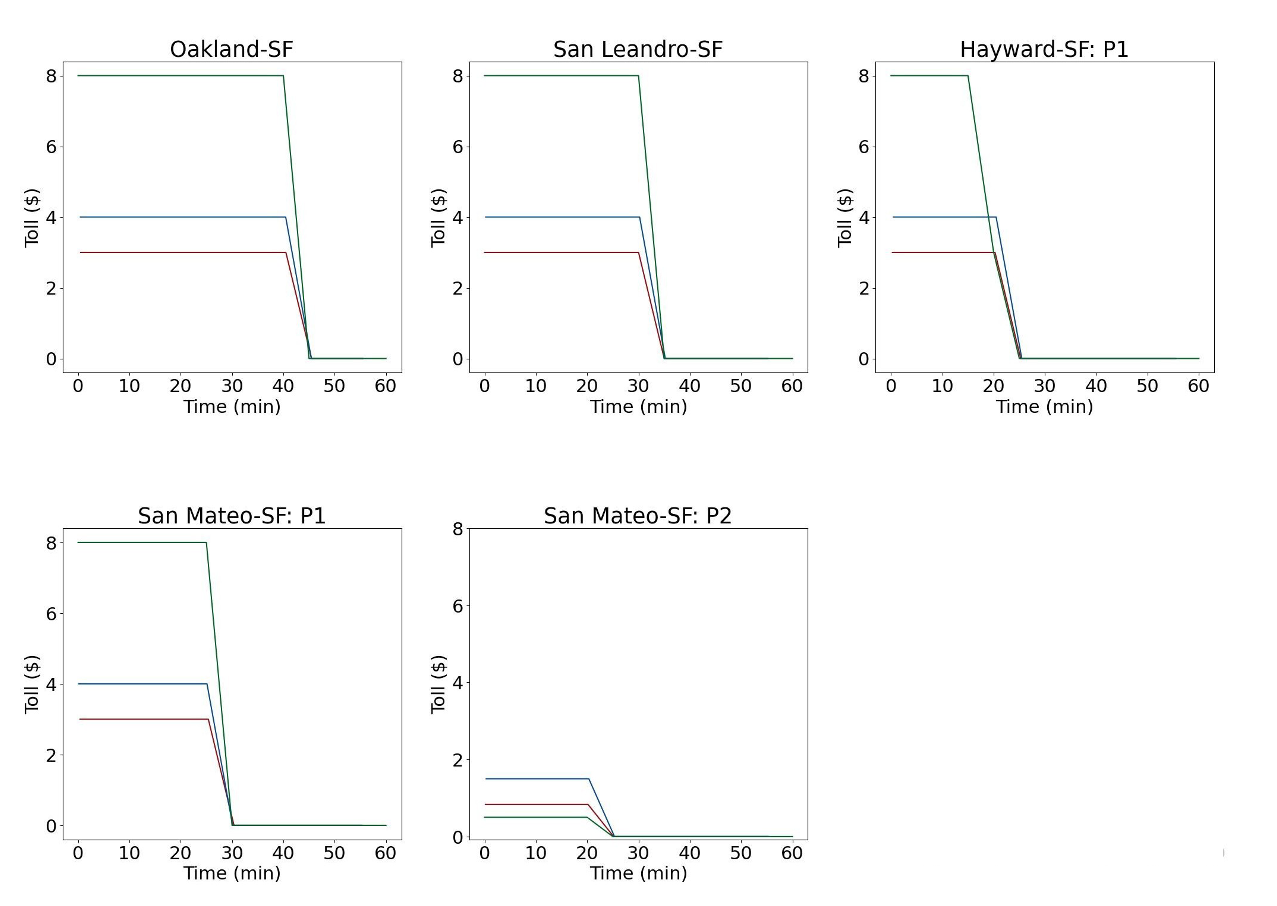}
    \caption{Dynamic toll prices of all routes for H, M, L and populations illustrated in green, blue, and red lines respectively.}
    \label{fig:tollDistr}
    \end{figure}
\FloatBarrier

\end{document}